\documentclass[10pt, letterpaper]{amsart}

\usepackage{bm}
\usepackage{latexsym, amsmath, amstext, amssymb, amsfonts, amscd, bm, array, multirow, amsbsy, mathrsfs}
\usepackage{amsthm}
\usepackage{t1enc}
\usepackage[mathscr]{eucal}
\usepackage{indentfirst}
\usepackage{pb-diagram}
\usepackage{graphicx}
\usepackage{fancyhdr}
\usepackage{fancybox}
\usepackage{enumerate}
\usepackage{color}
\usepackage[all]{xy}
\usepackage{hyperref}
\usepackage{tikz}
\usepackage{xparse}
\hypersetup{colorlinks=false,pdfborderstyle={/S/U/W 0}}
\usetikzlibrary{matrix}

\newcommand{\arxiv}[1]{{\tt
		\href{http://www.arXiv.org/abs/#1}{arXiv:#1}}}

\usepackage{url}
\usepackage[sort&compress,comma]{natbib}
\bibpunct{[}{]}{,}{n}{}{,}
\theoremstyle{plain}
\newtheorem{thm}{Theorem}[section]

\newtheorem{prop}[thm]{Proposition}

\newtheorem{lemma}[thm]{Lemma}
\newtheorem{cor}[thm]{Corollary}

\theoremstyle{definition}
\newtheorem{definition}[thm]{Definition}

\theoremstyle{remark}
\newtheorem{remark}[thm]{Remark}
\newtheorem{ep}[thm]{Example}

\newcommand{\End}{\mathrm{End}}
\newcommand{\Aut}{\mathrm{Aut}}

\newcommand{\Mat}{\mathrm{Mat}}
\newcommand{\Map}{\mathrm{Map}}

\newcommand{\eqdef}{\stackrel{{\rm def.}}{=}}

\newcommand{\Conf}{\mathrm{Conf}}
\newcommand{\Sol}{\mathrm{Sol}}

\DeclareFontFamily{U}{rsf}{}
\DeclareFontShape{U}{rsf}{m}{n}{<5> <6> rsfs5 <7> <8> <9> rsfs7 <10-> rsfs10}{}
\DeclareMathAlphabet\Scr{U}{rsf}{m}{n}

\def\Z{\mathbb{Z}}
\def\C{\mathbb{C}}
\def\R{\mathbb{R}}

\def\deg{{\rm deg}}

\def\Gl{\mathrm{Gl}}
\def\dd{\mathrm{d}}

\def\vol{\mathrm{vol}}

\def\Ad{\mathrm{Ad}}

\def\cM{\mathcal{M}}

\newcommand{\be}{\begin{equation*}}
\newcommand{\ee}{\end{equation*}}
\newcommand{\ben}{\begin{equation}}
\newcommand{\een}{\end{equation}}
\newcommand{\beqa}{\begin{eqnarray*}}
	\newcommand{\eeqa}{\end{eqnarray*}}
\newcommand{\beqan}{\begin{eqnarray}}
\newcommand{\eeqan}{\end{eqnarray}}

\newcommand{\Tr}{\mathrm{Tr}}

\def\cR{{\mathcal R}}

\def\cB{\Scr B}

\def\Cl{\mathrm{Cl}}

\def\Spin{\mathrm{Spin}}

\def\Pin{\mathrm{Pin}}
\def\Spin{\mathrm{Spin}}

\def\SO{\mathrm{SO}}
\def\O{\mathrm{O}}
\def\U{\mathrm{U}}

\def\cD{\mathcal{D}}

\def\cA{\mathcal{A}}
\def\cE{\mathcal{E}}
\def\cI{\mathcal{I}}
\def\cP{\mathcal{P}}
\def\cN{\mathcal{N}}
\def\cG{\mathcal{G}}

\def\cW{\mathcal{W}}

\def\G_2{\mathrm{G_2}}

\def\cO{\mathcal{O}}
\def\cL{\mathcal{L}}
\def\cS{\mathcal{S}}
\def\cH{\mathcal{H}}

\def\cV{\mathcal{V}}

\def\PP{\mathbb{P}}
\def\P{\mathrm{P}}

\def\Aut{\mathrm{Aut}}

\def\Fr{\mathrm{Fr}}

\def\Re{\mathrm{Re}}
\def\Im{\mathrm{Im}}

\def\G{\mathrm{G}}
\def\L{\mathrm{L}}

\def\R{\mathbb{R}}

\def\Pic{\mathrm{Pic}}

\def\tlambda{\tilde{\lambda}}

\def\mCl{\mathbb{C}\mathrm{l}}

\newcommand{\norm}[1]{\left\lvert#1\right\rvert}
\newcommand{\nnorm}[1]{\left\lVert#1\right\rVert}

\begin{document}

\title{$\cN=1$ Geometric Supergravity and chiral triples on Riemann surfaces}

\author[Vicente Cort\'es]{Vicente Cort\'es} \address{Department of Mathematics, University of Hamburg, Germany}
\email{vicente.cortes@uni-hamburg.de}

\author[C. I. Lazaroiu]{C. I. Lazaroiu}\address{Center for Geometry and Physics, Institute for Basic
	Science, Pohang, Republic of Korea 37673} \email{calin@ibs.re.kr}

\author[C. S. Shahbazi]{C. S. Shahbazi} \address{Department of Mathematics, University of Hamburg, Germany}
\email{carlos.shahbazi@uni-hamburg.de}


\begin{abstract}
We construct a global geometric model for the bosonic sector and Killing spinor equations of four-dimensional $\cN=1$ supergravity coupled to a chiral non-linear sigma model and a $\Spin^{c}_0$ structure. The model involves a Lorentzian metric $g$ on a four-manifold $M$, a complex chiral spinor and a map $\varphi\colon M\to \cM$ from $M$ to a complex manifold $\cM$ endowed with a novel geometric structure which we call \emph{chiral triple}. Using this geometric model, we show that if $M$ is spin the K\"ahler-Hodge condition on a complex manifold $\cM$ is enough to guarantee the existence of an associated $\cN=1$ chiral geometric supergravity, positively answering a conjecture proposed by D. Z. Freedman and A. V. Proeyen. We dimensionally reduce the Killing spinor equations to a Riemann surface $X$, obtaining a novel system of partial differential equations for a harmonic map with potential $\varphi\colon X\to \cM$ from $X$ into the K\"ahler moduli space $\cM$ of the theory. We characterize all Riemann surfaces admitting supersymmetric solutions with vanishing superpotential, proving that they consist on holomorphic maps of Riemann surfaces into $\cM$ satisfying certain compatibility condition with respect to the canonical bundle of $X$ and the chiral triple of the theory. Furthermore, we classify the biholomorphism type of all Riemann surfaces carrying supersymmetric solutions with complete Riemannian metric and finite-energy scalar map.
\end{abstract}

\maketitle

\setcounter{tocdepth}{1} %doesn't display subsections in TOC 
\tableofcontents

% % % % % % % % % % % % % % % % % % % % % % % % % % % % % % % % % % % % % % 
% % % % % % % % % % % % % % % % % % % % % % % % % % % % % % % % % % % % % % 

\section{Introduction}

% % % % % % % % % % % % % % % % % % % % % % % % % % % % % % % % % % % % % % 
% % % % % % % % % % % % % % % % % % % % % % % % % % % % % % % % % % % % % % 

Supergravity theories are supersymmetric theories of gravity which, at least from a differential-geometric point of view, can be viewed as a system of partial differential equations coupling a Lorentzian metric to a number of other objects, such as connections on principal bundles, spinors or forms of different degrees arising as curvatures on (possibly higher) abelian gerbes \cite{Ortin,FreedmanProeyen,Tanii:2014gaa}. The equations defining a supergravity theory are strongly constrained by the requirement of invariance under local supersymmetry transformations. The local structure of supergravity theories has been the object of intense study since the discovery of supersymmetry in the early seventies, and nowadays there exists a reasonably complete local classification of (ungauged) supergravities in Lorentzian signature, see \cite{Andrianopoli:1996cm,Andrianopoli:1996ve,Aschieri:2008ns,Tanii:2014gaa} and references therein for more details. However, geometric models which implement and extend at a global level the mathematical structures and symmetries appearing in the local formulation of these theories have not been systematically investigated. The motivations to construct global geometric formulations of supergravity theories are manifold. For instance, the understanding of the global structure of supergravity solutions, that is their \emph{U-fold} structure \cite{Dabholkar:2002sy,Hellerman:2002ax}, requires a precise knowledge of the global mathematical structure of the supergravity theory under consideration. This is, among other reasons, because the extension of a local supergravity solution demands an \emph{a priori} knowledge of the global mathematical definition of the local fields that are to be extended. A second important motivation is given by the (\emph{\'a la Kuranishi}) study of moduli spaces of solutions of supergravity theories, which requires having a precise description of the equations of motion of the theory in terms of globally well-defined differential operators. In this situation, the choice of global geometric model used for the supergravity theory under consideration may dramatically affect the corresponding moduli space of solutions.

A global geometric model for the generic bosonic sector of four-dimensional supergravity was constructed in \cite{Lazaroiu:2016iav,Lazaroiu:2017qyr}. This model involves a flat submersion $\pi\colon E\to M$ over space-time and a flat symplectic vector bundle $\cS$ over $E$. Whereas the previous model includes the bosonic sector of any ungauged four-dimensional supergravity theory, it does not yet implement supersymmetry. One of the goals of the present article is to contribute to this problem by giving a global construction of the Killing spinor equations in the case of (ungauged) $\cN=1$ supergravity coupled to chiral matter. This corresponds to the case in which $\pi$ is trivial and $\cS = E$ has trivial fibers, corresponding to the situation in which the theory is not coupled to gauge fields. Whereas a geometric model for the complete $\cN=1$ supergravity theory, including its fermionic sector, is desirable, for the differential-geometric applications we have in mind, in the spirit to the two motivational examples given above, constructing a geometric model for the bosonic sector together with the associated Killing spinor equations is sufficient. The Killing spinor equations are in this context crucial, as they determine the notion of \emph{supersymmetric solution} and give rise to well-defined moduli spaces which generalize and extend well-known moduli problems in the mathematical literature, such as the moduli problem of generalized instantons \cite{Carrion,DT1} or most notably the moduli problem of pseudo-holomorphic maps \cite{DuffSalamon}. From a physical point of view, $\cN=1$ four-dimensional supergravity is one of the most relevant supergravity theories because of its many applications to particle physics, cosmology and string phenomenology, see \cite{Dhuria:2014pla,Drees:2004jm} and references therein for more details. The present article is part of a long term program aimed at developing the mathematical theory and foundations of \emph{geometric supergravity}, which we define as the geometric global theory of bosonic supergravity together with its associated supersymmetry Killing spinor equations.

The geometric model we present for $\cN=1$ supergravity is based on the notion of a \emph{chiral triple}, introduced in Definition \ref{def:chiraltriple}, which encodes a set of sufficient conditions for a complex manifold $\cM$ to be an admissible target space for the non-linear sigma model of the theory. Strictly speaking, the notion of a chiral triple is not equivalent with the K\"ahler-Hodge condition usually considered in the literature \cite{Witten:1982hu,Andrianopoli:1996cm,FreedmanProeyen}, since the former contains more information and in particular requires the existence of an isomorphism between the pull-back of the holomorphic line bundle appearing in the chiral triple and the determinant line bundle of the $\Spin^c_0(3,1)$ structure $Q$. However, we prove that if $M$ is spin then every K\"ahler-Hodge manifold $\cM$ admits a chiral triple and thus it can occur as the scalar manifold of $\cN=1$ chiral supergravity, whence answering in the positive a question raised in \cite[Section 17.5]{FreedmanProeyen}. 

We couple the theory to a $\Spin^c_0(3,1)$ structure $Q$ by using a delicate interplay between the chiral triple and the determinant line bundle of $Q$ which exists thanks to very specific properties of the Clifford algebra in signature $(3,1)$. Consequently, we formulate the geometric model using complex chiral spinors associated to $Q$. We illustrate the construction with several examples. We find that the Killing spinor equations of chiral $\cN=1$ supergravity yield a very rich system of spinorial equations which are characterized by the crucial role played by the superpotential and its interplay with complex conjugation operation on the spinor bundle. To the best of our knowledge, this system remains unexplored in the mathematical literature and gives rise to a number of outstanding open problems, see Section \ref{sec:conclusions} for more details. Regarding the moduli theory of the Killing spinor equations with non-trivial superpotential, we expect to find a variation of the moduli theory of the perturbed pseudoholomorphicity equations for a map from a Riemann surface into a K\"ahler manifold. In the particular case in which the scalar manifold of the theory is a point, the Killing spinor equations reduce to a particular instance of \emph{generalized Killing spinor equation} as introduced in \cite{FriedrichKim,FriedrichKimII}, in which complex conjugation on the spinor bundle plays a pivotal role.

As an application of the geometric model that we present, we reduce the theory to a Riemann surface and characterize all supersymmetric solutions with vanishing superpotential, obtaining the following result.

\begin{thm}
Supersymmetric solutions of chiral $\cN=1$ ungauged supergravity with vanishing superpotential on a Riemann surface $X$ consist of holomorphic maps $\varphi\colon X \to \cM$ satisfying a stability condition with respect to the linearization given by the canonical bundle of $X$ and the chiral triple of the theory.
\end{thm}

\noindent
A detailed account of this result is given in Theorem \ref{thm:solsusyX}. Using a seminal result by A. Huber, see \cite[Pages 1-2]{HulinTroyanov}, we can exploit the previous theorem to prove the following classification result.
\begin{thm}
If $X$ admits a supersymmetric solution with complete Riemannian metric and finite-energy scalar map, then it is biholomorphic with either $\mathbb{P}^1$, a complex elliptic curve $\mathbb{E}$, the complex plane $\mathbb{C}$ or the punctured complex plane $\mathbb{C}^{\ast}$, with prescribed metric singularities on the one-point and two-point compactification of $\mathbb{C}$ or $\mathbb{C}^{\ast}$, respectively.
\end{thm}

\noindent
A detailed account of this result is given in Theorem \ref{thm:finitecurvature}. Taking the Riemann surface $X$ to be compact, the previous theorem recovers the black-hole horizon topologies classified in \cite{Gutowski:2010gv}, see also \cite{Meessen:2010ph}. In contrast to the $\cN>1$ case, supersymmetric solutions of local chiral $\cN=1$ supergravity have not been systematically classified even at the local level. Supersymmetric solutions with vanishing superpotential were considered in \cite{OrtinN1}, where the local generic form of the metric was obtained, and the problem of constructing supersymmetric solutions was reduced to solving a minimal set of partial differential equations. Local supersymmetric solutions with non-vanishing superpotential were studied in \cite{GGP}, where they were characterized in terms of a minimal set of partial differential equations. Other aspects of supersymmetric solutions of $\cN=1$ supergravity have been explored in \cite{Huebscher:2009bp,Gutowski:2010gv,Meessen:2010ph}. We hope that the geometric model presented in this work can constitute a stepping stone towards the understanding of the global structure of $\cN=1$ supergravity supersymmetric solutions as well as the associated moduli spaces. From a different point of view, a previous contribution to the development of a geometric formulation of supergravity appeared in \cite{Freedman:2016qnq}, where the authors introduced a formulation of supersymmetry which is \emph{covariant} with respect to coordinate transformations in the scalar manifold. We note that the $\cN=1$ Killing spinor equations that we present in this article are by construction automatically covariant in the sense of Op. Cit.

One of the key aspects of supersymmetry is that it provides (partial) integration of the bosonic equations of the theory in terms of (usually) first-order spinorial equations. It is then reasonable to consider the existence of \emph{different} spinorial equations whose solutions are nonetheless solutions of $\cN=1$ supergravity, albeit of a non-supersymmetric type. Inspired by this possibility, in Section \ref{sec:antisusysolutions} we introduce the notion of \emph{anti-supersymmetric solution}, which yields a family of non-supersymmetric solutions on manifolds given by the direct product of $\mathbb{R}^2$ and a hyperbolic Riemann surface. In view of the results of \cite{Gutowski:2010gv}, it is natural to wonder if some of these non-supersymmetric solutions can appear as near-horizon geometries of non-supersymmetric black holes in $\cN=1$ supergravity.

The outline of the paper goes as follows. In Section \ref{sec:CliffordAlgebras} we give a detailed account of Clifford algebras in four Lorentzian dimensions and the associated real and complex irreducible modules. In Section \ref{sec:bosonicsector} we introduce the notion of chiral triple as well as the geometric model for the bosonic sector and Killing spinor equations of chiral $\cN=1$ supergravity, giving several examples. In Section \ref{sec:redX} we reduce the geometric model to a Riemann surface and characterize all supersymmetric solutions with vanishing superpotential. Furthermore, we introduce the notion of \emph{anti-supersymmetric solution} in terms of a natural variation of the Killing spinor equations which still yield solutions of chiral $\cN=1$ supergravity. We provide numerous examples of supersymmetric and anti-supersymmetric solutions. In Section \ref{sec:conclusions} we conclude with a summary of results and a list of related open problems. 

% % % % % % % % % % % % % % % % % % % % % % % % % % % % % % % % % % % % % % 

\subsection*{Acknowledgments}

% % % % % % % % % % % % % % % % % % % % % % % % % % % % % % % % % % % % % % 

CSS would like to thank Mario Garc\'ia-Fern\'andez for useful comments. The work of C. I. L. was supported by grant IBS-R003-S1. The work of C.S.S. is supported by the Humboldt foundation through the Humboldt grant ESP 1186058 HFST-P. The work of V.C. and C.S.S. is supported by the German Science Foundation (DFG) under the Research Training Group 1670 \emph{Mathematics inspired by String Theory}.

% % % % % % % % % % % % % % % % % % % % % % % % % % % % % % % % % % % % % % 
% % % % % % % % % % % % % % % % % % % % % % % % % % % % % % % % % % % % % % 

\section{Clifford algebras in four dimensions}
\label{sec:CliffordAlgebras}

% % % % % % % % % % % % % % % % % % % % % % % % % % % % % % % % % % % % % % 
% % % % % % % % % % % % % % % % % % % % % % % % % % % % % % % % % % % % % % 

In this section we present the background material on four-dimensional Clifford algebras that shall be needed in the rest of the manuscript. Let $(V,h)$ be a four-dimensional real vector space equipped with a non-degenerate symmetric bilinear form $h$ of signature $(3,1) = (+,+,+,-)$. We fix a metric volume form $\nu$ on $(V,h)$. We denote by $\Cl(V,h)$ the real Clifford algebra associated to $(V,h)$.  We use the following convention for the Clifford relation:
\begin{equation*}
v^2 = h(v,v)\, , \qquad v\in V\, .
\end{equation*}

\noindent
The real Clifford algebra $\Cl(V,h)$ admits a unique irreducible real representation:
\begin{equation*}
\gamma\colon \Cl(V,h)\xrightarrow{\sim} \End_{\R}(\Sigma_{\R})\, ,
\end{equation*}

\noindent
on a four-dimensional real vector space $\Sigma_{\R}$. The representation $\gamma$ is in fact an isomorphism of unital, associative, real algebras, and hence, upon a choice of a basis in $\Sigma_{\R}$, $\gamma$ yields an isomorphism $\Cl(V,h)\simeq \Mat(4,\R)$, where $\Mat(4,\R)$ denotes the unital and associative algebra of real four by four square matrices. Using the canonical isomorphism of vector spaces $\Cl(V,h) \simeq \Lambda(V)$ between the Clifford algebra $\Cl(V,h)$ and the exterior algebra of $V$, we can consider the volume form $\nu$ as an element $\nu =e_0 e_1 e_2 e_3$ in $\Cl(V,h)$, which for simplicity we denote with the same symbol. Here $\left\{ e_0 , \hdots e_3\right\}$ denotes a positively oriented orthonormal basis. We have:
\begin{equation*}
\nu^2 = - \mathrm{Id}\, , \qquad \nu\in Z(\Cl^{ev}(V,h))\, ,
\end{equation*}

\noindent
where $\Cl^{ev}(V,h)\subset \Cl(V,h)$ denotes the even subalgebra of $\Cl(V,h)$ and $ Z(\Cl^{ev}(V,h))$ denotes the center of $\Cl^{ev}(V,h)$. Hence, $\gamma(\nu)$ defines an almost complex structure on $\Sigma_{\R}$ which however does not commute with Clifford multiplication. Restricting $\gamma$ to $\Cl^{ev}(V,h)$ we obtain the unique irreducible real representation of $\Cl^{ev}(V,h)$, which is again of real dimension four. Since $\nu\in Z(\Cl^{ev}(V,h))$, the image of $\Cl^{ev}(V,h)$ in $\End(\Sigma_\R)$ consists on endomorphisms complex-linear with respect to the complex structure $\gamma(\nu)$. In fact, it can be shown that the restriction of $\gamma$ to $\Cl^{ev}(V,h)$ gives an isomorphism of unital and associative algebras:
\begin{equation*}
\gamma^{+}\eqdef \gamma\vert_{\Cl^{ev}(V,h)}\colon \Cl^{ev}(V,h)\xrightarrow{\sim} \End(\Sigma_{\R},\gamma(\nu)) \simeq \End_{\C}(\Sigma^{+}_0)\, ,
\end{equation*}

\noindent
where $\End(\Sigma_{\R},\gamma(\nu))$ denotes the $\gamma(\nu)$-linear endomorphisms of $\Sigma_{\R}$ and $\End_{\C}(\Sigma^+_0)$ denotes the complex endomorphisms of $\Sigma^+_0\eqdef (\Sigma_{\R},\gamma(\nu))$, where the later is understood as a two-dimensional complex vector space with complex structure $\gamma(\nu)$. Defining $\Sigma^{-}_0\eqdef (\Sigma_{\R}, - \gamma(\nu))$ with the opposite complex structure with respect to $\Sigma_0$, we obtain the complex-conjugate representation:
\begin{equation*}
\gamma^{-}\eqdef \gamma\vert_{\Cl^{ev}(V,h)}\colon \Cl^{ev}(V,h)\xrightarrow{\sim} \End(\Sigma_{\R},-\gamma(\nu)) \simeq \End_{\C}(\Sigma^{-}_0)\, .
\end{equation*}

\noindent
This way we obtain the two complex-conjugate irreducible representations of $\Cl^{ev}(V,h)$, inducing equivalent real irreducible representations, which correspond to the two chiral irreducible complex representations of $\Cl^{ev}(V,h)$. The spin group $\Spin(V,h)$ injects in $\Cl^{ev}(V,h)$, and the restriction of $\gamma^{+}$ and $\gamma^{-}$ to the image of $\Spin(V,h)$ in $\Cl^{ev}(V,h)$ yields two, complex-conjugate, irreducible spinorial complex representations of $\Spin(V,h)$, which are again equivalent as real representations. Since we are interested in irreducible representations of 
\begin{equation*}
\Spin^c(V,h) = (\Spin(V,h)\times \U(1))/\left\{1,-1\right\}\, ,
\end{equation*}

\noindent
we complexify the previous set up. Let $V_{\C}\eqdef V\otimes_{\R}\C$ denote the complexification of $V$. Extending $h$ by $\C$-linearity to $V_{\C}$ we obtain a complex quadratic space $(V_{\C},h_{\C})$. We have the following isomorphism of complex unital and associative algebras:
\begin{equation*}
\mCl(V,h) \eqdef \Cl(V,h)\otimes\C \simeq \Cl(V_{\C},h_{\C})\, ,
\end{equation*}

\noindent
where $\Cl(V_{\C},h_{\C})$ denotes the complex Clifford algebra associated to $(V_{\C},h_{\C})$. Hence $\mCl(V,h) \simeq \Mat(4,\C)$. Let $\Sigma \eqdef \Sigma_{\R}\otimes \C$. Extending $\gamma$ to $\mCl(V,h)$ by $\C$-linearity we obtain the unique complex irreducible representation $\gamma_\C$ of $\mCl(V,h)$: 
\begin{equation*}
\gamma_{\C} \colon \mCl(V,h) \to \End(\Sigma)\, .
\end{equation*}

\noindent
%explicitly given in terms of $\gamma$ as follows:
%\begin{equation*}
%\gamma_\C (x\otimes z)(\xi\otimes w) =\gamma(x)(\xi) \otimes (zw)\, , \quad x\in \Cl(V,h)\, , \quad z, w \in \C\, , \quad \xi\in \Sigma_\R\, .
%\end{equation*}

%\noindent
%The real Clifford algebra $\Cl(V,h)$ injects in $\mCl(V,h)$ through the map $x\mapsto x\otimes 1$. 
\noindent
Note that we have a canonical inclusion $\Spin^c(V,h)\subset \mCl(V,h)$ and therefore $\Spin^c(V,h)$ has an induced action on $\Sigma$. Restricting $\gamma_{\C}$ to $\Cl(V,h)\subset \mCl(V,h)$ we obtain the unique complex irreducible representation of $\Cl(V,h)$, which for simplicity we denote again by $\gamma_{\C}$. Since the complex irreducible Clifford module $\Sigma$ is the complexification of the irreducible real Clifford module $\Sigma_{\R}$, $\Sigma$ admits a canonical $\Cl(V,h)$ equivariant real structure defined in the usual way:
\begin{equation*}
c\colon \Sigma\to \Sigma\, , \quad \xi\otimes z\mapsto \xi \otimes \bar{z}\, , \quad \xi \in \Sigma_\R\, , \quad z\in \C\, .
\end{equation*}

\noindent
Furthermore, the results of Reference \cite{ACortesI,ACortesII}, see also \cite{LazaroiuBC}, show that the real representation vector space $\Sigma_{\R}$ admits a unique non-degenerate symmetric inner product $\langle - , - \rangle$ such that:
\begin{equation*}
\langle \gamma(v) \xi_1, \xi_2 \rangle = \langle \xi_1, \gamma(v) \xi_2 \rangle\, ,
\end{equation*}

\noindent
for all $ \xi_1 , \xi_2 \in \Sigma_{\R}$ and all $v\in V$. In particular, $\langle -, -\rangle$ is $\mathfrak{spin}(V,h)$-invariant and thus also $\Spin_0(V,h)$-invariant, where $\Spin_0(V,h)\subset \Spin(V,h)$ denotes the connected component of $\Spin(V,h)$ containing the identity. We $\C$-linearly extend $\langle -, -\rangle$ to $\Sigma$. As a $\C$-bilinear inner product in $\Sigma$, $\langle - , - \rangle$ is still $\Spin_0(3,1)$-invariant but fails to be $\Spin_0^c(3,1)$-invariant, since:
\begin{equation*}
\langle [g,z] \xi , [g,z] \xi \rangle = z^2  \langle \xi ,  \xi \rangle \, ,\quad  \forall\,\,  [g,z] \in \Spin^c_0(3,1)\, ,\, \forall\,\,  \xi\in \Sigma\, .
\end{equation*}

\noindent
The simultaneous existence of $\langle - , - \rangle$ and $c(-)$ gives rise to a canonical $\Spin^c_0(3,1)$-invariant Hermitian scalar product, defined as follows:
\begin{equation*}
(\xi_1, \xi_2) = \langle \xi_1 , c(\xi_2)\rangle\, , \quad \forall\,\, \xi_1, \xi_2\in \Sigma\, ,
\end{equation*}

\noindent
From its definition and the properties of $\langle -, - \rangle$ we immediately conclude that:
\begin{equation*}
(\gamma(v) \xi_1, \xi_2) = ( \xi_1, \gamma(v)\xi_2)\, , 
\end{equation*}

\noindent
for all $\xi_1, \xi_2\in \Sigma$ and all $v\in V$. We split $\Sigma = \Sigma^{1,0}\oplus\Sigma^{0,1}$ in terms of eigenspaces of the complex structure $\gamma(\nu)$, which we assume to be extended $\C$-linearly to $\Sigma$. We have the following isomorphisms of complex vector spaces:
\begin{equation*}
\Sigma^{1,0}\simeq \Sigma^+_0\, , \qquad \Sigma^{0,1}\simeq \Sigma^-_0\, .
\end{equation*}  

\noindent
Note that $\Sigma^{1,0}$ and $\Sigma^{0,1}$ are maximally isotropic complex subspaces of $\Sigma$ with respect to $(-,-)$ and hence $(-,-)$ is of split signature $(2,2)$.

Let $\mCl^{ev}(V,h) \eqdef \Cl^{ev}(V,h)\otimes\C$ denote the even Clifford subalgebra of $\mCl(V,h)$. The restriction of $\gamma_{\C}$ to $\mCl^{ev}(V,h)$:
\begin{equation*}
\gamma_{\C}\colon \mCl^{ev}(V,h)\to \End(\Sigma^{1,0}\oplus \Sigma^{0,1})\, ,
\end{equation*}

\noindent 
preserves both $\Sigma^{1,0}$ and $\Sigma^{0,1}$. Therefore, the restriction of $\gamma_{\C}$ to $\mCl^{ev}(V,h)$  splits as a sum of the irreducible representations
\begin{equation*}
\gamma^{+}_\C \colon \mCl^{ev}(V,h)\to \End(\Sigma^{1,0})\, , \qquad \gamma^{-}_\C \colon \mCl^{ev}(V,h)\to \End(\Sigma^{0,1})\, ,
\end{equation*}

\noindent
defined by proyection of $\gamma_\C$ on the corresponding factor. Note that we have complex isomorphisms:
\begin{equation*}
\End(\Sigma^{1,0})\simeq \End_{\C}(\Sigma^{+}_{0})\, , \qquad  \End(\Sigma^{0,1}) \simeq \End_{\C}(\Sigma^{-}_{0})\, ,
\end{equation*}

\noindent
We define the complex volume form $\nu_{\C} \eqdef i\nu\in \mCl(V,h)$, which satisfies:
\begin{equation*}
\nu_{\C}^2 = \mathrm{Id}\, , \qquad \nu_{\C} \in Z(\mCl^{ev}(V,h))\, ,
\end{equation*}

\noindent
In terms of $\nu_{\C}$, $\mCl^{ev}(V,h)$ splits as the direct sum of two unital, associative complex algebras as follows:
\begin{equation*}
\mCl^{ev}(V,h) = \mCl^{ev}_+(V,h) \oplus \mCl^{ev}_-(V,h)\, , \qquad \mCl^{ev}_{\pm}(V,h) = \frac{1}{2}(1\mp i \nu)\mCl^{ev}(V,h)\, .
\end{equation*}

\noindent
For ease of notation we define the projectors $P_{\pm}\eqdef \frac{1}{2}(1\mp i\gamma(\nu))\colon \mCl^{ev}(v,h)\to \mCl^{ev}_{\pm}(V,h)$. We have the following isomorphisms of complex unital and associative algebras: 
\begin{equation*}
\mCl^{ev}_+(V,h) \simeq \End_{\C}(\Sigma^+_0)\, , \qquad \mCl^{ev}_-(V,h) \simeq \End_{\C}(\Sigma^-_0)\, .
\end{equation*}

\noindent
The restriction of $\gamma_{\C}$ to $\mCl^{ev}(V,h)$, which splits as a direct sum in terms of the two inequivalent irreducible complex representations $\gamma^{\pm}_{\C}$ of $\mCl^{ev}(V,h)$, factors through projection on the given factor of $\mCl^{ev}(V,h)$ as follows:
\begin{equation*}
\gamma^{\pm}_{\C}\colon\mCl^{ev}(V,h) \to \mCl^{ev}_{\pm}(V,h)\to \End_{\C}(\Sigma^{\pm}_0)\, , \quad x\otimes z \mapsto z\, P_{\pm} (x) \mapsto (\Re z \pm \Im z \gamma(\nu))\circ \gamma(x)\, .
\end{equation*}

\noindent
Note that $\gamma^+_{\C}$ and $\gamma^-_{\C}$ are not complex conjugate of each other, although they induce isomorphic real representations of $\Cl^{ev}(V,h)$ and are in fact complex-conjugate representations of the latter. Restricting $\gamma^{\pm}_{\C}$ to $\Spin^c(V,h)\subset \mCl^{ev}(V,h)$ we obtain two of the four irreducible representations of $\Spin^c(V,h)$ that will play a role in the mathematical formulation of four-dimensional chiral $\cN=1$ supergravity. We define:
\begin{equation*}
\tau \eqdef \gamma_{\C}\vert_{\Spin^c(V,h)}\colon \Spin^c(V,h)\to \Aut_{\C}(\Sigma)\, , \quad \tau^{\pm} \eqdef \gamma^{\pm}_{\C}\vert_{\Spin^c(V,h)}\colon \Spin^c(V,h)\to \Aut_{\C}(\Sigma^{\pm}_0)\, .
\end{equation*}

\noindent
The real structure $c\colon \Sigma \to \Sigma$ intertwines the complex representations $\tau^{\pm}$ through complex conjugation in the second factor, that is:
\begin{equation*}
c\circ \tau^{\pm}[g,z] = \tau^{\mp}[g,\bar{z}] \circ c\, , \qquad \forall\,  [g,z]\in \Spin^c(V,h)\, ,
\end{equation*}

\noindent
The failure for $\tau^+$ and $\tau^-$ to be complex-conjugate of each other is given by the outer-automorphism of $\Spin^c(V,h)$ defined by complex conjugation on the $\U(1)$ factor. For future reference, we introduce the following irreducible representations:
\begin{equation*}
\tau^{\pm}_c\colon \mCl^{ev}(V,h)\to \End_{\C}(\Sigma_0^{\pm})\, , \qquad \tau^{\pm}_c([g,z]) \eqdef \tau^{\pm}([g,\bar{z}])\, , \qquad \forall\,\, [g,z]\in \Spin^{c}(V,h)\, .
\end{equation*}

\noindent
by composition of $\tau^{\pm}$ with the complex-conjugation automorphism of $\Spin^c(V,h)$. With this definition we have $\tau^{\pm}[g,z]= c\circ \tau^{\mp}_c[g,z] \circ c$ and hence $\tau^{\mp}_c$ is the complex conjugate representation of $\tau^{\pm}$. With respect to Clifford multiplication, the irreducible representations $\tau^{\pm}$ and $\tau^{\pm}_c$ relate as follows:
\begin{equation*}
\tau^{\mp}([g,z])\circ \gamma(v) = \gamma(\Ad_{g^{-1}}(v))\circ \tau^{\pm}([g,z]) \, , \qquad \tau^{\mp}_c([g,z])\circ \gamma(v) = \gamma(\Ad_{g^{-1}}(v))\circ \tau^{\pm}_c([g,z])\, ,
\end{equation*}

\noindent
for all $v \in V$ and $[g,z]\in \Spin^c(V,h)$. Therefore, Clifford multiplication by non-zero elements of $V$ does not preserve the chiral splitting, or, in other words, Clifford multiplication by non-zero elements of $V$ defines a complex linear map form $\Sigma^{\pm}_0$ to $\Sigma^{\mp}_0$. To summarize, we have introduced four inequivalent irreducible complex representations $\tau^{\pm}$ and $\tau^{\pm}_c$ of $\Spin^c(3,1)$. Note that $\tau^{\pm}$ is complex conjugate to $\tau^{\mp}_c$ and therefore are equivalent as real representations.

% % % % % % % % % % % % % % % % % % % % % % % % % % % % % % % % % % % % % % 
% % % % % % % % % % % % % % % % % % % % % % % % % % % % % % % % % % % % % % 

\section{Bosonic sector and Killing spinor equations}
\label{sec:bosonicsector}

% % % % % % % % % % % % % % % % % % % % % % % % % % % % % % % % % % % % % % 
% % % % % % % % % % % % % % % % % % % % % % % % % % % % % % % % % % % % % % 

In this section we construct a global geometric model for the Killing spinor equations and bosonic sector of four-dimensional chiral $\cN=1$ supergravity \cite{Cremmer:1978bh,Cremmer:1978hn}, that is, $\cN=1$ supergravity coupled to an arbitrary number of chiral multiplets. We will refer to such theory as \emph{$\cN =1$ chiral supergravity} or chiral supregravity for short.

% % % % % % % % % % % % % % % % % % % % % % % % % % % % % % % % % % % % % % 

\subsection{Chiral spinor bundles}

% % % % % % % % % % % % % % % % % % % % % % % % % % % % % % % % % % % % % % 

Let $(M,g)$ be a connected and oriented Lorentzian four-manifold, with Lorentzian metric $g$. We denote by $\Cl(M,g)$ the bundle of real Clifford algebras associated to $(M,g)$, whose typical fiber is the real Clifford algebra $\Cl(3,1)\simeq \Mat(4,\mathbb{R})$ of signature $(3,1)$. The construction of supergravity crucially relies on the existence of a vector bundle over $M$ equipped with Clifford multiplication, whose sections correspond with the \emph{supersymmetry generators} or \emph{parameters} of the theory. Hence, we will assume that $(M,g)$ is equipped with a bundle of irreducible complex Clifford modules, that is, a pair $(S,\gamma)$ where $S$ is a complex vector bundle over $M$ and:
\begin{equation*}
\gamma\colon \Cl(M,g)\to \End(S)\, ,
\end{equation*} 

\noindent
is a morphism of bundles of unital associative algebras such that, for every point $p\in M$, the restriction $\gamma_p\colon \Cl(M,g)_{p}\to \End(S_{p})$ of $\gamma$ to the fiber of $\Cl(M,g)$ over $p$ is an irreducible complex Clifford representation. The irreducibility condition is required in order to match the local degrees of freedom of the $\cN=1$ supersymmetry generator of the theory. Existence of $(S,\gamma)$ is obstructed. The obstruction was computed in References \cite{FriedrichTrautman,LazaroiuComplex}, where it was shown that $(S,\gamma)$ exists on $(M,g)$ if and only if $(M,g)$ admits a $\Pin^c(3,1)$ structure, in which case there exists a unique (modulo isomorphisms) $\Pin^c(3,1)$ structure $Q^{\mathfrak{pin}}$ on $(M,g)$ such that $S$ is the vector bundle associated to $Q^{\mathfrak{pin}}$ through the tautological representation of $\Pin^c(3,1)\subset \Cl(3,1)\otimes \mathbb{C}$. Recall that $(M,g)$ admits a $\Pin^c(3,1)$ structure $Q^{\mathfrak{pin}}$ if and only if there exists a $\U(1)$ principal bundle $\P_{Q^{\mathfrak{pin}}}$ such that:
\begin{equation*}
w_{2}(M) +  w^{-}_1(M)^2 + w^{-}_1(M) w^{+}_1(M) = w_{2}(\P_{Q^{\mathfrak{pin}}})\, ,
\end{equation*}

\noindent
where $w_{2}(M)$ denotes the second Stiefel-Whitney class of $M$, $w^{-}_1(M)$ denotes the first Stiefel-Whitney class of the bundle of time-like lines of $(M,g)$, $w^{+}_1(M)$ denotes the first Stiefel-Whitney class of the bundle of space-like planes of $(M,g)$ and $w_{2}(\P_{Q^{\mathfrak{pin}}})$ denotes the second Stiefel-Whitney class of $\P_{Q^{\mathfrak{pin}}}$, understood as an $\SO(2)\simeq \U(1)$ principal bundle. Let $\nu$ be the Lorentzian volume form on $(M,g)$. The complex volume form $\nu_{\C} = i\nu$ acts as an involution on $S$ and hence induces a splitting
\begin{equation*}
S = S^{+} \oplus S^{-}\, ,
\end{equation*}

\noindent
of $S$ in terms of the so-called chiral bundles $S^{+}$ and $S^{-}$. Note that this splitting is not preserved by the full Clifford algebra but only by its even part. In fact, $S^{+}$ and $S^{-}$ are inequivalent bundles of irreducible complex Clifford modules over $\Cl^{ev}(M,g)$, where $\Cl^{ev}(M,g)\subset \Cl(M,g)$ denotes de bundle of even Clifford algebras over $(M,g)$, with typical fiber isomorphic to $\Cl^{ev}(3,1) \simeq \Mat(2,\mathbb{C})$. Furthermore, the $\Pin^c(3,1)$ structure $Q^{\mathfrak{pin}}$ reduces to a $\Spin^c(3,1)$ structure $Q$ and we can understand $S^{+}$ and $S^{-}$ as being vector bundles associated to $Q$ by means of the two inequivalent tautological representations $\tau^{\pm}$ of $\Spin^c(3,1)\subset \Cl^{ev}(3,1)\otimes \mathbb{C}$ introduced in Section \ref{sec:CliffordAlgebras}. Therefore, we can write:
\begin{equation*}
S = Q\times_{\tau} \Sigma\, , \qquad S^{\pm} = Q\times_{\tau^{\pm}} \Sigma^{\pm}_0\, .
\end{equation*}

\begin{remark}
The previous discussion shows that a necessary condition to formulate chiral $\cN=1$ supergravity on an oriented Lorentzian manifold $(M,g)$ is for $(M,g)$ to admit a $\Spin^c(3,1)$ structure.
\end{remark}

\noindent
For future reference, we define:
\begin{equation*}
S^{+}_c \eqdef Q\times_{\tau^{+}_c} \Sigma^{+}_0\, , \qquad S^{-}_c \eqdef Q\times_{\tau^{-}_c} \Sigma^{-}_0 \, .
\end{equation*}

\begin{remark}
We have introduced four complex spinor bundles, namely $S^{+}$, $S^{-}$, $S^{+}_c$ and $S^{-}_c$, all of which are associated to the same $\Spin^c(3,1)$ structure $Q$ and all of which will play a relevant role in the formulation of chiral supergravity. 
\end{remark}

\noindent
Given a $\Spin^c(3,1)$ structure $Q$ on $M$, we will denote by $\P_{Q}$ its associated characteristic $\U(1)$-bundle and by $\L_{Q}$ its associated determinant complex line bundle. These are defined in terms of $Q$ as follows:
\begin{equation*}
\P_{Q} \eqdef Q\times_{l} \U(1)\, , \qquad \L_Q \eqdef Q\times_{l} \C\, ,
\end{equation*}

\noindent
where $l\colon \Spin(3,1)\to \U(1)$ is the homomorphism of groups defined through $l([g,z]) = z^2$ acting by left-multiplication on $\U(1)$ and $\C$, respectively.
\begin{remark}
The space-time manifold $(M,g)$ may admit non-equivalent $\Spin^c(3,1)$ structures. The set of isomorphism classes of $\Spin^c(3,1)$ structures can be shown to have the structure of a torsor over the abelian group $H^2(M,\mathbb{Z})$, see for example \cite{Friedrich}. Different choices of $\Spin^c(3,1)$ structure yield, in principle, inequivalent chiral supergravity theories.	
\end{remark}
 
\noindent
\begin{lemma}
The real structure $c\colon \Sigma^{\pm}_0 \to \Sigma^{\mp}_0$ induces the following anti-isomorphism of complex bundles:
\begin{equation*}
\mathfrak{c}\colon S^{\pm}\to S^{\mp}_c\, , \qquad [q,\xi] \mapsto [q,c(\xi)]\, ,
\end{equation*}

\noindent
which yields an isomorphism of bundles of real Clifford modules.
\end{lemma}

\begin{proof}
The fact that $\mathfrak{c}$ is well-defined as an anti-isomorphism of complex vector bundles follows from the anti-linearity of $c\colon \Sigma^{\pm}_0\to \Sigma^{\mp}_0$ and the following computation:
\begin{equation*}
[q\cdot [g,z], c(\tau^{\pm}([g,z]^{-1}) \xi)] = [q\cdot [g,z], \tau^{\mp}_c([g,z]^{-1}) c(\xi)] = [q, c(\xi)]\, ,
\end{equation*}

\noindent
which holds for all $[g,z]\in \Spin^c(3,1)$. Additionally, $\mathfrak{c}$ induces an isomorphism of bundles of real Clifford modules since $\gamma(v)\circ \mathfrak{c} = \mathfrak{c}\circ \gamma(v)$ for all $v\in TM$.
\end{proof}

\begin{prop}
\label{prop:detS}
The following isomorphisms of complex line bundles hold:
\begin{equation*}
\L_Q  \simeq \Lambda^2(S^{+}) \simeq \Lambda^2(S^{-})\, , \qquad \L^{-1}_Q  \simeq \Lambda^2(S^{+}_c) \simeq \Lambda^2(S^{-}_c)\, .
\end{equation*}
\end{prop}

\begin{proof}
The result follows from the explicit form of the determinant representation of $\tau^{\pm}\colon \Spin^c(3,1)\to \Aut(\Sigma^{\pm}_0)$ and $\tau^{\pm}_c\colon \Spin^c(3,1)\to \Aut(\Sigma^{\pm}_0)$, which reads:
\begin{eqnarray*}
\det(\tau^{\pm}([g,z])) & = &  z^2 \, \in \U(1)\, , \\
\det(\tau^{\pm}_c[g,z])) & = &  \bar{z}^{2}\, \in \U(1)\, ,
\end{eqnarray*}

\noindent
where we have used that $\det(\gamma(g)) = 1$ for all $g\in \Spin(3,1)$.
\end{proof}

\noindent
To construct chiral supergravity it is convenient to endow $S$ with the sesquilinear pairing $(-,-)$ and the bilinear pairing $\langle -, -\rangle$ introduced in Section \ref{sec:CliffordAlgebras} on the representation spaces $\Sigma$ and $\Sigma^{\pm}_0$. Since they are respectively invariant under $\Spin^c_0(3,1)$ and $\Spin_0(3,1)$ transformations, for this to be possible we need the $\Spin^c(3,1)$ structure $Q$ to further reduce to a $\Spin^c_0(3,1)$ structure. This is in general obstructed.

\begin{prop}
\label{prop:Spinc0}
A $\Spin^c(3,1)$ structure $Q$ on $(M,g)$ reduces to a $\Spin^c_0(3,1)$ structure if and only if:
\begin{equation*}
w^{-}_1(M,g) = 0\, ,
\end{equation*}

\noindent
where $w^{-}_1(M,g)$ denotes the first Stiefel-Whitney class of the $\O(1)$ principal bundle associated to any $\O(1) \times\O(3)$ reduction of the orthonormal frame bundle. 
\end{prop}

\begin{proof}
The group $\Spin_0(3,1)$ is the universal cover of $\SO_0(3,1)$, where $\SO_0(3,1)$ denotes the connected component of $\SO(3,1)$ containing the identity. The group $\SO_0(3,1)$ is in particular time- and space-orientation preserving. If $Q$ is a $\Spin^c_0(3,1)$ structure on $(M,g)$, its image inside the orthonormal frame bundle induces a reduction to $\SO_0(3,1)$, implying that $(M,g)$ is time-oriented and space-oriented. Hence $w^{-}_1(M,g) = 0$. To prove the converse, we use that $(M,g)$ is oriented and note that if in addition $w^{-}_1(M,g) = 0$, then the frame bundle admits a reduction to $\SO_0(3,1)$. This reduction lifts to a $\Spin^c_0(3,1)$ structure by taking its preimage through the projection of the $\Spin^c(3,1)$ structure $Q$ to the oriented orthonormal frame bundle of $(M,g)$.
\end{proof}

\noindent
We will assume that the obstruction is satisfied and use for simplicity the same symbol $Q$ to denote the reduced $\Spin^c_0(3,1)$ structure. Note that $w^{-}_1(M,g) = 0$ is equivalent to $(M,g)$ being time-orientable, a condition which is usually assumed in Lorentzian geometry as part of the definition of a \emph{space-time}, see for example \cite{LorentzianGeometry}. In particular we will assume that, given a choice of Lorentzian metric $g$, $M$ is endowed with a fixed time-orientation and space-orientation. 

We define:
\begin{eqnarray}
\label{eq:bilinears}
& \cO &\colon  S \times S \to \mathrm{L}_Q\, , \quad [q,\xi_{0}\otimes z]\times [q,\eta_{0}\otimes w]\mapsto \langle \xi_{0},\eta_{0}\rangle\, z w\, , \\
& \cB &\colon S \times S \to \underline{\C}\, , \quad [q,\xi_{0}\otimes z]\times [q,\eta_{0}\otimes w]\mapsto (\xi_{0}\otimes z, \eta_{0}\otimes w) = \langle\xi_{0},\eta_{0}\rangle\, z \bar{w}\, , 
\end{eqnarray}

\noindent
for all $q\in Q$, $\xi_{0}, \eta_{0} \in \Sigma_{\R}$ and $z, w \in \C$. Note that $\cO$ is complex bilinear and $\cB$ is a non-degenerate  Hermitian inner product of split signature $(2,2)$ for which the chiral bundles $S^{\pm}$ are maximally isotropic vector subbundles of $S$. 

The tensor product connection of the Levi-Civita connection $\nabla^g$ on $(M,g)$ with any connection $\nabla_A$ on $\P_Q$ defines a connection on $\Fr^g_0(M)\times_M \P_Q$ which we denote by $\nabla^g_A$, where $\Fr^g_0(M)$ denotes principal $\SO_0(3,1)$ bundle of restricted $g$-orthonormal frames. Hence $\Fr^g_0(M)\times_M \P_Q$ is an $\SO_0(3,1)\times \U(1)$ principal bundle. The connection $\nabla^g_A$ on $\Fr^g_0(M)\times_M \P_Q$ canonically lifts to a connection on the $\Spin^c_0(3,1)$ bundle $Q$ which in turn induces a connection on $S$ which preserves the chiral splitting $S^{\pm}$. For simplicity, we denote the lifted connection by the same symbol.

We introduce now the description of $\Spin^c_0(3,1)$ structures that will be used in the definition of chiral triple. Let $M$ be an oriented four-manifold, and let $\P_{\Gl_{+}}$ denote the associated principal $\Gl_{+}(4,\mathbb{R})$ bundle of oriented frames, where $\Gl_{+}(4,\mathbb{R})$ denotes the maximal connected subgroup of $\Gl(4,\mathbb{R})$. We have $\pi_1(\Gl_{+}(4,\mathbb{R})) = \mathbb{Z}_2$ and $\Gl_{+}(4,\mathbb{R})$ fits in the following short exact sequence
\begin{equation*}
1\to \mathbb{Z}_{2} \to \widetilde{\Gl}_{+}(4,\mathbb{R}) \xrightarrow{\tilde{\lambda}} \Gl_{+}(4,\mathbb{R}) \to 1\, ,
\end{equation*}

\noindent
where $\widetilde{\Gl}_{+}(4,\mathbb{R})$ is the universal covering of $\Gl_{+}(4,\mathbb{R})$ and $\tlambda$ denotes the associated cover map. We define the group $\widetilde{\Gl}^c_{+}(4,\mathbb{R})$ as:
\begin{equation*}
\widetilde{\Gl}^c_{+}(4,\mathbb{R}) = \widetilde{\Gl}_{+}(4,\mathbb{R})\cdot \U(1) = \left( \widetilde{\Gl}_{+}(4,\mathbb{R})\times \U(1)\right)/\mathbb{Z}_2\, ,
\end{equation*}

\noindent
which fits into the following short exact sequence:
\begin{equation*}
1\to \mathbb{Z}_{2} \to \widetilde{\Gl}^c_{+}(4,\mathbb{R}) \xrightarrow{\tilde{\lambda}_c} \Gl_{+}(4,\mathbb{R})\times \U(1) \to 1\, ,
\end{equation*}

\noindent
where $\tilde{\lambda}_c$ denotes a $\mathbb{Z}_2$ cover map given by $[a,u]\mapsto (\tilde{\lambda}(a),u^2) $ for every $[a,u]\in  \widetilde{\Gl}^c_{+}(4,\mathbb{R})$.

\begin{remark}
Note that $\widetilde{\Gl}^c_{+}(4,\mathbb{R})$ is not a matrix group, since it does not admit any finite-dimensional faithful representation.
\end{remark}

\begin{definition}
A $\widetilde{\Gl}^c_{+}(4,\mathbb{R})$ structure on $M$ is a pair $(\P_{\widetilde{\Gl}^c_{+}}, \widetilde{\Lambda}_c)$ where $\widetilde{\P}_{\widetilde{\Gl}^c_{+}}$ is principal bundle with structure group $\widetilde{\Gl}^c_{+}(4,\mathbb{R})$ and $\widetilde{\Lambda}_c$ is a principal bundle map:
\begin{equation*}
\widetilde{\Lambda}_c \colon \P_{\widetilde{\Gl}^c_{+}} \to \P_{\Gl_{+}}\, ,
\end{equation*}

\noindent
satisfying:
\begin{equation*}
\widetilde{\Lambda}_c(q u) = \widetilde{\Lambda}_c(q) \tlambda_c(u)\, ,
\end{equation*}

\noindent
for every $q\in \P_{\widetilde{\Gl}^c_{+}}$ and every $u\in \widetilde{\Gl}^c_{+} (4,\mathbb{R})$.
\end{definition}

\noindent
To every $\widetilde{\Gl}^c_{+}(4,\mathbb{R})$ structure we can associate a determinant complex line bundle $\widetilde{\L}$ and a characteristic principal $\U(1)$ bundle $\widetilde{\P}$, just in the same way as we did for $\Spin^c(3,1)$ structures.

\begin{prop}
\label{prop:GLSpin}
A fixed choice of $\widetilde{\Gl}^c_{+}(4,\mathbb{R})$ structure $(\P_{\widetilde{\Gl}^c_{+}}, \widetilde{\Lambda}_c)$ on $M$ induces, for every Lorentzian metric $g$ such that $w^{-}_1(M,g) = 0$, a canonical $\Spin^c_0(3,1)$ structure $Q$ on $(M,g)$. Furthermore the associated characteristic and determinant bundles are isomorphic:
\begin{equation*}
\widetilde{\P} \simeq \P_Q\, , \qquad \widetilde{\L} \simeq \L_Q\, .
\end{equation*}
\end{prop}

\begin{proof}
A Lorentzian metric $g$ on $M$ induces a $\SO(3,1)$ reduction of the bundle of oriented frames of $M$, which, using that by assumption $w^{-}_1(M,g) = 0$, further reduces to a $\SO_0(3,1)$ structure $\mathrm{Fr}^g_{0}(M)$. Then:
\begin{equation*}
\iota \colon \mathrm{Fr}^g_0(M)\hookrightarrow \P_{\Gl_{+}}\, ,
\end{equation*} 

\noindent
is a connected embedded submanifold and:
\begin{equation*}
Q\eqdef \tilde{\Lambda}_c^{-1}(\mathrm{Fr}^g_0(M))\hookrightarrow \P_{\widetilde{\Gl}^c_{+}}\, ,
\end{equation*} 

\noindent
defines a reduction of $\P_{\widetilde{\Gl}^c_{+}}$ to a $\Spin^c_0(3,1)$ principal bundle $Q$. The restriction of $\tilde{\Lambda}_c$ to $Q$ defines a bundle covering map:
\begin{equation*}
\Lambda_c\colon Q\to \mathrm{Fr}^g_{0}(M)\, ,
\end{equation*}

\noindent
satisfying the equivariance property with respect to $\lambda_c\colon \Spin^c_0(3,1)\to \SO_0(3,1)$ which makes the pair $(Q,\Lambda_c)$ into a $\Spin^c_0(3,1)$ structure on $(M,g)$. The isomorphisms $\widetilde{\P} \simeq \P_Q$ and $\widetilde{\L} \simeq \L_Q$ follow from the fact that the homomorphism:
\begin{equation*}
\Spin^c_0(3,1) \to \U(1)\, , \qquad [g,z] \mapsto z^2\, , 
\end{equation*}

\noindent
factors through the inclusion $\Spin^c_0(3,1) \hookrightarrow\widetilde{\Gl}^c_{+}(4,\mathbb{R})$.
\end{proof}

\noindent
The choice of a $\widetilde{\Gl}^c_{+}(4,\mathbb{R})$ structure allows to define a canonical $\Spin^c_0(3,1)$ structure associated to every Lorentzian metric $g$. This fact will play a relevant role in the definition of chiral supergravity and the notion of chiral triple, see Definition \ref{def:chiraltriple}. 

% % % % % % % % % % % % % % % % % % % % % % % % % % % % % % % % % % % % % % 

\subsection{Chiral triples and the scalar potential}

% % % % % % % % % % % % % % % % % % % % % % % % % % % % % % % % % % % % % %

Once we have described the spinor bundles that shall be used to formulate the geometric model for the Killing spinor equations of chiral supergravity, we need to introduce the concept of \emph{chiral triple} on a complex manifold $\cM$, which plays the role of target space for the non-linear sigma model of the theory. This is the so-called \emph{scalar manifold} in the physical literature. Let us fix a complex manifold $\cM$ of K\"ahler type, that is, admitting a K\"ahler metric, of real dimension $2n$. When necessary, we will denote the complex structure on $\cM$ by $\cI$.

\begin{definition}
\label{def:positivecL}
Let $(\cL,\cH)$ be a Hermitian holomorphic line bundle over $\cM$, with Hermitian structure $\cH$. We say that $(\cL,\cH)$ is \emph{positive} (or has positive curvature) if the Chern curvature $\Theta$ associated to $\cH$ defines a Riemannian metric $\cG$ on $\cM$ through the following formula:
\begin{equation*}
2\pi \, \cG \eqdef  i\,\Theta \circ (\mathrm{Id}\otimes \cI)\, .
\end{equation*}

\noindent
We say, that $(\cL,\cH)$ is \emph{negative} (or has negative curvature) if its dual, equipped with the induced holomorphic and Hermitian structures, is positive. If $(\cL,\cH)$ is negative, we define the associated positive-definite metric as follows:
\begin{equation*}
2\pi \, \cG \eqdef - i\,\Theta \circ (\mathrm{Id}\otimes \cI)\, .
\end{equation*}

\noindent
To ease the notation we will sometimes denote the holomorphic Hermitian line bundle $(\cL,\cH$) simply by $\cL_{\cH}$.
\end{definition}

\noindent
Recall that:
\begin{equation}
\label{eq:KahlerChern}
\mathcal{V} \eqdef \frac{i}{2\pi} \Theta\in \Omega^{1,1}_{\mathbb{Z}}(\cM)\, ,
\end{equation}

\noindent
defines a real $(1,1)$ integral closed form whose associated de Rham cohomology class $[\cV]$ is equal to:
\begin{equation*}
[\cV] = (j\circ c_1)(\cL)\in j(H^2(\cM,\Z))\, .
\end{equation*}

\noindent
Here $c_1\colon \Pic(\cM)\to H^2(\cM,\Z)$ denotes the first Chern class map and $j\colon H^2(\cM,\Z)\rightarrow H^{2}(\cM,\C)$ denotes the canonical homomorphism mapping the singular integral cohomology of $\cM$ into its de Rahm cohomology with coefficients in $\C$. Given a holomorphic Hermitian line bundle $\cL_{\cH}$ we denote the associated principal $\U(1)$ bundle by $\cP_{\cH}$. Recall that the Chern connection $\cD$ of $(\cL,\cH)$ is induced by a unique connection on $\cP_{\cH}$, which we denote by $\cA$.

\begin{definition}
\label{def:chiraltriple}
Let $\cM$ be a complex manifold of K\"ahler type and $M$ an oriented four-manifold. A {\em chiral triple} on $(M,\cM)$ is a  tuple $\mathfrak{Q}\eqdef (\cL_{\cH}, \P_{\widetilde{\Gl}^c_{+}},\cW)$, where:

\begin{itemize}
\item $\cL_{\cH} \eqdef (\cL,\cH)$ is a negative holomorphic Hermitian line bundle on $\cM$ and $\P_{\widetilde{\Gl}^c_{+}}$ is a $\widetilde{\Gl}^c_{+}(4,\mathbb{R})$ structure on $M$ for which there exists a smooth map $\varphi \colon M\to \cM$ such that:
\begin{equation}
\label{eq:pullback}
\widetilde{\L} \simeq \cL^{\varphi}\, .
\end{equation}

\noindent
That is, the determinant line bundle $\widetilde{\L}$ of $\P_{\widetilde{\Gl}^c_{+}}$ is isomorphic with the pull back $\cL^\varphi$ by $\varphi$ of the holomorphic line bundle $\cL$. 

\

\item $\cW\in H^0(\cL)$ is a holomorphic section of $\cL$, which is usually referred to as the \emph{superpotential} in the physics literature.
\end{itemize} 
\end{definition}

\noindent
For future reference, we note the following. There is a canonical embedding:
\begin{equation*}
\iota\colon \widetilde{\P} = \P_{\widetilde{\Gl}^c_{+}}\times_l \U(1) \hookrightarrow \widetilde{\L} = \P_{\widetilde{\Gl}^c_{+}}\times_l\mathbb{C}\, ,\qquad [p,u]\mapsto [p,u]\, ,
\end{equation*}

\noindent
which is constructed by using the presentation of $\widetilde{\P}$ and $\widetilde{\L}$ as bundles associated to $\P_{\widetilde{\Gl}^c_{+}}$. Furthermore, there is a canonical projection:
\begin{equation*}
pr^{\varphi}_{\cH}\colon \cL^{\varphi}_{\cH} \to \cP^{\varphi}_{\cH}\, , \qquad l\mapsto \frac{l}{\cH^{\varphi}(l,l)^{\frac{1}{2}}}\, ,
\end{equation*}
and, for every choice of isomorphism $\Psi\colon \widetilde{\L}\xrightarrow{\simeq} \cL^{\varphi}$, there exists a canonical isomorphism of principal bundles:
\begin{equation*}
\Psi_{\cH}\colon \widetilde{\P}\xrightarrow{\simeq} \cP^{\varphi}_{\cH}\, ,
\end{equation*}

\noindent
defined as follows:
\begin{equation*}
\Psi_{\cH} \eqdef pr_{\cH}^{\varphi}\circ\Psi\circ\iota \, .
\end{equation*}

\noindent
Conversely, every isomorphism $\Psi_{\cH}\colon \widetilde{\P}\xrightarrow{\simeq} \cP^{\varphi}_{\cH}$ induces a canonical isomorphism of complex line bundles $\Psi\colon \widetilde{\L}\xrightarrow{\simeq} \cL^{\varphi}$ in the usual way by using that $\widetilde{\L}$ and $\cL^{\varphi}$ are complex line bundles associated to $\widetilde{\P}$ and $\cP^{\varphi}_{\cH}$. For simplicity we will denote $\Psi_{\cH}$ simply by $\Psi$.

\begin{remark}
A chiral triple $\mathfrak{Q}$ induces, for any choice of Lorentzian metric $g$ on the oriented four-manifold $M$ such that $w_1^{-}(M,g) = 0$, a canonical  $\Spin^c_0(3,1)$ structure $Q_g$, whose determinant line bundle and characteristic bundles are respectively isomorphic to those of $\P_{\widetilde{\Gl}^c_{+}}$, see Proposition \ref{prop:GLSpin}. For simplicity, we denote the $\Spin^c_0(3,1)$ structure $Q_g$ simply by $Q$, with the implicit understanding that, given a Lorentzian metric $g$ on $M$ with $w_1^{-}(M,g) = 0$, the chosen $\Spin^c_0(3,1)$ structure $Q$ on $(M,g)$ corresponds with the canonical $\Spin^c_0(3,1)$ structure determined by $\mathfrak{Q}$ on $(M,g)$. In particular, we have:
\begin{equation*}
\P_Q \simeq \cP^{\varphi}_{\cH}\, , \qquad \L_Q \simeq \cL^{\varphi}\, .
\end{equation*}

\noindent
We will only consider Lorentzian metrics whose associated bundle of time-like lines is trivializable, that is, satisfying $w_1^{-}(M,g) = 0$.
\end{remark}

\begin{remark}
Note that a given K\"ahler manifold $\cM$ need not admit, in general, any chiral triple for any four-manifold $M$. If $\cM$ is compact, Kodaira's embedding theorem implies that a necessary condition for $\cM$ to admit a chiral triple is that it be projective, since it must admit a positive line bundle. Therefore, non-algebraic complex manifolds of K\"ahler type, such as non-algebraic tori or K3 surfaces, give a large number of compact K\"ahler manifold not admitting chiral triples. The existence of a chiral triple is also obstructed for non-compact K\"ahler manifolds. To see this, we use an argument of Verbitsky on the non-existence of positive holomorphic line bundles on certain non-compact K\"ahler manifolds. Consider for instance a K3 surface with no non-zero integral $(1,1)$ classes and remove a point $p$. This is a non-compact K\"ahler manifold which does not admit positive line bundles. Indeed, by the previous assumptions, a positive line bundle would imply the existence of an exact positive form $\omega$. By Sibony's Lemma, see \cite[Theorem 5.1]{Verbitsky}, such $\omega$ would be locally integrable around $p$. Therefore, applying Skoda-El Mir Theorem, the trivial extension of $\omega$ to the original K3 surface is a closed, positive and hence exact current. This is impossible, because the K3 surface is closed and K\"ahler. 
\end{remark}

\noindent
In the following we will assume that $(M,\cM)$ does admit chiral triples. Given a chiral triple $\mathfrak{Q}$ we denote by $\Map(\mathfrak{Q})\subset C^{\infty}(M,\cM)$ the set of all smooth maps from $M$ to $\cM$ satisfying condition \eqref{eq:pullback} for $\mathfrak{Q}$. Note that if $\varphi\in\Map(\mathfrak{Q})$ then all maps in the same homotopy class as $\varphi$ also belong to $\Map(\mathfrak{Q})$. We will refer to the elements of $\Map(\mathfrak{Q})$ as \emph{scalar maps}. The differential of a scalar map is a real vector bundle map:
\begin{equation*}
\dd\varphi\colon TM \to T\cM\, ,
\end{equation*}

\noindent
from $TM$ to the \emph{real} tangent bundle $T\cM$ of $\cM$. We can consider $\dd\varphi$ as a map from $TM$ to the complexification $T_{\mathbb{C}}\cM = T\cM\otimes \mathbb{C}$ of $T\cM$ by composition with the canonical injection $T\cM\hookrightarrow T_{\mathbb{C}}\cM$. Splitting:
\begin{equation*}
T_{\mathbb{C}}\cM = T^{1,0}\cM \oplus T^{0,1}\cM\, ,
\end{equation*} 

\noindent
in the usual way by using $\cI$, we define:
\begin{equation*}
\dd\varphi^{1,0}\colon TM \to T^{1,0}\cM\, , \qquad \dd\varphi^{0,1}\colon TM \to T^{0,1}\cM\, , 
\end{equation*}

\noindent
respectively by projection onto the factors $T^{1,0}\cM$ and $T^{0,1}\cM$. As we will see later, these projections will be necessary for the global formulation of the Killing spinor equations of chiral supergravity. A choice of chiral triple will unambiguously define a bosonic sector of $\cN=1$ chiral supergravity and its associated Killing spinor equations. Let $\norm{\cdot}_{\cH}$ denote the norm induced by $\cH$ on $\cL$. The Hermitian metrics $\cG$ and $\cH$ induce a Hermitian metric on the vector bundle $T^{\ast}_{\mathbb{C}}\cM \otimes \cL$, whose norm we denote by $\norm{\cdot}_{\cH,\cG}$. 

\begin{definition}
The {\em scalar potential} associated to the chiral triple $\mathfrak{Q}$ is the smooth real-valued function $\Phi_k \in C^\infty(\cM)$ defined through:
\begin{equation*}
\label{ScalarPotential}
\Phi_k \eqdef \norm{\cD \cW}^2_{\cH,\cG}-k \norm{\cW}^2_{\cH}\, ,
\end{equation*}

\noindent
for a real constant $k > 0$, where $\cD$ denotes the Chern connection of $\cL_{\cH}$.
\end{definition}

\noindent
Given a scalar map $\varphi\colon M\to \cM$, we denote by $\Phi^{\varphi}_k \eqdef \Phi_k\circ \varphi$ the pull-back of $\Phi_k$.

% % % % % % % % % % % % % % % % % % % % % % % % % % % % % % % % % % % % % % 

\subsection{Bosonic sector and Killing spinor equations}

% % % % % % % % % % % % % % % % % % % % % % % % % % % % % % % % % % % % % %

We are ready to introduce the equations of motion and Killing spinor equations defining the bosonic sector of chiral supergravity associated to a particular chiral triple $\mathfrak{Q}$. The equations of motion follow from a variational principle involving a Lorentzian metric $g$ and a scalar map $\varphi\colon M\to \cM$.

\begin{definition}
The \emph{configuration pre-sheaf} $\mathrm{Conf}_{\mathfrak{Q}}$ of $\cN=1$ chiral supergravity associated to $\mathfrak{Q}$ is defined as the pre-sheaf of sets that assigns to every open set $U\subset M$ the following set:
	
\begin{equation*}
\mathrm{Conf}_{\mathfrak{Q}}(U) \eqdef \left\{ (g,\varphi) \,\, |\,\, g\in \mathrm{L}_0(U)\, , \,\, \varphi \in \Map_U(\mathfrak{Q})\right\}\, ,
\end{equation*}
	
\noindent
where $\mathrm{L}_0(U)$ denotes the space of Lorentzian metrics on $U\subset M$ such that $w^{-}_1(U,g) = 0$ and $\Map_U(\mathfrak{Q})$ denotes the set of smooth maps from $U$ to $\cM$ that satisfy condition \eqref{eq:pullback}.
\end{definition}

\noindent
Given a scalar map $\varphi\in \Map(\mathfrak{Q})$, we denote by $T\cM^{\varphi}$ the pull-back of $T\cM$ by $\varphi$. We equip $T\cM^{\varphi}$ with the bundle pull-back metric $\cG^{\varphi}$ of $\cG$. Likewise, we respectively denote by $\cW^{\varphi}$, $\cH^{\varphi}$ and $\cD^{\varphi}$ the bundle pull-backs of $\cW$, $\cH$ and $\cD$ by $\varphi$. In general, for any section $s$ of a bundle over $\cM$ we denote by $s^{\varphi} \eqdef s\circ \varphi$ the corresponding section of the pull-back bundle. We denote by $\norm{-}_{g,\cG}$ the norm induced by $g$ and $\cG^{\varphi}$ on $TM\otimes T\cM^{\varphi}$ and associated tensor powers. Inspired by the local formulation of standard chiral $\cN=1$ supergravity \cite{FreedmanProeyen,Ortin} we introduce the following definition. 

\begin{definition}
\label{def:BosSugraDef}
The Lagrange density $\mathfrak{Lag}\colon \Conf_{\mathfrak{Q}}(U)\to \Omega^4(U)$ of chiral $\cN=1$ ungauged supergravity associated to the scalar manifold $(\cM,\mathfrak{Q})$ is given by:
\begin{equation}
\label{eq:S}
\mathfrak{Lag}[g,\varphi]=\left[ \mathrm{R}_{g} - \norm{\dd \varphi}^2_{g,\cG} - \Phi^{\varphi}_k\right] \vol_g\, , \qquad (g,\varphi)\in \Conf_{\mathfrak{Q}}(U)\, ,
\end{equation}

\noindent
for every open subset $U$ of $M$, where $\vol_g$ is the Lorentzian volume form of $(M,g)$ and $\mathrm{R}_{g}$ is the scalar curvature of $g$.
\end{definition}

\noindent
The following Proposition follows from direct computation by using standard theory of variations, so we leave its proof to the reader.

\begin{prop}
The Euler-Lagrange equations associated to the Lagrangian \eqref{eq:S} is given by:
\begin{itemize}
\item The \emph{Einstein equations}:	
\begin{equation}
\label{eq:Einsteineqs}
\cE_{E}(g,\varphi) \eqdef \G(g) - \mathrm{T}(g,\varphi) = 0\, ,
\end{equation}
where the energy-momentum tensor $\mathrm{T}(g,\varphi)\in\Gamma(M,\odot^{2}T^{\ast}M)$ given by:
\begin{equation*}
\mathrm{T}(g,\varphi)  = \varphi^{\ast} \cG - \frac{1}{2} \left( \norm{\dd \varphi}^2_{g,\cG} + \, \Phi^{\varphi}_k \right) g\, .
\end{equation*} 
	
\noindent
and $\G(g) = \mathrm{Ric}^g - \frac{\mathrm{R}^g}{2} g$ is the Einstein tensor.

\

\item The \emph{scalar equations}:
\begin{equation}
\label{eq:Scalareqs}
\cE_{S}(g,\varphi) \eqdef \Tr_g \nabla \dd \varphi - \frac{1}{2}(\mathrm{grad}_{\cG} \Phi_k)^{\varphi}= 0\, ,
\end{equation}

\noindent
where $\nabla$ denotes the connection induced by the Levi Civita connections on $(M,g)$ and $(\cM,\cG)$.
\end{itemize}

\

\noindent
for pairs $(g,\varphi)\in\Conf_{\mathfrak{Q}}(M)$.
\end{prop}

\begin{remark}
We denote by $\Sol_{\mathfrak{Q}}\subset \Conf_{\mathfrak{Q}}$ the pre-sheaf of solutions of \eqref{eq:Einsteineqs} and \eqref{eq:Scalareqs}. 	
\end{remark}

\noindent
We proceed now to introduce the Killing spinor equations of chiral $\cN=1$ supergravity. First, we need the following result. 

\begin{prop}
\label{prop:Cwelldefined}
For every isomorphism $\Psi\colon \L_Q \xrightarrow{\simeq} \cL^{\varphi}$ of complex line bundles, there exists canonical isomorphisms of complex vector bundles:
\begin{equation*}
\mathfrak{T}_{\Psi}\colon \cL^{\varphi}\otimes S^{\pm}_{c} \xrightarrow{\simeq} S^{\pm}\, .
\end{equation*}  
\end{prop}

\begin{proof}
Using the identification $\Psi\colon \L_Q \xrightarrow{\simeq} \cL^{\varphi}$, such canonical isomorphism is given by:
\begin{equation*}
\mathfrak{T}_{\Psi}\colon\cL^{\varphi}\otimes S^{\pm}_{c}\xrightarrow{\simeq} S^{\pm}\, , [q,z]\otimes [q,\xi]\mapsto [q,z\,\xi]\, ,
\end{equation*}
	
\noindent
for all $[q,\xi]\in S^{\pm}_c = Q\times_{\tau^{\pm}_c}\Sigma^{\pm}_0$ and $[q,z]\in \cL^{\varphi} = Q\times_{l}\C$. This is well-defined since:
\begin{equation*}
[q\, [g,w], (w^{-2} z)\, (w\, g^{-1})\,\xi ] = [q\, [g,w], w^{-1} g^{-1} \, z\,\xi ] = [q,z\, \xi]\, ,
\end{equation*}
	
\noindent
where we have used that any representative of the class $[q,z]\in \cL^{\varphi}$ is of the form $(q [g,w], w^{-2}\, z)$ and any representative of the class $[q,\xi]\in S^{\pm}_c$ is of the form $(q [g,w], \tau^{\pm}_c([g,w]^{-1})\,\xi) = (q [g,w], w g^{-1}\,\xi)$ for an element $[g,w]\in \Spin^c_0(3,1)$.
\end{proof}

\noindent
Let $\varphi\in \Map(\mathfrak{Q})$ be a scalar map. For every superpotential $\cW$, scalar map $\varphi$, Chern connection $\cD$ and isomorphism $\Psi\colon \cL^{\varphi}\xrightarrow{\simeq} \L_Q$ we define the morphisms of real vector bundles:
\begin{equation*}
\mathfrak{C}^{\Psi}_{\cW,\varphi}\colon S^{\pm}\to T^{\ast}M\otimes S^{\pm}\, , \quad \mathfrak{C}^{\Psi,\cD}_{\cW,\varphi}\colon S^{\pm}\to \Lambda^{1,0} T^{\ast}\cM^{\varphi}\otimes S^{\mp}\, ,
\end{equation*}

\noindent
as follows:
\begin{equation*}
\mathfrak{C}^{\Psi}_{\cW,\varphi}(\epsilon)(v) \eqdef v\,\cdot\,\mathfrak{T}_{\Psi}\left(\cW^{\varphi} \otimes \mathfrak{c}(\epsilon)\right)\, , \qquad \mathfrak{C}^{\Psi,\cD}_{\cW,\varphi}(\epsilon) =  \mathfrak{T}_{\Psi} ( (\cD \cW)^{\varphi}\otimes \mathfrak{c}(\epsilon))\, , \quad \forall\, v\in \mathfrak{X}(M)\, ,\,\,\forall\, \epsilon\in \Gamma(S^{\pm}) 
\end{equation*}

\noindent
where in the definition of $\mathfrak{C}^{\Psi,\cD}_{\cW,\varphi}$ we have trivially extended $\mathfrak{T}_{\Psi}$ to sections of $\cL^{\varphi}\otimes S^{\pm}_{c}$ taking values on $\Lambda^{1,0} T^{\ast}\cM^{\varphi}$. Here the dot denotes Clifford multiplication $TM\otimes S\to S$.

\begin{remark}
Recall that the superscript $\varphi$ in $(\cD\cW)^{\varphi}$ denotes pull-back of $\cD\cW$ as a section of an abstract vector bundle instead of a one-form taking values on a complex line bundle. Hence, $(\cD^{1,0} \cW)^{\varphi}\in \Gamma(\Lambda^{1,0}T^{\ast}\cM^{\varphi} \otimes \cL^{\varphi})$, where $\Lambda^{1,0}T^{\ast}\cM^{\varphi}$ denotes the vector bundle over $M$ obtained from $\Lambda^{1,0}T^{\ast}\cM$ via bundle pull-back by $\varphi$.
\end{remark}

\begin{remark}
\label{remark:localiso}
For future reference, it is convenient to explicitly write the local form of the complex isomorphism $\mathfrak{T}_{\Psi}$ defined above. Let $E\colon U\to Q$ denote a local section of $Q$, with $U\subset M$ open. The local section $E$ canonically defines local frames  $u$, $\left\{ e_a\right\}_{a=1,2}$ and $\left\{ e^c_a\right\}_{a=1,2}$ of $\cP_{\cH}$, $S^{+}$ and $S^{-}_c$, respectively. A similar discussion applies if we consider instead $S^{-}$ and $S^{+}_c$. Furthermore, it can be seen that:
\begin{equation*}
\mathfrak{c}(e_a^c) = e_a\, , \qquad a = 1, 2\, .
\end{equation*}

\noindent
Let $l^{\varphi}$ denote the pull-back by $\varphi$ of a local holomorphic frame $l$ of $\cL$. Define $H^{\varphi}_l \eqdef \cH^{\varphi}(l^{\varphi},l^{\varphi})$. The local frame $E\colon U\to Q$ can be chosen such that:   
\begin{equation*}
u^{\varphi} = (H^{\varphi}_l)^{-\frac{1}{2}}\, l^{\varphi}\, ,
\end{equation*}

\noindent
where we are considering $u$ as a unitary section of $(\cL,\cH)$. The isomorphism of complex vector bundles $\mathfrak{T}_{\Psi}$ constructed in Proposition \ref{prop:Cwelldefined} can be evaluated at the homogeneous element $l^{\varphi}\otimes e^c_a$, yielding:
\begin{equation*}
\mathfrak{T}_{\Psi}(l^{\varphi}\otimes e^c_a) = (H^{\varphi}_l)^{\frac{1}{2}}\, \mathfrak{T}_{\Psi}(u^{\varphi}\otimes e^c_a) = (H^{\varphi}_l)^{\frac{1}{2}}\, e_a\, .
\end{equation*}

\noindent 
Any section $\eta \in \Gamma(S^{-}_c)$ can be locally written as:
\begin{equation*}
\eta = \eta^a\, e^c_a\, ,
\end{equation*}

\noindent
for some local complex valued smooth functions $\eta^a$. Extending $\mathfrak{T}_{\Psi}$ complex linearly we obtain:
\begin{equation*}
\mathfrak{T}_{\Psi}(l^{\varphi}\otimes \eta) = (H^{\varphi}_l)^{\frac{1}{2}}\,\mathfrak{T}_{\Psi}(\eta^a\, u^{\varphi}\otimes e^c_a) = (H^{\varphi}_l)^{\frac{1}{2}}\,\eta^a\, e_a\, ,
\end{equation*}

\noindent
which gives the local expression of $\mathfrak{T}_{\Psi}$. Let now $\epsilon \in \Gamma(S^{+})$ be a section of $S^{+}$, which we write as $\epsilon = \epsilon^a e_a$. We have:
\begin{equation*}
\mathfrak{T}_{\Psi}(l^{\varphi}\otimes \mathfrak{c}(\epsilon)) = (H^{\varphi}_l)^{\frac{1}{2}}\,\mathfrak{T}_{\Psi}(u^{\varphi}\otimes  \mathfrak{c}(\epsilon^a e_a)) = (H^{\varphi}_l)^{\frac{1}{2}}\,\bar{\epsilon}^a\, e_a\, ,
\end{equation*}

\noindent
This expression will be used in the proof of Proposition \ref{prop:localequivsugra}.
\end{remark}

\

\noindent
Given a chiral triple $\mathfrak{Q}$, a choice of isomorphism $\Psi\colon \L_Q \xrightarrow{\simeq} \cL^{\varphi}$ and a Lorentzian metric $g$ we have a canonical choice $\nabla^{g}_{\Psi^{\ast}\cA^{\varphi}}$ of connection on the spinor bundle $S$, which is constructed as follows. The choice of isomorphism $\Psi\colon \L_Q\to \cL^{\varphi}$ of complex line bundles induces a canonical isomorphism of principal $\U(1)$ bundles (which we denote by the same symbol) $\Psi\colon \P_Q\to \cP^{\varphi}_{\cH}$. We use this isomorphism to take the pull-back of $\cA^\varphi$ and define a connection $\Psi^{\ast}\cA^{\varphi}$ on $\P_Q$. The connection $\Psi^{\ast}\cA^{\varphi}$ on $\P_Q$ together with the Levi-Civita connection of $(M,g)$ yield, in the usual way through tensor product and lifting, a unique associated connection $\nabla^{g}_{\Psi^{\ast}\cA^{\varphi}}$ on $S$. For ease of notation we denote $\nabla^{g}_{\Psi^{\ast}\cA^{\varphi}}$ simply by $\nabla^\varphi$. Recall that we defined $\cA$ as the unique unitary connection on $\cP_{\cH}$ to which the Chern connection $\cD$ on $\cL_{\cH}$ is associated in the standard way, and that the superscript $\varphi$ denotes pull-back by $\varphi$.

\begin{definition}
We define the \emph{extended configuration} pre-sheaf $\Conf_{\mathfrak{Q},E}$ as the pre-sheaf of sets which assigns, to every open set $U\subset M$, the following set:
\begin{equation*}
\Conf_{\mathfrak{Q},E} \eqdef \left\{ (g,\varphi,\epsilon, \Psi)\in \Conf_{\mathfrak{Q}}(U)\times \Gamma(S^{-}\vert_{U})\times \mathrm{Iso}(\mathrm{L}_{Q}|_U,\cL^{\varphi}|_U)\right\}\, ,
\end{equation*}
	
\noindent
where $\Psi\in \mathrm{Iso}(\mathrm{L}_{Q}|_U,\cL^{\varphi}|_U)$ is a $C^{\infty}$-isomorphism of complex line bundles over $U$. 
\end{definition}

\begin{definition}
\label{def:KSE}
Let $\mathfrak{Q}$ be a chiral triple on $(M,\cM)$. The \emph{Killing spinor equations} (KSE) associated to $\mathfrak{Q}$ are defined as follows:
\begin{equation}
\label{eq:KSE}
\nabla^\varphi \epsilon = \mathfrak{C}^{\Psi}_{\cW,\varphi}(\epsilon)\, , \qquad (\dd \varphi^{0,1})^{\flat}\cdot \epsilon = \mathfrak{C}^{\Psi,\cD}_{\cW,\varphi}(\epsilon)\, ,
\end{equation}
	
\noindent
for tuples $(g,\varphi,\epsilon,\Psi)\in \Conf_{\mathfrak{Q},E}$. Here the dot denotes Clifford multiplication $T^{\ast}M\otimes S\to S$ and $\flat$ the musical isomorphism $ T^{0,1}\cM^{\varphi} \simeq \Lambda^{1,0} T\cM^{\varphi}$ associated to the Hermitian metric $\cG^{\varphi}$. A spinor $\epsilon \in \Gamma(S^{-})$ satisfying Equations \eqref{eq:KSE} is called a \emph{supersymmetry spinor} or \emph{supersymmetry generator}.
\end{definition}

\begin{remark}
In the definition of Killing spinor equations we have allowed only for supersymmetry spinors $\epsilon\in \Gamma(S^{-})$ of negative chirality. In principle, either choice of chirality is allowed. The right choice exclusively depends on the conventions used: once the conventions have been fixed only one of the two chiralities, in our conventions negative chirality, will be compatible with the supersymmetric structure of the theory. This is one of the reasons why this supergravity theory is called \emph{chiral}.
\end{remark}

\begin{definition}
We define the pre-sheaf of \emph{supersymmetric configurations} as the pre-sheaf of sets $\Conf_{\mathfrak{Q},S}$ which to every open subset $U\subset M$ assigns the following set:
\begin{equation*}
\Conf_{\mathfrak{Q},S}(U) \eqdef \left\{ (g,\varphi,\epsilon, \Psi)\in \Conf_{\mathfrak{Q},E}(U) \quad \mid \quad \mathrm{KSE\,\, are\,\, satisfied}\right\}
\end{equation*}

\noindent
We say that $(g,\varphi)\in \Sol_{\mathfrak{Q}}(M)$ is a \emph{supersymmetric solution} if there exists at least one spinor $\epsilon\in\Gamma(S^{-})$ and isomorphism $\Psi\in \mathrm{Iso}(\mathrm{L}_{Q},\cL^{\varphi})$ such that $(g,\varphi,\epsilon,\Psi)\in \Conf_{\mathfrak{Q},S}(M)$. We denote the pre-sheaf of supersymmetric solutions by $\Sol_{\mathfrak{Q},S}$.
\end{definition}

\begin{remark}
Notice that we have not used the condition $w_1^{-}(M,g) = 0$. Indeed, this is not required to formulate neither the equations of motion nor the Killing spinor equations. However, it is crucial in order to construct the full supersymmetric Lagrangian, since it requires the use the sesquilinear pairing $\cB$ introduced in Equation \eqref{eq:bilinears}, which is only invariant under $\Spin^c_0(3,1)$, compare Proposition \ref{prop:Spinc0}. 
\end{remark}

\noindent
The following Proposition shows that chiral $\cN=1$ supergravity and its associated Killing spinor equations, as introduced in Definitions \ref{def:BosSugraDef} and \ref{def:KSE}, reproduce the well-known local formulation of four-dimensional ungauged $\cN=1$ supergravity coupled to chiral multiplets, which the reader can find explained in detail in References \cite{Ortin,FreedmanProeyen}. This result summarizes the underlying motivation for the structures and definitions introduced so far, showing in addition that we have constructed an allowed global extension of the local formulas defining ungauged $\cN=1$ chiral supergravity in four Lorentzian dimensions.

\begin{prop}
\label{prop:localequivsugra}
Let $\mathfrak{Q}$ be a chiral triple on $(M,\cM)$, $(U,x^{\mu})$ a local coordinate chart on $M$ and $(V, w^i)$ a holomorphic coordinate chart on $\cM$. The restriction of the bosonic Lagrangian
\begin{equation*}
\mathfrak{Lag}\colon \Conf_{\mathfrak{Q}}(U) \to C^{\infty}(U)
\end{equation*}

\noindent
and Killing spinor equations to the chart $U$ are respectively given by:
\begin{equation*}
\mathfrak{Lag}[g,\varphi]\vert_U =\left[ \mathrm{R}_{g} - 2\,\cG_{i j}(z,\bar{z})\, \partial_{\mu} z^i \partial^{\mu} \bar{z}^{j} -\Phi_k(z,\bar{z})\right] \vol_g\, ,
\end{equation*}

\noindent
and:
\begin{equation*}
\nabla^\varphi_{\mu} \epsilon = e^{\mathcal{K}(z,\bar{z})/2} W(z)\gamma_{\mu} \bar{\epsilon}\, , \qquad \cG_{i k} ( \partial_{\mu} \bar{z}^{k}) \gamma^{\mu} \epsilon = e^{\mathcal{K}(z,\bar{z})/2} \cD_i W(z) \bar{\epsilon} \, ,
\end{equation*}

\noindent
where $\left\{ z^i \right\}$ denotes the local coordinate expression of $\varphi$, $\cG_{i j} = \cG(\frac{\partial}{\partial w^i}, \frac{\partial}{\partial \bar{w}^{j}}) =  \partial_{w^i} \partial_{\bar{w}^{j}} \mathcal{K}$ denotes the restriction of $\cG$ to $U$, and $\mathcal{K}$ is an associated K\"ahler potential. Here we assume $\varphi(U)\subset V$. Hence, Definition \ref{def:BosSugraDef} and Definition \ref{def:KSE} locally reproduce the standard local bosonic sector and Killing spinor equations of $\cN=1$ chiral supergravity.
\end{prop}

\begin{proof}
The complexified tangent bundle $T_{\mathbb{C}}\cM$ is locally the complex span of $\left\{\partial_{w^i}, \partial_{\bar{w}^i}\right\}$, whereas the tangent bundle $TM$ is locally the real span of $\left\{\partial_{x^\mu}\right\}$. We have:
\begin{equation*}
\dd\varphi(\partial_{\mu}) = \partial_{\mu} z^i\, \partial_{w^i} + \partial_{\mu} \bar{z}^i\, \partial_{\bar{w}^i}\, ,
\end{equation*}

\noindent
where for simplicity we denote by:
\begin{equation*}
z^i \eqdef w^i \circ \varphi\colon U \subset\R^4 \to \C\, ,
\end{equation*}

\noindent
the local components of $\varphi$. Hence:
\begin{equation*}
(\dd\varphi)^{1,0}(\partial_{\mu}) = \partial_{\mu} z^i\, \partial_{w^i}\, , \qquad (\dd\varphi)^{0,1}(\partial_{\mu}) = \partial_{\mu} \bar{z}^i\, \partial_{\bar{w}^i}\, .
\end{equation*}

\noindent
Explicit computation gives:
\begin{equation*}
\norm{\dd \varphi}^2_{g,\cG}|_{U} = 2\, g^{\mu\nu}\,\partial_{\mu} z^i \partial_{\nu}\bar{z}^{k} \,\cG_{i k}(z,\bar{z})\, ,
\end{equation*}

\noindent
which corresponds with the standard sigma model appearing in the local formulation of chiral supergravity. Let $ l\colon V\to \cL$ be a holomorphic local trivializing section of $\cL$. We set $K_{l}(w,\bar{w}) \eqdef \cH(l,l)\in C^{\infty}(V)$ and write $\cW\vert_V = W(w)\, l$ for some local holomorphic complex function $W(w)$ on $V$. We have:
\begin{equation*}
\left.\norm{\cW}^2_{\cH}\right|_V = K_{l}(w,\bar{w})\, |W(w)|^{2} 
\end{equation*}   

\noindent
Furthermore, we write $\cD l = \mathcal{Q}\otimes l$ for the Chern connection $\cD$, where $\mathcal{Q}\in\Omega^{1,0}(X)$ is a local $(1,0)$ form, which in local coordinates $\left\{ w^i\right\}$ is given by:
\begin{equation*}
\mathcal{Q} = \partial \log\, K_l(w,\bar{w}) = \partial_i \log\, K_l(w,\bar{w})\, \dd w^i\, .
\end{equation*}

\noindent
The associated Chern curvature is locally given by $\Theta(\cD) = - \partial \bar{\partial}\, \log\, K_l(w,\bar{w})$. By the definition of chiral triple, the two-form:
\begin{equation*}
\frac{i}{2\pi}\,\partial \bar{\partial}\, \log\, K_l(w,\bar{w})\, ,
\end{equation*}

\noindent
is the K\"ahler form associated to $\cG$ and $\cI$. Hence, we can write:
\begin{equation*}
\frac{i}{2\pi}\, \partial \bar{\partial}\, \log\, K_l(w,\bar{w}) = \frac{i}{2\pi}\,  \partial \bar{\partial}\, \mathcal{K}(w,\bar{w})\, ,
\end{equation*}
for some local K\"ahler potential $\mathcal{K}(w,\bar{w})\in C^{\infty}(V)$. Integrating, we can express $K_l(w,\bar{w})$ in terms of $\mathcal{K}(w,\bar{w})$ as follows:
\begin{equation*}
K_l(w,\bar{w}) = e^{\mathcal{K}(w,\bar{w})}\, ,
\end{equation*}

\noindent
uniquely up to the usual local K\"ahler transformations of $\mathcal{K}(w,\bar{w})$. With these provisos in mind we obtain:
\begin{equation*}
\cD\cW\vert_V = \cD_i\cW\, \dd w^i = (\partial_i + \partial_i\, \mathcal{K}(w,\bar{w})) W(w)\, dw^i\otimes l\, ,
\end{equation*}

\noindent
whence: 
\begin{equation*}
\Phi_k = e^{\mathcal{K}(w,\bar{w})} \left[\cG^{i j}(w,\bar{w}) \cD_i W(w) \bar{\cD}_{j}\bar{W}(\bar{w}) - k\, W(w)\bar{W}(\bar{w})\right]\, ,
\end{equation*}

\noindent
which, for $k=3$, corresponds with the standard potential of local chiral supergravity. Since by definition $\nabla^{\varphi}$ is constructed by lifting the product connection of the Levi-Civita connection with the pull-back by $\varphi$ of $\cA$, we can write locally:
\begin{equation*}
\nabla^{\varphi}_{\mu} \epsilon = \nabla^g_{\mu} \epsilon + \frac{1}{2} \cA^{\varphi}_{\mu}(z,\bar{z}) \epsilon\, ,
\end{equation*}

\noindent
where $\nabla^g_{\mu} $ is the local lift to the spinor bundle of the Levi-Civita connection associated to $g$ and:
\begin{equation*}
\cA^{\varphi}_{\mu} = \frac{1}{2} (\partial_{\mu} z^i \partial_i \mathcal{K}(z,\bar{z}) - \partial_{\mu} \bar{z}^{i} \bar{\partial}_{i} \mathcal{K}(z,\bar{z}))\, .
\end{equation*}

\noindent
Note that the formula above gives the local form of the Chern connection with respect to a unitary frame instead of a holomorphic frame. This gives the local form of the standard \emph{$\U(1)$-coupled} covariant derivative appearing in the supersymmetry transformation of the gravitino of chiral supergravity. Let now $E\colon U\to Q$ denote a local section of $Q$, with $U\subset M$ open subset. As explained in Remark \ref{remark:localiso}, the local section $E$ canonically defines local frames $u$, $\left\{ e_a\right\}$ and $\left\{ e^c_a\right\}$ of $S^{-}$ and $S^{+}_c$, respectively. We choose $E$ such that the following is satisfied:
\begin{equation*}
u = \frac{l}{\cH(l,l)^{\frac{1}{2}}}
\end{equation*}

\noindent
Let $l^{\varphi}\colon U \to \cL^{\varphi}$ be the pull-back by $\varphi$ of the local holomorphic frame $l$ over $V\subset \cM$. Upon the use of Proposition \ref{prop:Cwelldefined} and Remark \ref{remark:localiso} we obtain:
\begin{equation*}
\mathfrak{T}_{\Psi}(\cW^{\varphi} \otimes \mathfrak{c}(\epsilon))\vert_{U} = W(z)\, \mathfrak{T}_{\Psi}(l^{\varphi}\otimes \mathfrak{c}(\epsilon^a e_a)) = W(z)\, \bar{\epsilon}^a\, \mathfrak{T}_{\Psi}(l^{\varphi}\otimes e^c_a) = W(z)\, e^{\mathcal{K}(z,\bar{z})/2} \bar{\epsilon}^a\, e_a\, ,
\end{equation*}

\noindent
compare Remark \ref{remark:localiso}. Similarly, we have:
\begin{eqnarray*}
\mathfrak{T}_{\Psi}((\cD \cW)^{\varphi}\otimes \mathfrak{c}(\epsilon))\vert_{U} = \cD_i W(z)\, \mathfrak{T}_{\Psi}(dz^{i}\otimes l^{\varphi}\otimes \mathfrak{c}(\epsilon^a e_a)) = \bar{\epsilon}^a\,\cD_i W(z)\,dz^{i}\otimes\mathfrak{T}_{\Psi}(l^{\varphi}\otimes e^c_a)\\ = e^{\mathcal{K}(z,\bar{z})/2}\cD_i W(z)\, \bar{\epsilon}^a\, dz^{i}\otimes e_a\, .
\end{eqnarray*}

\noindent
Using now that:
\begin{equation*}
(\dd\varphi^{0,1})^{\flat}\cdot \epsilon\vert_{U} = \cG_{i k} ( \partial_{\mu} \bar{z}^{k}) \gamma^{\mu} \epsilon\, ,
\end{equation*}

\noindent
we conclude.
\end{proof}

\noindent
%One of the key amd more mysterious ingredients of chiral supergravity is its potential $\Phi_k$, which is given in a precise way in terms of the chiral triple $\mathfrak{Q}$ by equation \eqref{ScalarPotential}. For instance, $\Phi_k$ and the structure and type of its critical points is of utmost importance for the phenomenological applications of four-dimensional $\cN=1$ chiral supergravity, and a vast amount of physics literature has been devoted to the study of $\Phi$ from various points of view. 

% % % % % % % % % % % % % % % % % % % % % % % % % % % % % % % % % % % % % % 
% % % % % % % % % % % % % % % % % % % % % % % % % % % % % % % % % % % % % %

\subsection{Chiral supergravity on a spin four-manifold}
\label{sec:spinfourmanifold}

% % % % % % % % % % % % % % % % % % % % % % % % % % % % % % % % % % % % % % 
% % % % % % % % % % % % % % % % % % % % % % % % % % % % % % % % % % % % % %

Under the assumption that the $\Spin^c_0(3,1)$ structure $Q$ considered so far can be reduced to $\Spin_0(3,1)$ structure $Q_0$, we will show how now to formulate the Killing spinor equations in terms of the Dirac spinor bundle $S_0$ associated to $Q_0$ and a holomorphic line bundle $\cR$ over $\cM$. In the formalism of the previous subsection, the latter corresponds to a square root of the line bundle $\cL$. As a consequence, the pull-back $\cR^{\varphi}$ is the square root of $\cL^{\varphi}$. The existence of such square root follows from the relation $0 = w_2(M) = c_1(\L_Q)\mod 2$.

Hence, let $Q_0$ be a $\Spin_0(3,1)$ structure on $(M,g)$ and let $S_0$ denote the associated Dirac complex spinor bundle, which splits as $S_0 = S^{+}_0 \oplus S^{-}_0$ in terms of the chiral spinor bundles $S^{\pm}_0$. We denote by $(-,-)_0$ the $\Spin_0(3,1)$-invariant Hermitian product on $S_0$. Contrary to what happened in the $\Spin^c_0(3,1)$ case, the spinor bundle $S_0$ admits now a real structure, that is, a antilinear and involutive isomorphism $\mathfrak{c}_0\colon S^{\pm}_0\to S^{\mp}_0$. In this set up, the notion of chiral triple can be simplified in a way which we proceed to discuss. For the benefit of the reader we recall that, using the notion introduced in Section \ref{sec:CliffordAlgebras}, $S^0$, $S_0^{+}$ and $S^{-}_0$ are associated with the $\Spin_0(3,1)$ modules $\Sigma$, $\Sigma^{0,1}$ and $\Sigma^{1,0}$, respectively. 

\begin{definition}
Let $\cM$ be a complex manifold of K\"ahler type. A {\em chiral triple} defined on $\cM$ is a triple $(\cR, \cH_{\frac{1}{2}},\cW)$, where $(\cR,\cH_{\frac{1}{2}})$ is a negative (in the sense of Definition \ref{def:positivecL}) Hermitian holomorphic line bundle, with Hermitian structure $\cH_{\frac{1}{2}}$, and $\cW$ is a holomorphic section of $\cL \eqdef \cR\otimes \cR$.
\end{definition}

\begin{remark}
Note that the previous definition does not impose any conditions on the admissible scalar maps $\varphi\colon M\to \cM$ and, in fact, it is independent of $M$ as long as the latter $M$ is spin.
\end{remark}

\noindent
We equip $\cL$ with the Hermitian form $\cH \eqdef \cH_{\frac{1}{2}}\otimes \cH_{\frac{1}{2}}$ induced by $\cH_{\frac{1}{2}}$. Given any scalar map $\varphi\colon M\to \cM$, we define the \emph{(chiral) supergravity spinor bundle associated to $\varphi$} as follows:
\begin{equation*}
S \eqdef S_0\otimes \cR^{\varphi}\, , \qquad S^{\pm} \eqdef S^{\pm}_0\otimes \cR^{\varphi}\, ,
\end{equation*}

\noindent
where $\cR^{\varphi}$ denotes the pull-back of $\cR$ by $\varphi$, which we endow with the pull-back Hermitian form $\cH^{\varphi}_{\frac{1}{2}}$. Likewise, we define:
\begin{equation*}
S_{c} \eqdef S_0\otimes (\cR^{\varphi})^{-1}\, , \qquad S^{\pm}_{c} \eqdef S^{\pm}_0\otimes (\cR^{\varphi})^{-1}\, ,
\end{equation*}

\noindent
Here we have used the symbols $S^{\pm}$ and $S^{\pm}_c$ to denote the spinor bundles that correspond to those denoted by same symbols in the general $\Spin^c_0(3,1)$ case. Using the Hermitian form $(-,-)_0$ present on $S_0$ together with the Hermitian form $\cH_{\frac{1}{2}}$ on $\cR$, we define the Hermitian form $(-,-) = (-,-)_0\otimes \cH^{\varphi}_{\frac{1}{2}}$ on $S$. In addition, we define an extension of $\mathfrak{c}_0\colon S^{\pm}_0\to S^{\mp}_{0}$ to the supergravity spinor bundles $S^{\pm}$ and $S^{\pm}_c$ as follows:
\begin{equation*}
\mathfrak{c}\colon S^{\pm}\to S^{\mp}_{c}\, , \qquad \xi_{0}\otimes l^{\varphi}\mapsto \mathfrak{c}_{0}(\xi_0)\otimes \cH^{\varphi}_{\frac{1}{2}}(-,l^{\varphi})\, ,
\end{equation*}

\noindent
acting on homogeneous sections of $S^{\pm} = S^{\pm}_0\otimes \cR^{\varphi}$. With these definitions, we reproduce the action of $\mathfrak{c}$ as defined in the $\Spin^c_0(3,1)$ case, with the advantage that we can now isolate how it acts on each of the factors appearing in the definition of $S^+$ and $S^-$. The Killing spinor equations are now given formally by the same expression as in the general $\Spin^c_0(3,1)$ case, that is:
\begin{equation*}
\nabla^\varphi \epsilon = \mathfrak{C}^{\Psi}_{\cW,\varphi}(\epsilon)\, , \qquad (\dd \varphi^{0,1})^{\flat}\cdot \epsilon = \mathfrak{C}^{\Psi,\cD}_{\cW,\varphi}(\epsilon)\, .
\end{equation*}

\noindent
Note that in this situation the connection $\nabla^{\varphi}$ is an honest tensor product connection on the complex spinor bundle $S_0$ tensorized with a complex a line bundle.

\begin{ep}
Let $(M,g)$ be a four-dimensional Lorentzian spin four manifold and let us take $\cM$ to be the complex projective line $\mathbb{P}^1$. We have that $T\mathbb{P}^1\simeq \cO(2)$ is positive with respect to the Hermitian structure induced by the standard Fubini-Study metric $\cH$ and thus $(T\mathbb{P}^1,\cH)$ satisfies the positivity notion given in Definition \ref{def:positivecL}. Therefore, its dual bundle, that is, the canonical bundle $ K_{\mathbb{P}^1}$ is negative and furthermore has Chern number $(-2)$, whence it admits a (unique) holomorphic square root $ K^{1/2}_{\mathbb{P}^1}$. A series of chiral triples parametrized by a natural number $n\in \mathbb{N}_{\ast}$ is then given by $((K^{1/2}_{\mathbb{P}^1})^n ,\cH^n_{\frac{1}{2}}, \cW = 0)$. We are forced to take $\cW$ to be the zero-section since a negative line bundle over a closed complex manifold is well-known to not admit any non-zero holomorphic sections. In this situation, it is sometimes physically admissible to generalize the notion of superpotential and allow for meromorphic sections. 
\end{ep}

\begin{ep}
Let $(M,g)$ be a four-dimensional Lorentzian spin four manifold and let $\cM_l$ be a hyperbolic Riemann surface of genus $\ell$. We have $\deg(T^{1,0}\cM_l) = 2-2\ell < 0$, and thus $T^{1,0}\cM_l$ is a negative holomorphic line bundle with respect to any Hermitian metric $\cH_{\frac{1}{2}}$ of constant curvature $(-1)$. A family of chiral triples parametrized by a natural number $n\in \mathbb{N}_{\ast}$ is then given by $(\cR = (T\cM_{\ell})^n ,\cH^n_{\frac{1}{2}}, \cW)$, where $\cW$ is any holomorphic section of $(T\cM_{\ell})^{2n}$. 
\end{ep}

\noindent
We are now in disposition to show that every (possibly non-compact) complex manifold $\cM$ admitting integral K\"ahler forms can be endowed with a chiral triple and thus can be considered as the target space of chiral $\cN=1$ supergravity on a spin manifold $M$. Let $(\cM,\omega)$ be a K\"ahler-Hodge manifold, which we define as a complex manifold equipped with an integral K\"ahler form. Then, using a classical theorem by Weil \cite{Weil}, there exists a complex Hermitian line bundle $(L,\cH,D)$ equipped with a unitary connection $D$ such that:
\begin{equation*}
\omega = - \frac{i}{2\pi} \Theta(D)\, ,
\end{equation*}

\noindent
where $\Theta(D)$ denotes the curvature of $D$. Since $\omega$ is of type $(1,1)$ the previous equation implies $\Theta(D)^{0,2} = 0$ whence $D^{0,1}$ defines a holomorphic structure on $L$. Denote by $\cL$ the corresponding holomorphic line bundle. Then $D$ is the Chern connection of $(\cL,\cH)$, which becomes a negative Hermitian holomorphic line bundle, as required in order to define a chiral triple on $(M,\cM)$. Hence, if $(M,g)$ is spin, $\cM$ can be considered as the target space of the non-linear sigma model of chiral $\cN=1$ supergravity.

\begin{remark}
The argument given above may not work if $(M,g)$ is not spin. If $(M,g)$ is not spin, we have to prove that there exists at least one integral K\"ahler two-form on $\cM$ whose associated negative Hermitian holomorphic line bundle $(\cL,\cH)$ is isomorphic via pull-back by some map $\varphi$ to the determinant line bundle of some $\Spin^c_0(3,1)$ structure $Q$ on $(M,g)$. This may not be possible, as the following example shows. 
\end{remark} 

\begin{ep}
\label{ep:spinStein}
Take $(M,g)$ to be a non-spin Lorentzian manifold admitting $\Spin_0^c(3,1)$ structures (for instance, $\PP^2$ minus a point $p$), and let $\cM$ be any non-compact Riemann surface. Recall that every holomorphic line bundle $\cL$ over $\cM$ is holomorphically trivial by \cite[Theorem 30.3]{Forster} and, on the other hand, by \cite[Corollary 26.8]{Forster} every such $\cM$ is Stein and thus admits global subharmonic functions. Hence, every holomorphic line bundle $\cL$ over $\cM$ admits a Hermitian metric $\cH$ of negative Chern-curvature. We conclude that the pull-back of any holomorphic line bundle $\cL$, in particular any negative line bundle $(\cL,\cH)$, is topologically trivial. Since $(M,g)$ is not spin, the trivial line bundle can never be isomorphic to the determinant line bundle of a $\Spin_0^c(3,1)$ structure on $(M,g)$. Hence, even though $\cM$ admits negative line bundles, the pair $(M,\cM)$ with $M$ as above does not admit any chiral triple. In particular, this means that the scalar manifold of a supergravity theory on a non-spin Lorentzian manifold cannot be an open Riemann surface.
\end{ep}

% % % % % % % % % % % % % % % % % % % % % % % % % % % % % % % % % % % % % % 
% % % % % % % % % % % % % % % % % % % % % % % % % % % % % % % % % % % % % %

\subsection{Trivial scalar manifold}
\label{sec:Trivialscalarmanifold}

% % % % % % % % % % % % % % % % % % % % % % % % % % % % % % % % % % % % % % 
% % % % % % % % % % % % % % % % % % % % % % % % % % % % % % % % % % % % % %

As an example of the general formulation introduced in the previous sections, we consider now the simplest case of chiral supergravity, which goes under the name of \emph{pure} (AdS) $\cN=1$ supergravity in the physics literature. Despite being the simplest case of chiral $\cN=1$ supergravity, the associated Killing spinor equations pose an interesting problem involving \emph{generalized} Killing spinors \cite{FriedrichKim,FriedrichKimII} of a particular type, which we describe in this section. Pure (AdS) $\cN=1$ supergravity is defined as the unique $\cN=1$ supergravity not coupled to any matter content, whence exclusively containing the gravitational supermultiplet. Accordingly, we take the scalar manifold $\cM = \left\{ p\right\}$ to be a point. Let $\mathfrak{Q} = (\cL_{\cH},Q,\cW)$ be a chiral triple over $(M,\left\{ p\right\})$. The fact that $\cM$ is a point implies that $\cL_\cH$ is holomorphically trivial and can be identified with the one-dimensional Hermitian complex vector space $(\mathbb{C},H)$, where $H$ denotes the standard Hermitian form on $\C$. Furthermore, the superpotential $\cW$ becomes a complex number which we denote by $w$. In addition, the scalar map $\varphi$ is necessarily constant and the pull-back of $\cL_{\cH}$ is the trivial Hermitian line bundle over $M$, which we denote again by $(\mathbb{C},H)$. Since, by definition of chiral triple, we must have $\mathrm{P}_Q \simeq \cP^{\varphi}_{\cH}$, we conclude that $\L_Q$ is trivial and $\cD^{\varphi}$ is the trivial connection, which in turn implies that the determinant bundles of $S^+$ and $S^{-}$ are trivial, compare Proposition \ref{prop:detS}. In particular, $w_2(\L_Q) = 0$ whence $w_2(M) = 0$, and the $\Spin^c_0(3,1)$ structure $Q$ reduces to a $\Spin_0(3,1)$ structure $Q_0$. Hence, we consider that $S = S^+ \oplus S^{-}$ is the complex spinor bundle associated to $Q_0$ in the usual manner. With these provisos in mind, the scalar potential $\Phi = -k \vert w\vert^2$ becomes a non-positive constant and hence the lagrangian of the theory reduces to the Hilbert-Einstein Lagrangian coupled to a non-positive cosmological constant, that is:
\begin{equation*}
\mathfrak{Lag}[g]= \mathrm{R}_{g} + k \vert w\vert^2\, ,
\end{equation*}

\noindent
where $w\in \mathbb{C}$ is denotes the superpotential and $k$ is a positive real constant. The equations of motion associated to the previous functional read:

\begin{equation*}
\mathrm{Ric}(g) = - \frac{k}{2} \vert w\vert^2 g \, ,
\end{equation*}

\noindent
which are the standard Einstein equations coupled to a non-positive cosmological constant. The Killing spinor equations in turn reduce to:
\begin{equation}
\label{eq:KSEM0}
\nabla^g_v \epsilon = w\, v\cdot \mathfrak{c}(\epsilon)\, ,\qquad \forall v\in \mathfrak{X}(M)\, ,
\end{equation}

\noindent
where $\nabla^g$ denotes the lift of the Levi-Civita connection to the spinor bundle. Therefore, the set $\Sol_{S}(M)$ of supersymmetric solutions on $M$ consists of pairs $(g,\epsilon)$, with $g$ a Lorentzian metric and $\epsilon$ a chiral spinor, such that:
\begin{equation*}
\Sol_{S}(M) = \left\{ (g,\epsilon)\,\, \vert \,\, \mathrm{Ric}(g) =  - \frac{k}{2} \vert w\vert^2\, g\, , \,\, \,\, \nabla^g_v \epsilon = w\, v\cdot \mathfrak{c}(\epsilon) \, , \,\, \forall v\in \mathfrak{X}(M) \right\}
\end{equation*}

\noindent
It is important to point out that Equation \eqref{eq:KSEM0} does not correspond to a \emph{standard} Killing spinor equation \cite{CahenGuttLemaireSpindel,Bar,Baum,BaumII} (for neither a real nor an imaginary Killing spinor) even if $w$ is real. This is due to the complex-conjugate bundle map $\mathfrak{c}$ appearing in \eqref{eq:KSEM0}. In order to see this explicitly, we define:
\begin{equation*}
\epsilon_{1} \eqdef \epsilon\in \Gamma(S^{-})\, , \qquad \epsilon_{2} \eqdef \mathfrak{c}(\epsilon)\in \Gamma(S^{+})\, .
\end{equation*}

\noindent
From \eqref{eq:KSEM0} we deduce that the spinors $\epsilon_{1}$ and $\epsilon_{2}$ satisfy:
\begin{equation}
\label{eq:KSEM0II}
\nabla^g_v \epsilon_1 = w\, v\cdot \epsilon_2\, , \qquad \nabla^g_v \epsilon_2 = \bar{w}\, v\cdot \epsilon_1\, .
\end{equation}

\noindent
Crucially, the second equation above involves the complex conjugate $\bar{w}$ of $w$ instead of $w$. This in turn implies that $|w|^2$, instead of $w^2$, appears in the Einstein constant of the corresponding integrability condition, which allows for $w$ to be any complex number instead of only real or purely imaginary, as it happens in the standard theory of Killing spinors \cite{CahenGuttLemaireSpindel,Bar,Baum,BaumII}. 

\noindent
Defining now the following complex endomorphism of $S$:
\begin{equation*}
T_w\colon \Omega^{0}(S) \to \Omega^1(S) \, , \qquad T_w(\epsilon_1\oplus \epsilon_{2})(v) = v\cdot (w\, \epsilon_2\oplus \bar{w}\, \epsilon_{1})\, ,
\end{equation*}

\noindent
we can rewrite the Killing spinor equations \eqref{eq:KSEM0II} as a particular case of a \emph{generalized} Killing spinor equation:
\begin{equation}
\label{eq:KSEM0T}
\nabla^{g}\eta = T_{w}(\eta)\, ,
\end{equation}

\noindent
where $\eta = \epsilon_{1}\oplus \epsilon_{2}\in \Gamma(S)$. Solutions to the Killing spinor equation \eqref{eq:KSEM0T} which satisfy $\epsilon_1 = \mathfrak{c}(\epsilon_{2})$ are admissible supersymmetric solutions of $\cN=1$ supergravity with trivial scalar manifold. The study and classification of such solutions will be considered in a separate publication.

\begin{remark}
Note that equation \eqref{eq:KSEM0} is required by supersymmetry, which motivates \eqref{eq:KSEM0II} as natural Killing spinor equations to study on a four-dimensional Lorentzian manifold.
\end{remark}

% % % % % % % % % % % % % % % % % % % % % % % % % % % % % % % % % % % % % % 
% % % % % % % % % % % % % % % % % % % % % % % % % % % % % % % % % % % % % %

\subsection{Vanishing superpotential}
\label{sec:TrivialSuperpotential}

% % % % % % % % % % % % % % % % % % % % % % % % % % % % % % % % % % % % % % 
% % % % % % % % % % % % % % % % % % % % % % % % % % % % % % % % % % % % % %

Another particularly important special case of chiral supergravity is given by taking the superpotential $\cW\in H^{0}(\cM,\cL)$ to be the zero section of the holomorphic line bundle $\cL$. When $\cW$ is taken to be the zero section, the action functional of the theory reduces to:
\begin{equation*}
\mathfrak{Lag}[g,\varphi]= \mathrm{R}_{g} - \norm{\dd \varphi}^2_{\cG, g} \, , \qquad (g,\varphi)\in \Conf_{\mathfrak{Q}}(M)\, ,
\end{equation*}
	
\noindent
and therefore the theory reduces to Einstein gravity coupled to a non-linear sigma model with target space given by the complex  manifold $\cM$. The Killing spinor equations reduce in turn to:
\begin{equation*}
\nabla^\varphi \epsilon = 0\, , \qquad \dd \varphi^{0,1}\cdot\epsilon = 0\, ,
\end{equation*}

\begin{remark}
Note that taking $\cW=0$ the complex-conjugate map $\mathfrak{c}$ disappears from the Killing spinor equations and no longer plays a role in the formulation of the theory. This is the main source of simplification in this particular case, as the role played by $\mathfrak{c}$ in the formulation of the general theory is one of the genuine aspects brought by chiral local supersymmetry in four Lorentzian dimensions.
\end{remark}

\noindent
When $\cW = 0$ Lorentzian manifolds $(M,g)$ equipped with a solution $(g,\varphi,\epsilon)$ to the equations above are particular instances of Lorentzian $\Spin^c(3,1)$ manifolds admitting parallel spinors. Simply connected and complete Lorentzian manifolds of this type have been studied and classified in the literature, see Reference \cite{Moroianu} for the Riemannian case and Reference \cite{Ikemakhen} for the pseudo-Riemannian case. We do not expect every $\Spin^{c}(3,1)$ manifold admitting a parallel spinor to admit a solution to the Killing spinor equations and, on the other hand, assuming completeness of $(M,g)$ rules out many physically interesting Lorentzian four-manifolds. Adapting the main Theorem of \cite{Ikemakhen} to our situation we obtain the following result.

\begin{prop}
Let $M$ be a geodesically complete and simply-conneted Lorentzian four-manifold admitting a supersymmetric solution $(g,\varphi,\epsilon)$ to $\cN=1$ chiral supergravity with vanishing superpotential. Then we can have at most the following possibilities:

\

\begin{enumerate}
	\item $(M,g)$ is isometric to four-dimensional flat Minkowski space.

\
	
	\item $(M,g)$ is isometric to $(M,g) \simeq (\R^2\times X, \eta_{1,1}\times h)$, where $\eta_{1,1}$ is the flat two-dimensional Minkowski metric and $X$ is a Riemann surface equipped with a K\"ahler metric $h$.

\

	\item The holonomy group $H$ of $(M,g)$ is a subgroup of the parabolic subgroup $\SO(2)\ltimes \mathbb{R}^2 \subset \SO_0(3,1)$.
\end{enumerate}.

\
\end{prop}

\begin{remark}
Every geodesically complete and simply connected supersymmetric solution must be of the form described by the previous proposition. However, the converse may not be true, since a supersymmetric solution requires $(M,g)$ to admit a parallel spinor with respect to the specific connection $\nabla^{\varphi}$, which is coupled to the scalar map $\varphi$, which is in turn required to satisfy its corresponding Killing spinor equation. It is indeed an interesting open problem to classify which of the Lorentzian four-manifolds specified above can carry supersymmetric solutions of chiral supergravity. We will consider this problem in detail in Section \ref{sec:redX} for the specific case $(2)$ above, namely for the case in which $M = \mathbb{R}^2\times X$. Case (1) admits the obvious solution given by taking $\varphi$ as a constant map.
\end{remark}

% % % % % % % % % % % % % % % % % % % % % % % % % % % % % % % % % % % % % % 
% % % % % % % % % % % % % % % % % % % % % % % % % % % % % % % % % % % % % %

\section{Reduction to a Riemann surface}
\label{sec:redX}

% % % % % % % % % % % % % % % % % % % % % % % % % % % % % % % % % % % % % % 
% % % % % % % % % % % % % % % % % % % % % % % % % % % % % % % % % % % % % %

In this section we consider the reduction of $\cN=1$ chiral supergravity to an oriented two-manifold $X$ by assuming that the space-time manifold $M$ is of the form:
\begin{equation*}
M = \mathbb{R}^2\times X\, , 
\end{equation*}

\noindent
and it is equipped with the product metric $g_4 = \delta_{1,1}\times g$, where $\delta_{1,1}$ denotes the flat Minkowski metric on $\mathbb{R}^{2}$. The reduction is natural in the sense that the type of spinorial structures and Clifford modules coincide in dimension $(3,1)$ and $(2,0)$ by Clifford periodicity, since:
\begin{equation*}
(3-1) = (2-0) = 2\mod 8\, ,
\end{equation*}

\noindent
which corresponds to the real case of simple type in the standard classification of real Clifford algebras and their modules. Hence, all the structures introduced in the formulation of chiral $\cN=1$ supergravity on a Lorentzian four-manifold $M$ exist also on $X$, and we can consider directly the formulation of the theory on $X$. This would be equivalent to perform a reduction trivially along the $\mathbb{R}^2$ factor in $M$. We leave the details to the reader and proceed instead by directly formulating chiral supergravity on $X$. Using the fixed orientation on $X$, every Riemannian metric $g$ on $X$ defines a canonical complex structure $J_g$ given by point-wise counter-clockwise rotation. We define $Q_{g}$ to be the anti-canonical $\Spin^c(2)$ structure associated to $g$, which means that the determinant line bundle of $Q_g$ is the canonical bundle of $(X,g)$, compare \cite{Friedrich} for the terminology. See Remarks \ref{remark:canonicalspinc} and \ref{remark:antiholomorphic} for the reasons behind this choice of $\Spin^c(2)$ structure. For each Riemannian metric $g$, we define $S$ to be the complex spinor bundle canonically associated to $Q_{g}$. The spinor bundle $S$ admits an explicit model given by:
\begin{equation*}
S = \Lambda^{\ast,0}(X)\, ,
\end{equation*}

\noindent
where the splitting is performed with respect to the complex structure $J_g$. Clifford multiplication is given by:
\begin{equation}
\label{eq:Cliffordm}
\beta\cdot \alpha = 2 \beta^{1,0}\wedge \alpha + \iota_{(\beta^{\sharp})^{1,0}} \alpha\, ,
\end{equation}

\noindent
for all $\alpha\in \Omega^{\ast,0}(X)$ and all $\beta\in \Omega^{1}(X)$. Here we have set:
\begin{equation*}
\beta^{1,0} = \frac{1}{2} (\beta - i J_g \beta)\, , \qquad (\beta^{\sharp})^{1,0} = \frac{1}{2}\left( \beta^{\sharp} - i J \beta^{\sharp}\right)
\end{equation*}

\noindent
with the musical isomorphism $\sharp$ taken with respect to the metric $g$. The determinant line bundle associated to the anti-canonical $\Spin^c(2)$ structure $Q_g$ is given by the canonical bundle $K_g$ of $(X,g)$:
\begin{equation*}
\L_{Q_g} = K_g = \Lambda^{1,0}(X)\, ,
\end{equation*}

\noindent
The complex spinor bundle $S$ is thus a complex vector bundle of rank two, which splits in the usual way:
\begin{equation*}
S = S^{+} \oplus S^{-}\, ,
\end{equation*}

\noindent 
in terms of the chiral bundles $S^{+}$ and $S^{-}$, of complex rank one. In the presentation $S = \Lambda^{\ast,0}(X)$ of the spinor bundle, the chiral spinor bundles $S^{\pm}$ respectively correspond with:
\begin{equation*}
S^{+} \simeq \Lambda^{even,0}(X) \simeq \Lambda^{0,0}(X)\, , \qquad S^{-} \simeq \Lambda^{odd,0}(X) \simeq \Lambda^{1,0}(X)\, , 
\end{equation*}

\noindent
whereas the chiral spinor bundles $S^{\pm}_{c}$ correspond with:
\begin{equation*}
S^{+}_{c} \simeq \Lambda^{0,1}(X)\, , \qquad S^{-}_{c} \simeq  \Lambda^{0,0}(X)\, . 
\end{equation*}

\noindent
As required, we have:
\begin{equation*}
S^{+} \simeq S^{+}_{c}\otimes K_g\, , \qquad S^{-} \simeq S^{-}_{c}\otimes K_g\, .
\end{equation*}

\begin{remark}
In the standard terminology used in the literature, $S$ corresponds to the complex spinor bundle associated to the anti-canonical $\Spin^c(2)$ structure of $(X,g)$, whereas $S_c$ corresponds to the complex spinor bundle associated to the canonical $\Spin^c(2)$ structure of $(X,g)$, see for example \cite{Friedrich} for more details.
\end{remark}

\begin{remark}
Recall that we have a canonical isomorphism of real vector bundles:
\begin{equation*}
K_g = \Lambda^{1,0}(X)\simeq T^{\ast}X\, , 
\end{equation*}

\noindent
where $T^{\ast}X$ denotes the real cotangent bundle of $X$. Hence, the isomorphism type of $K_g$ as a real vector bundle does not depend on $g$. Not only this, the isomorphism type of $K_g$ as a $C^{\infty}$ complex line bundle does not depend on $g$ either. To see this, note that if $X$ is compact the Chern number of $K_g$ is minus the Euler characteristic of $X$ whereas if $X$ is open every holomorphic line bundle over $(X,g)$, in particular $K_g$, is holomorphically trivial. 
\end{remark}

\noindent
In the set-up introduced above, the notion of chiral triple, see Definition \ref{def:chiraltriple}, simplifies. This is due to the fact that, by assumption, we have established a canonical choice of complex structure and $\Spin^c(2)$ structure for every Riemannian metric $g$ on $X$, whose associated characteristic line bundle is $K_g$. More precisely, since the isomorphism class of $K_g$ as a complex line bundle does not depend on $g$ and furthermore the definition of chiral triple only requires $\cL^{\varphi}_{\cH}$ to be $C^{\infty}$ isomorphic to $K_g$, we can define the isomorphism of complex line bundles $\Psi$ between $\cL^{\varphi}$ and the determinant line bundle $K_g$ of the given $\Spin^c(2)$ structure independently of the metric $g$. Consequently, we arrive to the following simplification of a chiral triple, which we proceed to define.

\begin{definition}
A chiral triple $(\cL, \cH, \cW)$ at $(X,\cM)$ consists of a negative Hermitian holomorphic line bundle $(\cL,\cH)$ and a holomorphic section $\cW\in H^{0}(\cM,\cL)$ such that there exists a map $\varphi\colon X\to \cM$ and a metric $g$ on $X$ for which:
\begin{equation*}
K_g \simeq \cL^{\varphi}\, ,
\end{equation*}

\noindent
as complex line bundles.
\end{definition}

\

\noindent
We fix a chiral triple $\mathfrak{Q} = (\cL, \cH,\cW)$ on $(X,\cM)$ and consider the associated chiral supergravity on $X$. The fact that $\L_{Q}= K_g $ implies, directly from the definition of chiral triple, that the pull-back of the holomorphic line bundle $\cL$ by a scalar map must be $C^{\infty}$-isomorphic to $\L_{Q}$:
\begin{equation*}
\L_{Q} = K_{g} \simeq \cL^{\varphi}\, ,
\end{equation*}

\noindent
For every choice of isomorphism $\Psi\colon K_{g}\to \cL^{\varphi}$, we endow $K_{g}$ with the pull-back connection $\Psi^{\ast}\cD^{\varphi}$ with respect to $\Psi$, where $\cD$ denotes the Chern connection on $\cL_{\cH}$ and $\cD^{\varphi}$ its pull-back by $\varphi$. For ease of notation, we will sometimes denote $\Psi^{\ast}\cD^{\varphi}$ simply by $\cD^{\varphi}$. Recall that the definition of chiral triple also establishes the existence of a $C^{\infty}$ isomorphism:
\begin{equation*}
\P_{Q}\simeq \cP^{\varphi}_{\cH}\, ,
\end{equation*}

\noindent
Indeed, every choice of isomorphism $\Psi\colon K_{g}\to \cL^{\varphi}$ induces a canonical isomorphism between $\P_{Q}$ and $\cP^{\varphi}_{\cH}$, which we denote for simplicity by the same symbol, namely $\Psi\colon \P_{Q}\to \cP^{\varphi}_{\cH}$. For the type of complex spinor $S$ we are considering, which is associated to the anti-canonical $\Spin^c(2)$ structure of $(X,g)$, $\P_{Q}$ corresponds with the principal bundle of unitary coframes defined by $g^{\ast}$ (the dual of $g$) on $T^{\ast}X$. On the other hand, the principal $\U(1)$ bundle $\cP^{\varphi}_{\cH}$ corresponds to the $\U(1)$ reduction induced on $\L_Q$ by the metric $\Psi^{\ast}\cH^{\varphi}$. The isomorphism of principal $\U(1)$ bundles $\Psi\colon \P_{Q}\to \cP^{\varphi}_{\cH}$ gives an isomorphism between these two reductions.

Using $\Psi$ we equip $\P_Q$ with the pull-back connection $\Psi^{\ast}\cA^{\varphi}$, where $\cA^{\varphi}$ denotes the pull-back by $\varphi$ of the $\U(1)$ connection $\cA$ associated to the Chern connection $\cD$. Again, for ease of notation, we will sometimes denote $\Psi^{\ast}\cA^{\varphi}$ simply by $\cA^{\varphi}$. Lifting the Levi-Civita connection $\nabla^g$ associated to $g$ and $\cA^{\varphi}$ to the spinor bundle $S$ we obtain a connection on $S$ which we denote by $\nabla^{\varphi}$. Furthermore, using the isomorphism $\Psi\colon K_{g} \to \cL^{\varphi}$, the pull-backed superpotential $\cW^{\varphi}$ can be identified with a complex one-form on $X$ of $(1,0)$ type. This complex one-form is in principle not holomorphic since $\varphi$ is only assumed to be a $C^{\infty}$-map. The formulas defining the action functional and Killing spinor equations of the theory are formally the same as in $(3,1)$ dimensions, namely:
\begin{equation*}
\mathfrak{Lag}[g,\varphi]= \mathrm{R}_{g} - \norm{\dd \varphi}^2_{\cG,g} - \Phi^{\varphi}_k\, , \qquad (g,\varphi)\in \Conf_{\mathfrak{Q}}(X)\, ,
\end{equation*}

\begin{equation}
\label{eq:KSEX}
\nabla^\varphi \epsilon = \mathfrak{C}^{\Psi}_{\cW,\varphi}(\epsilon)\, , \qquad (\dd \varphi^{0,1})^{\flat}\cdot \epsilon = \mathfrak{C}^{\Psi,\cD}_{\cW,\varphi}(\epsilon)\, ,	
\end{equation}
	
\noindent	
for tuples $(g,\varphi,\epsilon,\Psi)\in \Conf_{\mathfrak{Q},E}$. We remind the reader that the vector bundle maps $\mathfrak{C}^{\Psi}_{\cW,\varphi}\colon S^{+}\to T^{\ast}X\otimes S^{-}$ and $\mathfrak{C}^{\Psi,\cD}_{\cW,\varphi}\colon S^{+}\to \Lambda^{1,0} T^{\ast}\cM^{\varphi}\otimes S^{+}$ are defined as follows:
\begin{equation*}
\mathfrak{C}^{\Psi}_{\cW,\varphi}(\epsilon)(v) \eqdef v\cdot \mathfrak{T}_{\Psi}(\left(\cW^{\varphi} \otimes \mathfrak{c}(\epsilon)\right))\, , \quad \forall\, v\in \mathfrak{X}(X)\, , \qquad \mathfrak{C}^{\Psi,\cD}_{\cW,\varphi}(\epsilon) =   \mathfrak{T}_{\Psi}((\cD \cW)^{\varphi}\otimes \mathfrak{c}(\epsilon))\, ,
\end{equation*}

\noindent
for $\epsilon\in \Gamma(S^{+})$. Here $\mathfrak{T}_{\Psi}\colon \cL^{\varphi}\otimes S^{\pm}_{c} \to S^{\pm}$ denotes the canonical isomorphism constructed in Proposition \ref{prop:Cwelldefined}. For the reasons behind the choice of positive chirality for the supersymmetry generator $\epsilon$ we refer the reader to Remark \ref{remark:antiholomorphic}. The first Killing spinor equation in \eqref{eq:KSEX} can be written simply as:
\begin{equation*}
\mathrm{D} \,\epsilon = 0\, .
\end{equation*}

\noindent
in terms of the following \emph{real-linear} connection $\mathrm{D}$ on the real rank-two vector bundle $S$:
\begin{equation*}
\mathrm{D}\eqdef \nabla^{\varphi}  - \mathfrak{C}^{\Psi}_{\cW,\varphi}\colon S \to S\, ,
\end{equation*}

\noindent
Hence, if $\epsilon$ is non-zero at some point (which we will assume in the following) it will be non-zero everywhere. Note however that given such $\epsilon$, the section $i\epsilon\in \Gamma(S)$ may be non-parallel even if $\epsilon$ is. Using the fact that $X$ is a Riemann surface, the Einstein equations \eqref{eq:Einsteineqs} of the theory drastically simplify and can be written as follows:
\begin{equation*}
\varphi^{\ast}\cG = \frac{\norm{\dd\varphi}^2_{g,\cG}}{2} \, g\, , \qquad \Phi^{\varphi} = 0\, ,
\end{equation*}

\noindent
where we have used that $\G(g)$ is identically zero and $\mathrm{Tr}_g(\varphi^{\ast}\cG) = \norm{\dd\varphi}^2_{\cG,g} $. Vanishing of the potential $\Phi^{\varphi}$ is equivalent to:
\begin{equation*}
\norm{\cD\cW}^{2}_{\cH,\cG}\circ\varphi = c \norm{\cW}^{2}_{\cH}\circ \varphi\, .
\end{equation*}

\noindent
Condition $\Phi^{\varphi}_k = 0$ does not imply in general neither $\Phi_k = 0$ nor $\dd\Phi_k = 0$ identically on $\cM$. Hence, the equation of motion for $\varphi$ does not reduce in general to the standard harmonicity condition and instead we have:
\begin{equation*}
(\Tr_g \nabla\dd\varphi)^{\flat} =  \frac{1}{2} \dd \Phi_k\circ \varphi\, .
\end{equation*}

\noindent
for $\varphi\colon X\to \cM$. Therefore, maps $\varphi\colon X\to \cM$ satisfying the equations of chiral $\cN=1$ supergravity on $X$ are particular instances of harmonic maps with potential \cite{Lemaire,FardounRatto}. Because of this, we will call solutions $\varphi$ to the previous equation \emph{harmonic maps with potentail} $\Phi_k$ (rather than $\frac{1}{2}\Phi_k$ as in Reference \cite{FardounRatto}). We obtain:

\begin{cor}
Let $\mathfrak{Q}$ be a chiral triple and let $(X, \cM)$ be an oriented two-manifold. A pair $(g,\varphi)$ satisfies the equations of chiral supergravity associated to $\mathfrak{Q}$ on $X$ if and only if:
\begin{equation}
\label{eq:sugraX}
\varphi^{\ast}\cG = \frac{\norm{\dd\varphi}^2_{\cG,g}}{2} \, g\, , \qquad \norm{\cD \cW}^{2}_{\cH,\cG}\circ \varphi = c \norm{\cW}^{2}_{\cH}\circ\varphi\, , \qquad (\Tr_g \nabla\dd\varphi)^{\flat} =  \frac{1}{2} \dd \Phi_k\circ \varphi\, .
\end{equation}

\noindent
In particular, if $\dd \Phi_k\circ \varphi = 0$ then $\varphi\colon X\to \cM$ is a harmonic map.
\end{cor}

\begin{remark}
Recall that condition $\Phi^{\varphi}_k = 0$ implies:
\begin{equation*}
\varphi(X)\subset \Phi^{-1}_k(0)\subset \cM\, ,
\end{equation*}
	
\noindent
In principle, the zero level set of $\Phi_k$ may not be a smooth $(2n-1)$-dimensional submanifold of $\cM$, since it is not guaranteed that $0$ be a regular value of $\Phi_k$. However, regularity of the critical points of $\Phi_k$ is indeed of physical relevance in relation for example with the stabilization of moduli in string theory compactifications.
\end{remark}

\noindent
The following proposition settles the classification of solutions to chiral supergravity on $X$ in the simple case in which $\varphi$ is the constant scalar map.

\begin{prop}
Let $k >0$ and $X$ connected. A pair $(g,\varphi)$ with $\dd\varphi =0$ is a solution to Equations \eqref{eq:sugraX} if and only if $\Phi_k(q) = 0$ and $\dd\Phi_k\vert_{q} = 0$, where $q$ is the constant value of $\varphi$. In particular, a $\cN=1$ chiral supergravity with chiral triple $\mathfrak{Q}$ admits solutions with constant scalar map if and only if $0\in\mathbb{R}$ is a critical value of the scalar potential $\Phi_k\colon\cM\to \mathbb{R}$. 
\end{prop}

\begin{proof}
Condition $\dd\varphi = 0$ implies that the first equation \eqref{eq:sugraX} is automatically solved for any metric $g$ on $X$. The second equation in \eqref{eq:sugraX} is equivalent with $\Phi\vert_{q} = 0$, whereas the third equation in \eqref{eq:sugraX} is equivalent with $\dd\Phi\vert_{q} = 0$. 
\end{proof}

\noindent
The previous corollary recovers the well-known fact that gravity in two dimensions is \emph{topological}, which translates into the fact that in dimension two every metric satisfies the vacuum Einstein equations. In the following we will use the symbol $\norm{\dd\varphi}^2_{g, \cG} > 0$ to denote that $\norm{\dd\varphi}^2_{g, \cG}$ is nowhere vanishing.

\begin{lemma}
\label{lemma:conf}
Let $\norm{\dd\varphi}^2_{g, \cG} > 0$. If $(g,\varphi)$ is a solution to the first equation in \eqref{eq:sugraX} then $\varphi\colon (X,g) \to (\cM, \cG)$ is a conformal immersion.
\end{lemma}

\begin{proof}
Since by assumption $\norm{\dd\varphi}^2_{g, \cG} \neq 0$ everywhere, the first equation in \eqref{eq:sugraX} implies that $\varphi^{\ast}\cG$ is conformal to $g$ and therefore $\varphi$ is a conformal immersion. 
\end{proof}

\begin{prop}
\label{prop:minimalimW0}
Let $\mathfrak{Q}$ be a chiral triple with vanishing superpotential $\cW$. A pair $(g,\varphi)$ with $\norm{\dd\varphi}^2_{g, \cG} > 0$ is a solution to the associated chiral supergravity if and only if $\varphi$ is a minimal immersion of $X$ into $\cM$.
\end{prop}

\begin{proof}
If $\cW$ is the zero section then $\Phi_k$ is identically zero and thus $(\mathrm{grad}_{\cG} \Phi_k)^{\varphi} = 0$, implying that the scalar equation for $\varphi$ reduces to the harmonicity condition:
\begin{equation*}
\Tr_g \nabla\dd\varphi = 0\, .
\end{equation*} 	

\noindent
Hence, using Lemma \ref{lemma:conf}, it follows that $\varphi\colon X\to \cM$ is a harmonic conformal immersion, which for a two-dimensional source is equivalent with $\varphi$ being a minimal immersion \cite{Report1,Report2}. 
\end{proof}

\noindent
For every solution $(g,\varphi)$ that does not satisfy condition $\norm{\dd\varphi}^2_{\cG,g} > 0$ there exists a non-empty closed subset $C\subset X$ of $X$ at which $\norm{\dd\varphi}^2_{\cG,g}\vert_C = 0$. Since $g$ is by assumption a non-degenerate Riemannian metric, the Einstein equation implies that the real rank of $\dd\varphi$ is zero over $C$ and two over $X\backslash C$. The critical set of $\varphi$ coincides with $C$ and does not contain points at which the rank of $\varphi$ is one. When restricted to $C$, the Einstein equation of chiral supergravity is automatically satisfied independently of the restriction of $g$ to $C$:
\begin{equation*}
\varphi^{\ast}\cG\vert_C = 0\, , \qquad \norm{\dd\varphi}^2_{\cG,g}\vert_C = 0\, ,
\end{equation*}

\noindent
Fixing $\varphi$, the Einstein equation becomes a quadratic algebraic equation $\varphi^{\ast}\cG = \frac{\norm{\dd\varphi}^2_{g,\cG}}{2} \, g$ for $g$ and can in principle be use to determine $g$ on $X\backslash C$. However, over $C$ the Einstein is automatically satisfied and provides no information on the behavior of $g$ over $C$, which is thus not determined over $C$. More explicitly, let us fix real coordinates $\left\{ x^a\right\}$, $a =1,2$, and $\left\{\phi^A \right\}$, $A =1,\hdots,2n$ around $p\in X$ and $\varphi(p)$, respectively. We define:
\begin{equation*}
 L_{ab}(p) \eqdef \partial_a \varphi^A(p)\, \partial_b \varphi^B(p)\, \cG_{AB}(p)\, ,
\end{equation*} 

\noindent
where $\varphi^A$, $\cG^A$ are the local expressions of the components of $\varphi$ and $\cG$ in the given coordinates $\left\{ x^a\right\}$ and $\left\{\phi^A \right\}$. Note that $\left\{ L_{ab}(p)\right\}$ is a matrix of real numbers. After a quick manipulation, the Einstein equation evaluated at $p$ can be written as follows:
\begin{equation*}
L_{ab}(p) \det(g(p)) = \frac{g_{ab}(p)}{2} \, \left( g_{22}(p) L_{11}(p) + g_{11}(p) L_{22}(p) - 2 g_{12}(p) L_{12}(p)\right)\, ,
\end{equation*}

\noindent
which is a non-trivial system of quadratic equations for the components of $g_{ab}$ if $p\in X\backslash C$. 

Let $\varphi\colon X\to \cM$ be a smooth map with critical set $C$. Let $g^o$ be a smooth metric on $X\backslash C$ such that equations \eqref{eq:sugraX} are satisfied on the interior of $X\backslash C$. A natural question is: can $g^o$ be extended to a smooth metric $g$ on $X$ such that $(g,\varphi)$ is a solution of equations \eqref{eq:sugraX} on the whole $X$? The answer to this question highly depends on the properties enjoyed by $C$. If $\cW$ is not the zero section, we have little control over $C$, which in principle could even have interior points. Nonetheless, in Section \ref{sec:susysolX} we will see that for supersymmetric solutions there always exists a smooth extension to $X$ canonically induced by the chiral triple of the theory. 
 
\begin{remark}
As previously mentioned, a smooth map $\varphi\colon X\to \cM$ from a Riemann surface into a Riemannian manifold is a harmonic conformal immersion if and only if it is a minimal immersion. The latter is defined as a critical point of the area functional among compactly supported variations. In particular, it is well-known that every holomorphic immersion of a Riemann surface into K\"ahler manifold is a minimal immersion. Minimal immersions and holomorphic maps have been extensively studied in the literature, and many results on existence and non-existence of such maps are by now available, see for example \cite{Report1,Report2} and references therein. 

%In particular, under some specific conditions it can be guaranteed that every minimal immersion is holomorphic. For instance, every harmonic map from the two-sphere to a hyperbolic Riemann surface is holomorphic \cite{Lemaire,Wood}, and every stable map from $\mathbb{P}^1$ to a K\"ahler manifold with positive holomorphic bi-sectional curvature is holomorphic \cite{SiuYau} .
\end{remark}

\noindent
For chiral supergravities of possibly non-vanishing superpotential $\cW^{\varphi}\neq 0$ we obtain the following corollary, which characterizes the space of solutions of chiral $\cN=1$ supergravity on $X$ satisfying $\norm{\dd\varphi}^2_{\cG,g} > 0$ as maps $\varphi\colon X\to \cM$.

\begin{cor}
There is a canonical bijection between the set of solutions of chiral $\cN=1$ supergravity on $X$ such that $\norm{\dd\varphi}^2_{g,\cG} > 0$ and the set of non-constant conformal harmonic immersions $\varphi\colon (X,g)\to (\cM,\cG)$ with potential $\Phi_k$ vanishing along $\varphi$.
\end{cor} 

\noindent
If $\cW$ is the zero section, the $\Phi_k = 0$ on $\cM$, and hence given a chiral triple $\mathfrak{Q}$ on $(X,\cM)$ with vanishing superpotential, solutions to chiral supergravity with $\norm{\dd\varphi}^2_{\cG,g} > 0$ correspond to minimal immersions $\varphi\colon (X,g)\to (\cM,\cG)$. A systematic study of Riemann surfaces admitting supersymmetric solutions to chiral $\cN=1$ supergravity with possibly non-vanishing superpotential is beyond the scope of this manuscript and will be considered elsewhere. In Section \ref{sec:susysolX} we will consider the classification problem of supersymmetric solutions with vanishing superpotential.

\begin{remark}
\label{remark:canonicalspinc}
The formulation introduced in this Section of chiral $\cN=1$ supergravity on $X$ fixes the $\Spin^c(2)$ structure of the theory to be, given a Riemannian metric $g$ on $X$, the anti-canonical one on $(X,g)$, and takes $S$ to be the associated complex spinor bundle through the tautological representation of $\Spin^c(2)$. It is in principle possible to use a different (inequivalent) complex spinor bundle to construct the theory. If we change the complex spinor bundle, we should expect to obtain a non-equivalent chiral $\cN=1$ supergravity on $X$, see for example \cite{Figueroa-OFarrill:2005vxy}. However, if we assume the existence of a non-vanishing chiral complex spinor $\epsilon \in \Gamma(S)$, then the allowed choices of complex spinor bundle $S$ are restricted and we can construct a canonical isomorphism of complex spinor bundles between $S$ and either $\Lambda^{0,\ast}(X)$ or $\Lambda^{\ast,0}(X)$ (note that every chiral spinor in two-dimensions is automatically pure). Using the map
\begin{equation*}
\Lambda^{\ast,\ast} (X) \to S\, , \qquad \alpha \mapsto \alpha \cdot \epsilon\, ,
\end{equation*}

\noindent
it can be seen that if there exists an everywhere non-zero spinor $\epsilon$ of positive chirality then we have an isomorphism of complex spinor bundles:
\begin{equation*}
S \simeq \Lambda^{\ast,0}(X)\, ,
\end{equation*}

\noindent
with Clifford multiplication given by \eqref{eq:Cliffordm}. On the other hand, if the complex spinor $\epsilon$ is of negative chirality, we obtain the isomorphism:
\begin{equation*}
S \simeq \Lambda^{0,\ast}(X)\, ,
\end{equation*}

\noindent
with Clifford multiplication given by:
\begin{equation*}
\beta\cdot \alpha = 2 \beta^{0,1}\wedge \alpha + \iota_{(\beta^{\sharp})^{0,1}} \alpha\, ,
\end{equation*}

\noindent
for all $\alpha\in \Omega^{\ast,0}(X)$ and all $\beta\in \Omega^{1}(X)$. The determinant line bundle is in this case the anti-canonical line bundle $K^{\ast}_g$ of $(X,g)$ and the resulting complex spinor bundle corresponds to the canonical $\Spin^c(2)$ structure on $(X,g)$. Hence, if $X$ admits supersymmetric solutions (a condition that implies the existence of a non-vanishing $\epsilon \in \Gamma(S^{\pm})$), there is no loss of generality in assuming that the complex spinor bundle $S$ is associated to either the canonical $\Spin^c(2)$ structure on $(X,g)$ if $\epsilon$ has negative-chirality or the anti-canonical $\Spin^c(2)$ structure on $(X,g)$ if $\epsilon$ has positive-chirality.
\end{remark}

% % % % % % % % % % % % % % % % % % % % % % % % % % % % % % % % % % % % % % 
% % % % % % % % % % % % % % % % % % % % % % % % % % % % % % % % % % % % % %

\subsection{Supersymmetric solutions with vanishing superpotential}
\label{sec:susysolX}

% % % % % % % % % % % % % % % % % % % % % % % % % % % % % % % % % % % % % % 
% % % % % % % % % % % % % % % % % % % % % % % % % % % % % % % % % % % % % %

We consider the Killing spinor equations on $X$ with vanishing superpotential. The local structure of the supersymmetric solutions of chiral $\cN=1$ supergravity has been considered in References \cite{OrtinN1,GGP}, where the generic local form of the supersymmetric solutions of the theory in four dimensions has been partially characterized in terms of a minimal set of partial differential equations. The goal of this sub-section is to obtain a global classification result in the special case in which the superpotential $\cW$ is zero. In the course of the proof of Theorem \ref{thm:solsusyX} we will see that every supersymmetric configuration is actually a solution and therefore it is not necessary to consider the equations of motion explicitly in order to classify supersymmetric solutions. Let $\mathfrak{Q}$ be a chiral triple on $(X,\cM)$ such that $\cW = 0$. The Killing spinor equations of chiral supergravity associated to $(\cM,\mathfrak{Q})$ on $X$ reduce to:
\begin{equation}
\label{eq:KSEXW0}
\nabla^{\varphi}\epsilon = 0 \, , \qquad \dd \varphi^{0,1} \cdot \epsilon = 0\, ,
\end{equation}

\noindent
for $\epsilon \in \Gamma(S^{+})$. 

\begin{thm}
\label{thm:solsusyX}
Let $\mathfrak{Q}$ be a chiral triple on $(X,\cM)$ such that $\cW = 0$. A triple $(g,\varphi,\Psi)$ with non-constant $\varphi$ is a supersymmetric solution of the chiral supergravity associated to $(\cM,\mathfrak{Q})$ if and only if the following conditions hold:

\

\begin{enumerate}
	\item The smooth map $\varphi\colon (X,g)\to (\cM,\cG)$ is a holomorphic map with respect to $J_g$ and the fixed complex structure $\cI$ on $\cM$. 
	
	\
	
	\item $\Psi\colon K_g \xrightarrow{\simeq} \cL^{\varphi}$ is an isomorphism of holomorphic line bundles such that: 
	\begin{equation*}
	g^{\ast}_c  = \kappa\, \Psi^{\ast} \cH^{\varphi}\, , 
	\end{equation*}
	
	\noindent
	for a constant $\kappa\in \mathbb{R}_{>0}$, where $g^{\ast}_c$ denotes the Hermitian metric induced by $g$ on $\Lambda^{1,0}(X)$. 
\end{enumerate}

\noindent
These conditions imply that:
\begin{equation*}
\mathrm{R}_{g} = \norm{\dd\varphi}^2_{\cG,g}\, , 
\end{equation*}

\noindent
and that the K\"ahler metric $\cG$, the Riemannian metric $g$ and the map $\varphi$ satisfy:
\begin{equation*}
\varphi^{\ast}\cG = \frac{\norm{\dd\varphi}^2_{\cG,g}}{2}\, g = \frac{\mathrm{R}_{g}}{2}\, g\, , 
\end{equation*}

\noindent
Hence, $\varphi$ is a conformal immersion of $X\backslash C$ into $\cM$, where $C\subset X$ denotes the critical set of $\varphi$. Furthermore, if $(X,g)$ is compact, then it is biholomorphic with the Riemann sphere $\mathbb{P}^1$.
\end{thm}

\begin{proof}
Let $(g,\varphi,\Psi)$ be a supersymmetric solution on $X$ with non-constant $\varphi$. The Riemannian metric $g$ induces a complex structure $J_g$ upon use of the fixed orientation of $X$. Assuming that $\epsilon$ is non-zero at a point, the first Killing spinor equation implies that $\epsilon$ is everywhere non-vanishing. Therefore, we obtain an isomorphism of complex spinor bundles $S \simeq \Lambda^{0,0}(X) \oplus \Lambda^{1,0}(X) $ which maps $\epsilon$ to $1\in \Lambda^{0,0}(X)$. Hence, the second Killing spinor equation is equivalent with:
\begin{equation*}
\partial\varphi^{0,1}= 0\, .
\end{equation*}

\noindent
implying that a map $\varphi\colon X\to \cM$ satisfying the Killing spinor equations is necessarily holomorphic, a condition that immediately implies harmonicity of $\varphi$. Using that $\varphi$ is holomorphic it follows that $(\cL^{\varphi},\cH^{\varphi})$ is a holomorphic line bundle over $(X,J_g)$ and in addition $\cD^{\varphi}$ coincides with the Chern connection on $(\cL^{\varphi},\cH^{\varphi})$. Let $\left\{ U, w\right\}$, $U\subset X$ open, be a local holomorphic coordinate on $X$ and let $\left\{ V, z^i \right\}$, $V\subset \cM$ open, $i=1,\hdots ,n$, be local complex coordinates on $\cM$ such that $\varphi(U)\subset V$. We write $g$ as follows:
\begin{equation*}
g = e^{F} dw \odot d\bar{w}\, ,
\end{equation*}

\noindent
for a smooth function $F\in C^{\infty}(V)$. Using that $\epsilon =1$, the first Killing spinor equation in \eqref{eq:KSEXW0} can be locally written as:
\begin{equation*}
i \mathrm{A}^{g \ast} = \Psi^{\ast}\cA^{\varphi}\, ,
\end{equation*}

\noindent
where $i\mathrm{A}^{g \ast}$ denotes the Chern connection of the Hermitian metric $g^{\ast}(\cdot,\bar{\cdot})$ on  $K_g$ and $\cA^\varphi$ denotes the connection one-form associated to the Chern connection $\cD^{\varphi}$. By type decomposition, the previous equation locally reads:
\begin{equation}
\label{eq:localformAA}
\partial_{w} z^i \partial_{z^{i}} \mathcal{K}(z(w),\bar{z}(\bar{w})) dw = - \partial_{w} F(w,\bar{w}) dw\, ,
\end{equation}

\noindent
where we are using a local holomorphic trivialization $l$ of $\cL$ in which $\cH(l,l) = e^{\mathcal{K}}$. Since this equation must hold on any pair of complex coordinate charts $\left\{ U, w\right\}$ and $\left\{ V, z^i \right\}$ as above, we conclude that:
\begin{equation}
\label{eq:proofequalconnections}
\nabla^{\mathrm{C}}_{g^{\ast}} = \Psi^{\ast}\cD^{\varphi}\, ,
\end{equation}

\noindent
globally on $X$, where $\nabla^{\mathrm{C}}_{g^{\ast}}$ denotes the Chern connection on $K_g$. Therefore, the Chern connection on $(\cL^{\varphi},\cH^{\varphi})$ is mapped, through the $C^{\infty}$ diffeomorphism $\Psi$, to the Chern connection $\nabla^{\mathrm{C}}_{g^{\ast}}$ on the holomorphic cotangent bundle of $X$ equipped with the Hermitian inner product induced by $g^{\ast}$. Since $K_g$ and $\cL^{\varphi}$ come equipped with the holomorphic structures induced respectively by $\nabla^{\mathrm{C}}_{g^{\ast}}$ and $\cD^{\varphi}$, we conclude that $\Psi\colon K_g \xrightarrow{\simeq} \cL^{\varphi} $ is in fact an isomorphism of holomorphic line bundles. By integrating \eqref{eq:proofequalconnections}, it can be seen that the fact that the Chern connections associated to $g^{\ast}$ and $\Psi^{\ast}\cH^{\varphi}$ are equal is equivalent with:
\begin{equation*}
g^{\ast}_c = \kappa\, \Psi^{\ast} \cH^{\varphi}\, , \qquad \kappa\in \mathbb{R}_{>0}\, ,
\end{equation*}

\noindent
which in local coordinates implies $F(w,\bar{w}) = - \mathcal{K}(z(w),\bar{z}(\bar{w})) - \log\,\kappa$ and gives the following local expression for $g$:
\begin{equation*}
g = \frac{2}{\kappa} e^{-\mathcal{K}(z(w),\bar{z}(\bar{w}))} dw\odot d\bar{w}\, .
\end{equation*}

\noindent
Taking the derivative of Equation \eqref{eq:localformAA} with respect to $\bar{w}$ we obtain:
\begin{equation*}
\partial_{w}\partial_{\bar{w}} F = - \partial_{w} z^{i}  \partial_{\bar{w}}\bar{z}^{j} \partial_{z^{i}} \partial_{\bar{z}^{j}}\mathcal{K}(z(w),\bar{z}(\bar{w}) \, ,
\end{equation*} 

\noindent
where we have used that $\varphi$ is holomorphic. This equation directly implies, after suitable identifications, the following relation:
\begin{equation*}
\frac{\mathrm{R}_{g}}{2}\, g =  \varphi^{\ast}\cG\, ,
\end{equation*}

\noindent
where we have used that the Riemannian scalar curvature $\mathrm{R}_g$ of $g$ is locally explicitly given by $\mathrm{R}_g = - \Delta_g F$ in terms of the Laplacian $\Delta_g$ on $(X,g)$. Taking the trace of the previous equation, we obtain:
\begin{equation*}
\mathrm{R}_{g} =  \norm{\dd\varphi}^2_{\cG,g}\, ,
\end{equation*}

\noindent
and thus the curvature of $g$ is non-negative and prescribed by the norm of $\dd\varphi$. If $X$ is compact, this implies that the Euler characteristic $\chi(X)$ of is strictly positive (recall that we are assuming that $\varphi$ is non-constant) and thus $X$ must be biholomorphic with the Riemann sphere $\mathbb{P}^1$. In particular, for any $p\in X$ we have $\mathrm{R}_{g}\vert_p = 0$ if and only if $\norm{\dd\varphi}^2_{\cG,g}\vert_p = 0$. Since, $\varphi$ is holomorphic, $C\subset X$ is a discrete subset of $X$, whence finite if $X$ is compact. Note that $\mathrm{R}_g\vert_C = 0$ and $\mathrm{R}_g \vert_{X\backslash C} > 0$. Combining the previous two equations we conclude:
\begin{equation*}
\frac{\norm{\dd\varphi}^2_{\cG,g}}{2}\, g  = \varphi^{\ast}\cG\, ,
\end{equation*}

\noindent
whence every supersymmetric configuration is in fact a supersymmetric solution. On $X\backslash C$ the scalar curvature $\mathrm{R}_g$ of $g$ as well as $\norm{\dd\varphi}^2_{\cG,g}$ are both nowhere vanishing. In particular:
\begin{equation*}
g =  \frac{2}{\norm{\dd\varphi}^2_{\cG,g}} \,\varphi^{\ast}\cG\, ,
\end{equation*}

\noindent
whence $\varphi$ is a holomorphic conformal immersion of $X\backslash C$ into $\cM$. Direct computation in local coordinates on $X\backslash C$ shows (we take $\kappa=2$ for simplicity):
\begin{equation*}
\frac{2}{\norm{\dd\varphi}^2_{\cG,g}} \,\varphi^{\ast}\cG = \frac{2}{4\, e^{\mathcal{K}}\partial_{w} z^{i}  \partial_{\bar{w}}\bar{z}^{j} \partial_{z^{i}}\, \partial_{\bar{z}^{j}}\mathcal{K} } \,2\,\partial_{w} z^{i}  \partial_{\bar{w}}\bar{z}^{j} \partial_{z^{i}} \partial_{\bar{z}^{j}}\,\mathcal{K}\, \dd w\odot \dd\bar{w} = e^{-\mathcal{K}} \dd w \odot \dd\bar{w}\, ,
\end{equation*}

\noindent
as expected. For the converse, we just note that any pair $(g,\varphi,\Psi)$ obeying conditions (1) and (2) of the theorem satisfies the Killing spinor equations \eqref{eq:KSEXW0} for $\epsilon = 1$.
\end{proof}

\begin{remark}
If $\varphi$ is constant, then \eqref{eq:KSEXW0} implies that $g$ is flat and $K_g$ must be holomorphically trivial. In particular, if $X$ is compact then it is biholomorphic with an elliptic curve $\mathbb{E}$. 
\end{remark}

\begin{remark}
When $X$ is open, the requirement that $K_g$ be negative does not pose any restriction on $X$, see Example \ref{ep:spinStein} for more details. 
\end{remark}

\begin{remark}
\label{remark:antiholomorphic}
The supersymmetric structure of chiral $\cN=1$ supergravity in four dimensions requires the supersymmetry generator to be of negative chirality. When reduced to two dimensions, negative chirality in four dimensions allows for either negative or positive chirality in two dimensions, depending on the decomposition chosen for the supersymmetry spinor in four dimensions. Here we have fixed the supersymmetry parameter $\epsilon$ in two dimensions to have positive chirality, see Remark \ref{remark:canonicalspinc}. However, we could have chosen $\epsilon$ to have negative chirality and $(X,g)$ to be endowed with the canonical $\Spin^c(2)$ structure. In this case, it can be shown that the Killing spinor equations require $\varphi$ to be an \emph{anti-holomorphic} map, and that a supersymmetric solution with $\epsilon$ of negative chirality corresponds with a supersymmetric solution with $\epsilon$ of positive chirality after complex-conjugating the complex structure $J_g\mapsto -J_g$ on $X$. We have chosen to work with anti-canonical $\Spin^c(2)$ structures and positive chirality supersymmetry generators in order to avoid working with supersymmetric solutions having anti-holomorphic, instead of holomorphic, scalar maps $\varphi$.
\end{remark}

\noindent
When $X$ admits supersymmetric solutions, the existence of a nowhere vanishing positive-chirality spinor $\epsilon \in \Gamma(S^{+})$ determines $S$ as the complex spinor bundle associated to the canonical $\Spin^c(2)$ structure on $(X,g)$. Hence, we obtain the following Corollary.

\begin{cor}
The spinor bundle $S$ of a superymmetric solution is necessarily the complex spinor bundle associated to either the canonical or anti-canonical $\Spin^c(2)$ structure on $(X,g)$ through the tautological representation of $\Spin^c(2)$.
\end{cor}

\noindent
This corollary gives a very explicit example of how the existence of supersymmetric solutions depends on the choice of isomorphism class of the spinor bundle. Inspired by Theorem \ref{thm:solsusyX} we introduce the notion of chiral map.

\begin{definition}
Let $X$ be an oriented real two-manifold and let $(\cM,\cL,\cH)$ be a complex manifold equipped with a negative Hermitian holomorphic line bundle $(\cL,\cH)$. We say that a pair $(\varphi,\Psi)$ is a \emph{chiral map} with respect to $(\cM,\cL,\cH)$ if there exists a complex structure $J=J_{\varphi}$ on $X$ such that:
\begin{equation*}
\varphi\colon (X,J)\to (\cM,\cI)\, ,
\end{equation*}

\noindent
is holomorphic and: 
\begin{equation*}
\Psi\colon K\xrightarrow{\simeq} \cL^{\varphi}\, ,
\end{equation*}

\noindent
is an isomorphism of holomorphic line bundles, where $K$ is the canonical bundle of $(X,J)$.
\end{definition}

\noindent
Associated to every chiral map $(\varphi,\Psi)$ we have a supersymmetric solution $(g,\varphi,\Psi)$ where $g^{\ast} = \kappa\, \Psi^{\ast}\cH^{\varphi}$ and, vice-versa, every supersymmetric solution of chiral supergravity on $X$ gives rise to a chiral map. Aside from the role they play in chiral supergravity, chiral maps are interesting because they are particular instances of holomorphic maps of Riemann surfaces into K\"ahler manifolds and provide solutions to the \emph{coupled} problem of prescribing the scalar curvature of a Riemann surface to:
\begin{equation*}
\mathrm{R}_g = \norm{\dd\varphi}^2_{g,\cG}\, .
\end{equation*}

\begin{remark}
The problem of prescribing the Gaussian curvature of a compact Riemann surface has been extensively studied in the literature, see for example \cite{Kzdan} and references therein. Reference \cite{KzdanWarner} solves the problem in the case of $X$ being a closed Riemann surface. In particular, it asserts that the sphere $S^2$ admits a metric $g$ with curvature $f\in C^{\infty}(X)$ if and only if the obvious Gauss-Bonnet sign condition is satisfied. Since, as shown in Theorem \ref{thm:solsusyX}, the curvature of a supersymmetric solution $(g,\varphi,\Psi)$ is non-negative we conclude that supersymmetry imposes a non-trivial restriction on the allowed metrics appearing in a supersymmetric solution on $S^2$.
\end{remark}

\noindent
When $\cM$ is also an oriented two-dimensional real manifold,  Theorem \ref{thm:solsusyX} implies the following corollary.

\begin{cor}
Let $\cM$ be connected and complex one-dimensional and let $(\varphi,\Psi)$ be a chiral map with respect to $(\cM,\cL,\cH)$. If $\varphi$ is proper then it is a holomorphic branched covering of Riemann surfaces. If in addition $X$ is compact, whence biholomorphic to the Riemann sphere $\mathbb{P}^1$, then $\cM$ is a compact Riemann surface whose Euler characteristic $\chi(\cM)$ satisfies:
\begin{equation*}
\chi(\cM) = \frac{2 + k}{d}\, ,
\end{equation*}

\noindent
where $d$ denotes the degree of $\varphi$ and $k$ denotes its total ramification index. In particular, $(2+k)$ is divisible by $d$.
\end{cor}

\begin{proof}
By definition of chiral map, $\varphi$ is holomorphic. Since by assumption $\cM$ is connected and $\varphi$ is  proper, $\varphi$ must be surjective and thus it is a holomorphic branched covering, implying also that if $X$ is compact then $\cM$ is also compact. Theorem \ref{thm:solsusyX} implies now that $X = \mathbb{P}^1$ and the Riemann-Hurwitz formula implies the remaining statements. 
\end{proof}

\noindent
Let $(g,\varphi,\Psi)$ be a supersymmetric solution. When $X$ is non-compact the norm of $\varphi$ may diverge and $g$ may not be complete. In order to guarantee that $(X,g)$ is physically admissible we assume $g$ is complete. Furthermore we assume finite L$^2$ energy:
\begin{equation*}
 \nnorm{\dd\varphi}^2_{g,\cG} < \infty\, ,
\end{equation*} 

\noindent
Under these conditions, we can obtain the following, perhaps surprising, classification result on the possibly non-compact surfaces $X$ admitting supersymmetric solutions with complete Riemannian metric $g$ and finite energy.

\begin{thm}
\label{thm:finitecurvature}
Let $X$ be an oriented real two-manifold carrying a supersymmetric solution $(g,\varphi,\Psi)$ (with respect to some chiral triple $\mathfrak{Q}$) such that $g$ is complete and $\varphi$ has positive finite energy. Then, one of the following holds:
\

\begin{itemize}
\item $(X,J_g)$ is biholomorphic to the complex projective line $\mathbb{P}^1$ and:
\begin{equation*}
 \nnorm{\dd\varphi}^2_{g,\cG} = 8\pi\, .
\end{equation*}

	\
	
\item $(X,J_g)$ is biholomorphic to the complex plane $\mathbb{C}$, and there exists a neighborhood $U(\infty)$ of infinity with complex coordinate $w$ which is isomorphic to the punctured unit disk equipped with the metric: 
\begin{equation*}
g\vert_{U(\infty)} = \frac{e^{F}}{|w|^2}\, dw\otimes d\bar{w}\, , \qquad F\in L^1\, , \qquad \Delta F \in L^1\, ,
\end{equation*}
	
\noindent
and furthermore:
\begin{equation*}
\nnorm{\dd\varphi}^2_{g,\cG} = 4\pi\, .
\end{equation*}

\item $(X,J_g)$ is biholomorphic to the complex plane $\mathbb{C}$, and there exists a neighborhood $U(\infty)$ of infinity with complex coordinate $w$ which is isomorphic to the punctured unit disk equipped with the metric: 
\begin{equation*}
g\vert_{U(\infty)} = \frac{e^{F}}{|w|^4}\, dw\otimes d\bar{w}\, , \qquad F\in L^1\, , \qquad \Delta F \in L^1\, ,
\end{equation*}

\noindent
and $\varphi$ is constant.

\item $(X,J_g)$ is biholomorphic to the punctured complex plane $\mathbb{C}^{\ast}$, $\varphi$ is constant and there exists neighborhoods $U(\infty)$ of infinity with complex coordinate $w_{\infty}$ and $U(0)$ of zero with complex coordinate $w_{0}$ which are isomorphic to the punctured unit disk equipped respectively equipped with the metric: 
\begin{equation*}
g\vert_{U(\infty)} = \frac{e^{F_{\infty}}}{|w_{\infty}|^2}\, dw_{\infty}\otimes d\bar{w}_{\infty}\, , \qquad F_{\infty}\in L^1\, , \qquad \Delta F_{\infty} \in L^1\, ,
\end{equation*}

\begin{equation*}
g\vert_{U(0)} = \frac{e^{F_0}}{|w_0|^2}\, dw_0\otimes d\bar{w}_0\, , \qquad F_0\in L^1\, , \qquad \Delta F_0 \in L^1\, ,
\end{equation*}

\item $(X,J_g)$ is biholomorphic to a complex elliptic curve and $\varphi$ is constant.
\end{itemize}
\end{thm}
 
\begin{proof}
Let $(g,\varphi,\Psi)$ be such a supersymmetric solution. Finiteness of the energy of $\varphi$ imply, upon use of Theorem \ref{thm:solsusyX}, that $(X,g)$ has non-negative and finite total curvature. Applying now A. Huber's Theorem as stated in \cite[Pages 1-2]{HulinTroyanov} we conclude. 
\end{proof}

\begin{remark}
The idea behind the previous Theorem is simple, once we have the results of A. Huber at our disposal: finiteness of the total curvature of $(X,g)$ implies that $X$ has a natural compactification, that is, it is biholomorphic with a compact Riemann surface with a finite number of points removed. The fact that the total curvature not only finite but non-negative further restricts they type of Riemann surfaces $(X,g)$, since removing more point \emph{decreases} the total finite curvature. Since the total finite curvature needs to remain non-negative, only the cases appearing in the previous Theorem can happen.
\end{remark}

\noindent
The four-dimensional Lorentzian manifold $(M_4,g_4)$ associated to the supersymmetric solutions considered in Theorem \ref{thm:solsusyX} is of the form:
\begin{equation*}
(M_4,g_4) = (\mathbb{R}^2\times X, \eta_{1,1}\times g)\, ,
\end{equation*}

\noindent
where $(X,g)$ is a Riemann surface with the Riemannian metric $g$ being part of a supersymmetric solution $(g,\varphi,\Psi)$ on $X$. As we have said, if $\varphi$ is non-constant and $X$ is compact then it is biholomorphic with $\mathbb{P}^1$ and in this case:
\begin{equation*}
(M_4,g_4) = (\mathbb{R}^2\times \mathbb{P}^1, \eta_{1,1}\times g)\, .
\end{equation*}

\noindent
Therefore, we recover the spherical \emph{near horizon} geometry characterized in \cite{Gutowski:2010gv}, see also \cite{Meessen:2010ph}. Note that $g$ may not be the round metric on $S^2$. If $\varphi$ is constant, $X$ is biholomorphic with an elliptic curve and we recover the \emph{toroidal} near horizon geometry characterized in \cite{Gutowski:2010gv}.

% % % % % % % % % % % % % % % % % % % % % % % % % % % % % % % % % % % % % % 
% % % % % % % % % % % % % % % % % % % % % % % % % % % % % % % % % % % % % %

\subsection{Anti-supersymmetric solutions}
\label{sec:antisusysolutions}

% % % % % % % % % % % % % % % % % % % % % % % % % % % % % % % % % % % % % % 
% % % % % % % % % % % % % % % % % % % % % % % % % % % % % % % % % % % % % %

As explained in Remark \ref{remark:antiholomorphic}, had we chosen the canonical $\Spin^c(2)$ structure on $(X,g)$ to formulate chiral $\cN=1$ supergravity, we would have found that scalar maps $\varphi$ of supersymmetric solutions are anti-holomorphic and supersymmetric solutions in this case are equivalent to those characterized in Theorem \ref{thm:solsusyX} with the complex structure $J_g$ replaced by $-J_g$. Nonetheless, we can propose a natural modification of the Killing spinor equations which still yields solutions to chiral $\cN=1$ supergravity with holomorphic scalar maps. These solutions are however not supersymmetric. In this section we introduce and classify these solutions, which we call \emph{anti-supersymmetric}.

We assume $X$ to be equipped with a fixed orientation. We define now $Q_{g}$ to be the canonical $\Spin^c(2)$ structure associated to $g$ and the complex structure $J_g$. For each Riemannian metric $g$, we define $S$ to be the tautological complex spinor bundle associated to $Q_{g}$. The spinor bundle $S$ admits an explicit model given by:
\begin{equation*}
S = \Lambda^{0, \ast}(X)\, ,
\end{equation*}

\noindent
where the splitting is performed with respect to the complex structure $J_g$. Clifford multiplication is given by:
\begin{equation}
\label{eq:Cliffordmanti}
\beta\cdot \alpha = 2 \beta^{0,1}\wedge \alpha + \iota_{(\beta^{\sharp})^{0,1}} \alpha\, ,
\end{equation}

\noindent
for all $\alpha\in \Omega^{0,\ast}(X)$ and all $\beta\in \Omega^{1}(X)$. The determinant line bundle associated to the canonical $\Spin^c(2)$ structure $Q_g$ is given by the complex-conjugate canonical bundle $K^{\ast}_g$ of $(X,g)$:
\begin{equation*}
\L_{Q_g} = K^{\ast}_g = \Lambda^{0,1}(X)\, ,
\end{equation*}

\noindent
which is complex-isomorphic to the anticanonical bundle $T^{1,0}X$ of $(X,g)$. The complex spinor bundle $S$ is thus a complex vector bundle of rank two, which splits in the usual way:
\begin{equation*}
S = S^{+} \oplus S^{-}\, ,
\end{equation*}

\noindent 
in terms of the chiral bundles $S^{+}$ and $S^{-}$. In the polyform presentation $S = \Lambda^{0, \ast}(X)$ of the spinor bundle, the chiral spinor bundles $S^{\pm}$ respectively correspond with:
\begin{equation*}
S^{+} \simeq \Lambda^{0,odd}(X) \simeq \Lambda^{0,1}(X)\, , \qquad S^{-} \simeq \Lambda^{0,even}(X) \simeq \Lambda^{0,0}(X)\, , 
\end{equation*}

\noindent
whereas the chiral spinor bundles $S^{\pm}_{c}$ correspond with:
\begin{equation*}
S^{+}_{c} \simeq \Lambda^{0,0}(X)\, , \qquad S^{-}_{c} \simeq  \Lambda^{1,0}(X)\, . 
\end{equation*}

\noindent
As required, we have:
\begin{equation*}
S^{+} \simeq S^{+}_{c}\otimes K^{\ast}_g\, , \qquad S^{-} \simeq S^{-}_{c}\otimes K^{\ast}_g\, .
\end{equation*}

\noindent
The notion of chiral triple is modified accordingly.

\begin{definition}
A chiral triple $(\cL, \cH, \cW)$ on $(X,\cM)$ consists on a negative Hermitian holomorphic line bundle $(\cL,\cH)$ and a holomorphic section $\cW\in H^{0}(\cM,\cL)$ such that there exists a map $\varphi\colon X\to \cM$ and a metric $g$ on $X$ for which:
\begin{equation*}
K^{\ast}_g \simeq \cL^{\varphi}\, ,
\end{equation*}

\noindent
as complex line bundles.
\end{definition}

\

\noindent
We fix a chiral triple $\mathfrak{Q}$ on $(X,\cM)$ with vanishing superpotential. Instead of considering the Killing spinor equations required by the supersymmetric structure of chiral $\cN=1$ supergravity, we consider the following equations:

\begin{equation}
\label{eq:KSEXW0anti}
\nabla^{\varphi}\epsilon = 0 \, , \qquad \dd \varphi^{1,0}\cdot \epsilon = 0\, ,
\end{equation}

\noindent
for $\epsilon \in \Gamma(S^{-})$. Note that the honest Killing spinor equations would require:
\begin{equation*}
\dd \varphi^{0,1}\cdot \epsilon = 0\, ,
\end{equation*}

\noindent
instead of the second equation appearing in \eqref{eq:KSEXW0anti}. 
\begin{definition}
We call triples $(g,\varphi,\Psi)$ satisfying Equations \eqref{eq:KSEXW0anti} \emph{anti-supersymmetric solutions}.
\end{definition}

\noindent
The key point is that anti-supersymmetric solutions are not supersymmetric solutions yet they are honest solutions of chiral $\cN=1$ supergravity. The proof of the following theorem is completely analogous to the proof of Theorem \ref{thm:solsusyX}.

\begin{thm}
	\label{thm:solsusyXanti}
	Let $\mathfrak{Q}$ be a chiral triple on $(X,\cM)$ such that $\cW = 0$. A triple $(g,\varphi,\Psi)$ with non-constant $\varphi$ is an anti-supersymmetric solution of the chiral supergravity associated to $(\cM,\mathfrak{Q})$ if and only if the following conditions hold:
	
	\
	
	\begin{enumerate}
		\item The smooth map $\varphi\colon (X,g)\to (\cM,\cG)$ is a holomorphic map with respect to $J_g$ and the fixed complex structure $\cI$ on $\cM$. 
		
		\
		
		\item $\Psi\colon K^{\ast}_g \xrightarrow{\simeq} \cL^{\varphi}$ is an isomorphism of holomorphic line bundles such that: 
		\begin{equation*}
		g_c  = \kappa\, \Psi^{\ast} \cH^{\varphi}\, , 
		\end{equation*}
		
		\noindent
		for a constant $\kappa\in \mathbb{R}_{>0}$, where $g_c$ denotes the Hermitian metric induced by $g$ on $T^{1,0}X$. 
	\end{enumerate}
	
	\noindent
	These conditions imply that:
	\begin{equation*}
	\mathrm{R}_{g} = - \norm{\dd\varphi}^2_{\cG,g}\, , 
	\end{equation*}
	
	\noindent
	and that the K\"ahler metric $\cG$, the Riemannian metric $g$ and the map $\varphi$ satisfy:
	\begin{equation*}
	\varphi^{\ast}\cG = \frac{\norm{\dd\varphi}^2_{\cG,g}}{2}\, g = - \frac{\mathrm{R}_{g}}{2}\, g\, , 
	\end{equation*}
	
	\noindent
	Hence, $\varphi$ is a conformal immersion of $X\backslash C$ into $\cM$, where $C\subset X$ denotes the critical set of $\varphi$. Furthermore, if $(X,g)$ is compact, then it is hyperbolic.
\end{thm}

\noindent
We can adapt the notion of chiral map to accommodate anti-supersymmetric solutions.

\begin{definition}
Let $X$ be an oriented real two-manifold and let $(\cM,\cL,\cH)$ be a complex manifold equipped with a negative Hermitian holomorphic line bundle $(\cL,\cH)$. We say that a pair $(\varphi,\Psi)$ is a \emph{anti-chiral map} with respect to $(\cM,\cL,\cH)$ if there exists a complex structure $J=J_{\varphi}$ on $X$ such that:
\begin{equation*}
\varphi\colon (X,J)\to (\cM,\cI)\, ,
\end{equation*}
	
\noindent
is holomorphic and: 
\begin{equation*}
\Psi\colon K^{\ast}\xrightarrow{\simeq} \cL^{\varphi}\, ,
\end{equation*}
	
\noindent
is an isomorphism of holomorphic line bundles, where $K^{\ast}$ is the anti-canonical bundle of $(X,J)$.
\end{definition}

\noindent
As it happened with chiral maps, aside from the role they play in chiral supergravity, anti-chiral maps are interesting because they are particular instances of holomorphic maps of Riemann surfaces into K\"ahler manifolds and provide solutions to the \emph{coupled} problem of prescribing the scalar curvature of a Riemann surface to:
\begin{equation*}
\mathrm{R}_{g} = - \norm{\dd\varphi}^2_{\cG,g}\, .
\end{equation*}

\noindent
When $\cM$ is also an oriented two-dimensional real manifold,  Theorem \ref{thm:solsusyX} implies the following corollary.

\begin{cor}
Let $\cM$ be connected and complex one-dimensional and let $\varphi\colon X\to \cM$ be a anti-chiral map with respect to $(\cL,\cH)$. If $\varphi$ is proper then it is a $d$-fold holomorphic branched covering of Riemann surfaces. If in addition $X$ is compact then $\cM$ is necessarily compact and we have:
\begin{equation*}
\mathrm{deg}(\cL) = \chi (\cM) - \frac{k}{d}\, ,
\end{equation*}	
	
\noindent
where $k$ is the total branching number of $\varphi$. In particular:
\begin{equation*}
\vert\mathrm{deg}(\cL)\vert \geq \vert\chi(\cM)\vert\, ,
\end{equation*}
	
\noindent
and if $\cL \simeq T^{1,0}\cM$ then $k=0$ and $\varphi$ is a holomorphic unbrached covering of compact Riemann surfaces.
\end{cor}

\begin{proof}
By definition of anti-chiral map, $\varphi$ is holomorphic. Since by assumption $\cM$ is connected and $\varphi$ is  proper, $\varphi$ must be surjective and thus it is a holomorphic branched covering, implying also that if $X$ is compact then $\cM$ is also compact. Assume now that $X$ (and thus also $\cM$) is compact. The Riemann-Hurwitz formula implies:
\begin{equation*}
\chi(X) = d\, \chi(\cM) - k\, ,
\end{equation*}
	
\noindent
Using now that $\chi(X) = \mathrm{deg}(T^{1,0}X)  = \mathrm{deg}(\cL^{\varphi}) = d \,\mathrm{deg}(\cL)$ we obtain the first formula of the corollary. Now, the definition of chiral triple requires $\cL$ to be negative and thus both $\deg(\cL)$ and $\chi(X)$ are negative. Hence:
\begin{equation*}
-\mathrm{deg}(\cL) \geq\, - \chi(\cM)\, ,
\end{equation*}
	
\noindent
and we conclude.
\end{proof}

\noindent
The four-dimensional Lorentzian manifold $(M_4,g_4)$ associated to the anti-supersymmetric solutions considered in Theorem \ref{thm:solsusyXanti} is of the form:
\begin{equation*}
(M_4,g_4) = (\mathbb{R}^2\times X, \eta_{1,1}\times g)\, ,
\end{equation*}

\noindent
where $(X,g)$ is a Riemann surface with the Riemannian metric $g$ being part of a supersymmetric solution $(g,\varphi,\Psi)$. As  mentioned in Theorem \ref{thm:solsusyXanti}, if $X$ is compact then it is biholomorphic with a hyperbolic Riemann surface. Reasoning by analogy with the situation for supersymmetric solutions, we wonder if solutions of this type, with $X$ compact hyperbolic, can appear as non-supersymmetric near horizon geometries of black hole solutions in four dimensions.

% % % % % % % % % % % % % % % % % % % % % % % % % % % % % % % % % % % % % % 
% % % % % % % % % % % % % % % % % % % % % % % % % % % % % % % % % % % % % %

\subsection{Examples of chiral and anti-chiral maps}
\label{sec:examples}

% % % % % % % % % % % % % % % % % % % % % % % % % % % % % % % % % % % % % % 
% % % % % % % % % % % % % % % % % % % % % % % % % % % % % % % % % % % % % %

In this section we construct several examples of (anti) chiral maps and (anti) supersymmetric solutions.

\begin{ep}
Take $X = \mathbb{P}^1$, $\cM = \mathbb{P}^1$ and $\cL = K_{\mathbb{P}^1}$. Denote by $g$ the round metric on $\mathbb{P}^1$ and take $\cH$ to be the Hermitian metric induced by $g$ on $K_{\mathbb{P}^1}$. Then, any triple $(g,\varphi,\Psi)$, with $\varphi\colon (\mathbb{P}^1,g) \to (\mathbb{P}^1,g)$ holomorphic isometry and $\Psi$ induced by $\varphi$ gives a supersymmetric solution to chiral $\cN=1$ supergravity. 
\end{ep}

\begin{ep}
Let $(X,g)$ be a hyperbolic Riemann surface with $g$ of constant negative curvature $-1$ and admitting non-trivial holomorphic isometries (consider for example $X$ hyperelliptic curve and consider its hyperelliptic involution). Take $\cM = X$ and define $\cL = T^{1,0}X$ to be the holomorphic tangent bundle of $X$. Furthermore, we take $\cH = g_c$, where $g_c$ denotes the Hermitian structure induced by $g$ on $T^{1,0}X$. With this choice of Hermitian structure $\cH$, the K\"ahler metric associated to the Chern curvature of $\cH$ in the sense of \ref{def:chiraltriple} is again $g$. Hence $\cG = g$. We claim that every isometry:
\begin{equation*}
\varphi \colon (X,g)\to (X,g)\, ,
\end{equation*}
	
\noindent
gives rise to an anti-supersymmetric solution $(g,\varphi,\Psi)$ with respect to the chiral triple $(T^{1,0}X,g_c, \cW = 0)$ on the pair $(X,\cM= X)$. To see this, note that we can define $\Psi = \dd \varphi$ since:
\begin{equation*}
\dd\varphi\colon T^{1,0}X \xrightarrow{\simeq} (T^{1,0}X)^{\varphi} \, ,
\end{equation*}
	
\noindent
is an isomorphism of holomorphic line bundles. Hence, the conditions of Theorem \ref{thm:solsusyX} are satisfied and such $(g,\varphi,\Psi)$ is a solution of chiral $\cN=1$ supergravity on $X$ associated to the chiral triple $(T^{1,0}X,g_c,0)$. Since $\varphi$ is assumed to be an isometry, direct computation shows that:
\begin{equation*}
\norm{\dd \varphi}^2_{g,\cG} = 2\, ,
\end{equation*}
	
\noindent
and thus:
\begin{equation*}
\varphi^{\ast} \cG = g = \frac{g}{2} \norm{\dd \varphi}^2_{g,\cG}\, , 
\end{equation*}
	
\noindent
as claimed in Theorem \ref{thm:solsusyXanti}.
\end{ep}

\begin{ep}
Let $X$ be a compact hyperbolic Riemann surface not of hyperelliptic type. Then, the canonical bundle $K_X$ of $X$ is very ample \cite{Hartshorne} and Kodaira's embedding theorem implies that for an appropriate $n>1$ there exists a holomorphic embedding of $X$ into $n$-dimensional projective space:
\begin{equation*}
\varphi\colon X \hookrightarrow \mathbb{P}^n\, ,
\end{equation*}

\noindent
satisfying:
\begin{equation*}
K_{X} \simeq \varphi^{\ast}\cO(1)\, ,
\end{equation*}

\noindent
where $\cO(1)$ denotes the tautological bundle of $\mathbb{P}^n$. The previous equation implies that there exists an isomorphism $\Psi$ of holomorphic line bundles:
\begin{equation*}
\Psi\colon T^{1,0}X \xrightarrow{\simeq} \varphi^{\ast}\cO(-1)\, .
\end{equation*}

\noindent
Furthermore, the Hermitian structure $\cH$ on $\cO(-1)$ induced by the Fubini-Study metric on $\mathbb{P}^n$ makes $(\cO(-1), \cH)$ into a negative line bundle. Therefore, $(\cO(-1), \cH, 0)$ is a chiral triple on the pair $(X,\mathbb{P}^n)$ and $(g,\varphi,\Psi)$ is a supersymmetric solution, where $g$ is constructed in terms of $\cH$, $\varphi$ and $\Psi$ as prescribed by Theorem \ref{thm:solsusyXanti}.
\end{ep}

\noindent
We present now a large family of supersymmetric solutions $(g,\varphi, \Psi)$ to chiral $\cN=1$ supergravity associated to a special class of scalar manifolds $(\cM,\mathfrak{Q})$ admitting plurisubharmonic functions and having vanishing superpotential.

Let $X\subset\mathbb{C}$ be a complex domain in $\mathbb{C}$ with complex coordinate $w$ and let $\cM$ be an $n$-dimensional complex manifold admitting smooth strictly plurisubharmonic functions. For example, we can take $\cM = \mathbb{C}^n$, we can take $\cM$ to be any open Riemann surface, or more generally we can take $\cM$ to be a Stein complex $n$-manifold. Fix a smooth plurisubharmonic function $\phi \colon\cM\to \mathbb{R}$ on $\cM$. Let $\cL$ be the holomorphically trivial complex line bundle over $\cM$, and let us fix a holomorphic trivialization $\cL= \cM\times \mathbb{C}$. In this trivialization, we define a Hermitian structure $\cH$ as follows:
\begin{equation*}
\cH(f_1,f_2) \eqdef e^{\phi} f_1 \bar{f}_2\, , 
\end{equation*}

\noindent
where $f_1, f_2 \colon \cM \to \mathbb{C}$ are smooth functions. Since the cotangent bundle of $X$ is holomorphically trivial, the triple $\mathfrak{Q} \eqdef (\cL,\cH,0)$ is a chiral triple on $(X,\cM)$ for any holomorphic map:
\begin{equation*}
\varphi \colon X\to \cM\, .
\end{equation*}

\noindent
In particular, $\cL^{\varphi}$ is holomorphically isomorphic to $\Lambda^{1,0}(X)$. Let us trivialize $\cL^{\varphi}$ as $\cL^{\varphi} = X\times \mathbb{C}$ by using the pull-back of the fixed trivialization of $\cL$. The holomorphic cotangent bundle of $X$ is the complex span of $\left\{\dd w\right\}$. Using the previous trivializations and the complex coordinate $w$, we define an isomorphism of holomorphic line bundles $\Psi\colon \Lambda^{1,0}(X)\xrightarrow{\simeq} \cL^{\varphi}$ as follows:
\begin{equation*}
\Psi(\dd w) = 1
\end{equation*}

\noindent
With these provisos in mind, Theorem \ref{thm:solsusyX} implies the following result.

\begin{cor}
\label{cor:susyfamily}
In the set-up introduced above, any triple $(g,\varphi,\Psi)$ with $\varphi\colon X\to \cM$ holomorphic and $g$ given through its associated Hermitian metric on $\Lambda^{1,0}(X)$ as follows:
\begin{equation*}
g^{\ast}_c = \Psi^{\ast} \cH^{\varphi}\, ,
\end{equation*}

\noindent
is a supersymmetric solution of the chiral $\cN=1$ supergravity associated to $\mathfrak{Q}$.
\end{cor}

\begin{remark}
This corollary can be easily adapted to yield anti-supersymmetric solutions instead of supersymmetric solutions. In order to do this, simply consider the  holomorphic tangent bundle of $X$ instead of its cotangent bundle. We leave the details to the reader. 
\end{remark}

\noindent
Let us explore in more detail the solution provided by the previous corollary. Let $(g,\varphi,\Psi)$ be such a solution. The Hermitian structure $\cH^{\varphi}$ on $\cL^{\varphi}$ reads:
\begin{equation*}
\cH^{\varphi}(f_1,f_2) = e^{\phi (\varphi)} f_1 \bar{f}_2\, , 
\end{equation*}

\noindent
Using the explicit form of $\Psi$, the two-dimensional metric $g$ associated to $\varphi$ as described in Theorem \ref{thm:solsusyX} reads (we take $\kappa =2$ for simplicity):
\begin{equation*}
g = (\Psi^{\ast}\cH^{\varphi})^{\ast} = e^{-\phi(\varphi)(w)} \dd w \odot \dd\bar{w}\, .
\end{equation*}

\noindent
Explicit computation gives the following formula for the scalar curvature:
\begin{equation*}
\mathrm{R}_g = \Delta_g (\phi (\varphi))\, , 
\end{equation*}

\noindent
showing that it is non-negative, as required by Theorem \ref{thm:solsusyX}, since $\phi$ is plurisubharmonic. Direct computation shows that:
\begin{equation*}
\norm{\dd \varphi}^2_{g,\cG} = \Delta_g (\phi (\varphi))\, ,
\end{equation*}

\noindent
and thus, as required in order to have a supersymmetric solution, we have:
\begin{equation*}
\mathrm{R}_g = \norm{\dd \varphi}^2_{g,\cG}\, .
\end{equation*}

\noindent
Furthermore:
\begin{equation*}
\varphi^{\ast} \cG = \frac{ e^{-\phi(\varphi)}}{2} \Delta_g (\phi (\varphi))\, \dd w\odot \dd\bar{w} = \frac{\norm{\dd \varphi}^2_{g,\cG}}{2}\, g\, ,
\end{equation*}

\noindent
whence, as required by Theorem \ref{thm:solsusyX}, the Einstein equation for $(g,\varphi)$ is satisfied.

\begin{remark}
Corollary \ref{cor:susyfamily} provides us with an \emph{infinite family} of supersymmetry solutions to chiral $\cN=1$ supergravity associated to the type of chiral triple introduced above. Remembering that the supergravity theory considered on $X$ is a reduction of the Lorentzian theory, we can easily reconstruct the associated family of Lorentzian solutions. We have $M = \mathbb{R}^2\times X$ and:  
\begin{equation*}
g = -\dd t\odot \dd t + \dd x\odot\dd x + e^{-\phi(\varphi)} \dd w\odot \dd\bar{w}\, , 
\end{equation*}

\noindent
which gives, for fixed $X$ and $\cM$, a family of Lorentzian metrics on $M$ depending on the choice of plurisubharmonic function $\phi$ on $\cM$ and holomorphic map $\varphi\colon X \to \cM$.
\end{remark}

\noindent
We present now two explicit examples of the previous construction.
 
\begin{ep}
Let us set $X=\mathbb{C}$ and $\cM = \mathbb{C}^{\ast}$ with its standard K\"ahler form. Let $\cL$ be the holomorphically trivial complex line bundle over $\cM$, and let us fix a trivialization $\cL= \mathbb{C}^{\ast}\times \mathbb{C}$. In this trivialization, we define $\cH$ as follows:
\begin{equation*}
\cH(f_1,f_2) \eqdef e^{\vert z\vert^2} f_1 \bar{f}_2\, , 
\end{equation*}

\noindent
where $z$ is a fixed complex coordinate of $\mathbb{C}^{\ast}$ and $f_1, f_2 \colon \mathbb{C}^{\ast} \to \mathbb{C}$ are smooth functions. The curvature of the Chern connection associated to $\cH$ reads:
\begin{equation*}
\Theta = - \dd z \wedge \dd\bar{z}\, ,
\end{equation*}

\noindent
whence it is negative and its associated K\"ahler form $\cV$ and K\"ahler metric $\cG$ are given by:
\begin{equation*}
\cV  = \frac{i}{2\pi} \dd z \wedge \dd\bar{z}\, , \qquad \cG = \frac{1}{2\pi} \dd z \odot \dd\bar{z}
\end{equation*}

\noindent
We define:
\begin{equation*}
\varphi \eqdef \mathbb{C} \to \mathbb{C}^{\ast}\, , \qquad w\mapsto e^{w}\, ,
\end{equation*} 

\noindent
as a potential candidate for chiral map and in particular a solution of chiral $\cN=1$ supergravity on $X$. Clearly, $\dd\varphi$ is an isomorphism of vector bundles. The complex line bundle $\cL^{\varphi}$ is holomorphically trivial over $X$ and we use the pull-back trivialization of $\cL$ to set $\cL^{\varphi} = X\times \mathbb{C}$. Let now $w$ denote a global complex coordinate on $X = \mathbb{C}$. The holomorphic tangent bundle of $\mathbb{C}$ is the complex span of $\left\{\dd w\right\}$. We define an isomorphism of holomorphic line bundles $\Psi\colon \Lambda^{1,0}(X)\xrightarrow{\simeq} \cL^{\varphi}$ as follows:
\begin{equation*}
\Psi(\dd w) = 1
\end{equation*}

\noindent
in the chosen trivializations. With these provisos in mind, $(\cL,\cH,0)$ is a chiral triple on $(\mathbb{C},\mathbb{C}^{\ast})$. The Hermitian structure $\cH^{\varphi}$ on $\cL^{\varphi}$ reads:
\begin{equation*}
\cH^{\varphi}(f_1,f_2) = e^{e^{w + \bar{w}}} f_1 \bar{f}_2\, , 
\end{equation*}

\noindent
Using the explicit form of $\Psi$, the two-dimensional metric $g$ associated to $\varphi$ as described in Theorem \ref{thm:solsusyX} reads:
\begin{equation*}
g = (\Psi^{\ast}\cH^{\varphi})^{\ast} = e^{-e^{w + \bar{w}}} dw \odot d\bar{w}\, .
\end{equation*}

\noindent
Since $\varphi$ is a holomorphic immersion, in fact it is a holomorphic unramified covering, and $g$ is constructed as required by Theorem \ref{thm:solsusyX}, we conclude that $(g,\varphi,\Psi)$ is a supersymmetric solution with respect to the pair $(\mathbb{C},\mathbb{C}^{\ast})$ and the chiral triple specified above. Explicit computation gives the following formula for the scalar curvature:
\begin{equation*}
\mathrm{R}_g = 4\,\frac{e^{\omega + \bar{w}}}{e^{-e^{w + \bar{w}}}} \, , 
\end{equation*}

\noindent
showing that it is negative definite and bounded. Direct computation shows that:
\begin{equation*}
\norm{\dd \varphi}^2_{g,\cG} = 4\, \frac{e^{w + \bar{w}}}{e^{-e^{\omega + \bar{w}}}}\, ,
\end{equation*}

\noindent
and thus, as required for a supersymmetric solution, we have:
\begin{equation*}
\mathrm{R}_g = \norm{\dd \varphi}^2_{g,\cG}\, .
\end{equation*}

\noindent
Furthermore:
\begin{equation*}
\varphi^{\ast} \cG = 2\,e^{w+\bar{w}} \dd w\odot \dd\bar{w} = \frac{g}{2}  \norm{\dd \varphi}^2_{g,\cG}\, ,
\end{equation*}

\noindent
and thus, as expected, the Einstein equation for $(g,\varphi)$ is satisfied.
\end{ep}

\begin{ep}
Let us set $X=\mathbb{C}$ and $\cM = \mathbb{C}$ with its standard K\"ahler form. Let $\cL$ be the holomorphically trivial complex line bundle over $\cM$, and let us fix a trivialization $\cL= \mathbb{C} \times \mathbb{C}$. In this trivialization, we define $\cH$ again as follows:
\begin{equation*}
\cH(f_1,f_2) \eqdef e^{\vert z\vert^2} f_1 \bar{f}_2\, , 
\end{equation*}

\noindent
where $z$ is the complex coordinate of $\cM = \mathbb{C}$ and $f_1, f_2 \colon \mathbb{C}\to \mathbb{C}^{\ast}$ are smooth functions. Then, the curvature of the Chern connection associated to $\cH$ reads:
\begin{equation*}
\Theta = - \dd z \wedge \dd\bar{z}\, ,
\end{equation*}

\noindent
whence the associated K\"ahler form $\cV$ and K\"ahler metric $\cG$ read:
\begin{equation*}
\cV  = \frac{i}{2\pi} \dd z \wedge \dd\bar{z}\, , \qquad \cG = \frac{1}{2\pi} \dd z \odot \dd\bar{z}
\end{equation*}

\noindent
We define:
\begin{equation*}
\varphi_k \eqdef \mathbb{C} \to \mathbb{C}\, , \qquad w\mapsto w^k\, , \qquad k\geq 1\, ,
\end{equation*} 

\noindent
as a potential candidate for chiral map and in particular a solution of chiral $\cN=1$ supergravity on $X$, where $w$ denote a global complex coordinate on $X = \mathbb{C}$. The complex line bundle $\cL^{\varphi}$ is holomorphically trivial over $X$ and we use the pull-back trivialization of $\cL$ to set $\cL^{\varphi} = X\times \mathbb{C}$. The holomorphic cotangent bundle of $\mathbb{C}$ is the complex span of $\left\{\dd w\right\}$. We define an isomorphism of holomorphic line bundles $\Psi\colon \Lambda^{1,0}(X) \xrightarrow{\simeq} \cL^{\varphi}$ as follows:
\begin{equation*}
\Psi(\dd w) = 1
\end{equation*}

\noindent
in the chosen trivializations. With these provisos in mind, $(\cL,\cH,0)$ is a chiral pair on the pair $(\mathbb{C},\mathbb{C})$ for any holomorphic map $\varphi\colon \mathbb{C}\to \mathbb{C}$. The Hermitian structure $\cH^{\varphi}$ on $\cL^{\varphi}$ reads (we take $\kappa=1$ for simplicity):
\begin{equation*}
\cH^{\varphi}(f_1,f_2) = e^{(w \bar{w})^k} f_1 \bar{f}_2\, , 
\end{equation*}

\noindent
Using the explicit form of $\Psi$, the two-dimensional metric $g$ associated to $\varphi$ as described in Theorem \ref{thm:solsusyX} reads:
\begin{equation*}
g = (\Psi^{\ast}\cH^{\varphi})^{\ast} = e^{-(w \bar{w})^k} dw \odot d\bar{w}\, .
\end{equation*}

\noindent
Since $\varphi$ is a holomorphic map and $g$ is constructed as required by Theorem \ref{thm:solsusyX}, we conclude that $(g,\varphi,\Psi)$ is a supersymmetric solution with respect to the pair $(\mathbb{C},\mathbb{C})$ and the chiral triple specified above. Explicit computation gives the following formula for the scalar curvature:
\begin{equation*}
\mathrm{R}_g = 4\,k^2 (w\bar{w})^{k-1} e^{(w \bar{w})^k}\, , 
\end{equation*}

\noindent
showing that it is positive semi-definite. Direct computation shows that:
\begin{equation*}
\norm{\dd \varphi}^2_{g,\cG} = 4\, k^2 (w\bar{w})^{k-1} e^{(w \bar{w})^k}\, ,
\end{equation*}

\noindent
and thus as required for a supersymmetric solution we have:
\begin{equation*}
\mathrm{R}_g = \norm{\dd \varphi}^2_{g,\cG}\, .
\end{equation*}

\noindent
Furthermore:
\begin{equation*}
\varphi^{\ast} \cG = 2\, k^2 (w\bar{w})^{k-1} \dd w\odot \dd\bar{w} = \frac{g}{2}  \norm{\dd \varphi}^2_{g,\cG}\, ,
\end{equation*}

\noindent
and thus, as expected, the Einstein equation for $(g,\varphi)$ is satisfied. The only critical point of $\varphi$ is $0\in\mathbb{C}$, which corresponds to the only point in $\mathbb{C}$ at which the Gaussian curvature of $g$ vanishes, as required by Theorem \ref{thm:solsusyX}. Note that, crucially, although the symmetric bilinear form $\varphi^{\ast}\cG$ is degenerate at $0\in \mathbb{C}$, the \emph{physical} metric $g$ constructed as prescribed by Theorem \ref{thm:solsusyX} is regular at $0\in \mathbb{C}$. Recall that $(g,\varphi,\Psi)$ cannot be smoothly extended to the one-point compactification of $\mathbb{C}$, given by the Riemann sphere $\mathbb{P}^1$. Furthermore, this example is not covered by Theorem \ref{thm:finitecurvature} since the total curvature of $g$ is not finite.
\end{ep}

% % % % % % % % % % % % % % % % % % % % % % % % % % % % % % % % % % % % % % 
% % % % % % % % % % % % % % % % % % % % % % % % % % % % % % % % % % % % % %

\section{Conclusions and open problems}
\label{sec:conclusions}

% % % % % % % % % % % % % % % % % % % % % % % % % % % % % % % % % % % % % % 
% % % % % % % % % % % % % % % % % % % % % % % % % % % % % % % % % % % % % %

In this article we have constructed a global geometric model for the bosonic sector and Killing spinor equations of four-dimensional $\cN=1$ supergravity coupled to a $\Spin^c_0(3,1)$ structure and to a non-linear sigma model, whose target space is given by a complex manifold admitting a novel geometric structure which we call chiral triple. We have dimensionally reduced the theory to a Riemann surface $X$ and we have characterized all supersymmetric solutions on a Riemann surface $X$, classifying the possible biholomorphism types of $X$ for supersymmetric solutions with complete Riemannian metric and finite scalar energy. Furthermore, we have introduced the notion of anti-supersymmetric solution as a solution of a natural variation of the Killing spinor equations and we have characterized all anti-supersymmetric solutions on a Riemann surface. More generally, we have obtained a system of partial differential equations for a harmonic map with potential mapping a Riemann surface into a K\"ahler manifold which appears to be novel in the literature and arises here as a consequence of $\cN=1$ chiral supersymmetry.

The geometric model we have constructed involves a number of new mathematical structures which have not been explored in the mathematical literature. As a consequence, the present work opens up several possible novel lines of research. We summarize some of them:

\begin{itemize}
	
	\item The geometric model presented in this article is based on the notion of chiral triple. Associated to this concept there are two natural problems. The first consists on characterizing which pairs $(X,\cM)$ consisting of a Riemann surface $X$ and a complex manifold $\cM$ admit chiral triples and thus give rise to admissible non-linear sigma models for $\cN=1$ chiral supergravity. The second problem consists on classifying the space of chiral triples on a given pair $(X,\cM)$, studying if these structures come in families and give rise to finite-dimensional moduli spaces.

\

	\item We have characterized all supersymmetric solutions on Riemann surfaces $X$ for which the superpotential vanishes. However, the problem of classifying which Riemann surfaces or Lorentzian four-manifolds admit supersymmetric solutions with non-vanishing superpotential (or even vanishing superpotential in the latter case) remains completely open. We expect the solution to this problem to heavily depend on the choice of scalar manifold and superpotential, whence a general solution will be probably out of reach. It is then reasonable to consider first the classification problem in the four-dimensional case for which the superpotential $\cW$ vanishes, as explained in Section \ref{sec:TrivialSuperpotential}, which is still open and where we can expect a complete classification result.

\
	
	\item  A very interesting class of supergravity solutions consist on globally hyperbolic Lorentzian four-manifolds. The main result of \cite{Bernal:2003jb} characterizes the most general form of a globally hyperbolic Lorentzian manifold. Using this explicit presentation for the Lorentzian manifold, it would be very interesting to consider $\cN=1$ chiral supergravity on such manifolds and reduce it to the corresponding Riemannian three-manifold, obtaining a set of \emph{supersymmetric} flow Killing spinor equations on a three-manifold and studying its properties.

\
	
	\item In some specific situations, such as in the Riemannian product of Minkowski two-dimensional space and a Riemann surface, the Killing spinor equations of the theory give rise to well-defined moduli spaces of solutions. In these cases, one is likely to obtain novel moduli problems, involving maps to a K\"ahler manifold, which have not been studied in the literature and may help understanding the space of solutions to supergravity theories and may have further applications in differential topology. For the case considered in Section \ref{sec:susysolX}, we obtain a moduli problem which consists on a variation of the moduli problem of holomorphic maps from a Riemann surface into a complex manifold.

\

	\item We have considered a geometric model in which, using the language introduced in \cite{LazaroiuBC}, the Lorentzian submersion $\pi$ describing the scalar sector of the theory is metrically trivial. It remains open to study the problem of constructing a geometric model for the Killing spinor equations in the case in which $\pi$ is a non-trivial flat Lorentzian submersion.

\

	\item We have considered a geometric model for chiral $\cN=1$ supergravity not coupled to gauge fields. It would be very interesting to extend the construction as to include gauge fields by constructing the appropriate duality bundle over the scalar manifold of the theory.

\

	\item A particularly interesting open problem in the context of the program to construct the mathematical foundations of geometric supergravity, is to extend the present geometric model to $\cN=2$ ungauged supergravity coupled to vector multiplets, understanding the notion of chiral triple and symplectic duality bundle appropriate for this case, which needs to  encode the appropriate notion of projective Special K\"ahler geometry \cite{Strominger:1990pd,Freed:1997dp,Alekseevsky:1999ts} required for the construction of the theory.
	
\end{itemize}

% % % % % % % % % % % % % % % % % % % % % % % % % % % % % % % % % % % % % % % % 
% % % % % % % % % % % % % % % % % % % % % % % % % % % % % % % % % % % % % % % %

\newpage

% % % % % % % % % % % % % % % % % % % % % % % % % % % % % % % % % % % % % % % % 
% % % % % % % % % % % % % % % % % % % % % % % % % % % % % % % % % % % % % % % %

\appendix

% % % % % % % % % % % % % % % % % % % % % % % % % % % % % % % % % % % % % % % % 
% % % % % % % % % % % % % % % % % % % % % % % % % % % % % % % % % % % % % % % %

% % % % % % % % % % % % % % % % % % % % % % % % % % % % % % % % % % % % % % 
% % % % % % % % % % % % % % % % % % % % % % % % % % % % % % % % % % % % % %

\end{document}